\documentclass[12pt]{article}
\usepackage[utf8]{inputenc}
\usepackage[unicode, colorlinks=true, linkcolor=blue, filecolor=cyan, citecolor=blue, urlcolor=magenta, bookmarks=true, bookmarksopen=false, bookmarksnumbered=true]{hyperref}
\usepackage[numbered]{bookmark}
\makeatletter
\newcommand*{\sectionbookmark}[1][]{
  \bookmark[
    level=section,
    dest=\@currentHref,%
    #1%
  ]%
}
\makeatother
\usepackage{lmodern}
\usepackage{placeins}
\usepackage{needspace}
\usepackage{setspace}
\usepackage{float}
\usepackage{pdflscape}
\usepackage{longtable}
\usepackage[nopatch=footnote]{microtype}
\usepackage[T1]{fontenc}
\usepackage{amsthm}
\usepackage[letterpaper,margin=1in]{geometry}
\usepackage[american]{babel}
\usepackage{apacite}
\usepackage[noblocks]{authblk}
\usepackage{etoolbox}
\AtBeginEnvironment{thebibliography}{\normalsize}
\usepackage{booktabs}
\usepackage{multirow}
\usepackage{threeparttable}
\usepackage{siunitx}
\usepackage{caption}
\usepackage{tcolorbox}
\usepackage{adjustbox}
\usepackage{amsmath}
\usepackage{enumitem}
\usepackage{amssymb}
\usepackage{bm}
\usepackage{amsfonts}
\usepackage{extarrows}
\usepackage{mathtools}
\usepackage{upgreek}
\usepackage{xtab}
\usepackage{makecell}
\usepackage{tabu}
\usepackage{pifont}
\usepackage{tocloft}
\usepackage{titlesec}
\usepackage{fancyhdr}
\usepackage{indentfirst}
\usepackage{accents}
\usepackage{subcaption}
\usepackage{newtxtext, newtxmath}
\usepackage[ruled]{algorithm2e}
\usepackage{pgfplots}
\pgfplotsset{compat=1.18}
\usetikzlibrary{decorations.markings, arrows.meta, backgrounds}
\usepackage{epstopdf}
\allowdisplaybreaks[4]

\newtheorem{theorem}{Theorem}
\newtheorem{assumption}{Assumption}

\newtheorem{techlemma}{Lemma}

\newtheorem{proposition}{Proposition}
\newtheorem{corollary}{Corollary}
\newtheorem{remark}{Remark}

\title{Bayesian Smoothed Quantile Regression}

\author{%
  Bingqi Liu\thanks{Corresponding author, E-mail: \href{mailto:bqliu@zju.edu.cn}{bqliu@zju.edu.cn} (Bingqi Liu), ORCID: \href{https://orcid.org/0000-0003-0948-8930}{0000-0003-0948-8930}.}
  \hspace{1.6em}
  Kangqiang Li\thanks{E-mail: \href{mailto:11935023@zju.edu.cn}{11935023@zju.edu.cn} (Kangqiang Li), ORCID: \href{https://orcid.org/0000-0002-4253-6730}{0000-0002-4253-6730}.}
  \hspace{1.3em}
  Tianxiao Pang\thanks{E-mail: \href{mailto:txpang@zju.edu.cn}{txpang@zju.edu.cn} (Tianxiao Pang).}
  
  \vspace*{1.2em}
  
  \small \textit{School of Mathematical Sciences, Zhejiang University, Hangzhou 310058, China}
}

\date{}

\begin{document}
\maketitle
\doublespacing
\vspace{-1cm}
\begin{abstract}
The standard asymmetric Laplace framework for Bayesian quantile regression (BQR) suffers from a fundamental decision-theoretic misalignment, yielding biased finite-sample estimates, and precludes gradient-based computation due to non-smoothness. We propose Bayesian smoothed quantile regression (BSQR), a principled framework built on a kernel-smoothed, fully differentiable likelihood. Methodologically, the symmetrizing property of our objective reduces inferential bias and aligns the posterior mean with the true conditional quantile. Theoretically, we establish posterior consistency and a Bernstein--von Mises theorem under misspecification, delivering asymptotic normality and valid frequentist coverage via a generalized Wilks phenomenon, while guaranteeing global posterior existence unlike empirical likelihood approaches. Computationally, BSQR enables Hamiltonian Monte Carlo for BQR, alleviating high-dimensional mixing bottlenecks. In simulations, BSQR reduces out-of-sample prediction error by up to 50\% and improves sampling efficiency by up to 80\% relative to asymmetric Laplace benchmarks, with uniform and triangular kernels performing particularly well. In a financial application to asymmetric systemic risk, BSQR uncovers distinct regime shifts around the COVID-19 period and yields sharper yet well-calibrated predictive quantiles, underscoring its practical relevance.

\vspace{1em}
	
  
\end{abstract}

\section{Introduction}\label{sec:introduction}

Quantile regression (QR) \shortcite{Koenker1978} constitutes a cornerstone of modern econometrics, enabling robust characterization of distributional heterogeneity essential for financial risk management and policy evaluation \shortcite{Chernozhukov2013}. Rigorous specification testing and inference for quantiles have been extensively studied in the frequentist literature \shortcite{Linton2005,Whang2006}. In the Bayesian paradigm, QR is prized for coherent uncertainty quantification and flexibility in hierarchical and latent-variable modeling. In practice it is used for tail risk, inequality, and heterogeneous treatment effects, often with short time spans or moderately high-dimensional covariates. In such settings we find that standard Bayesian quantile regression (BQR) can be systematically misleading: Bayesian quantile regression based on the asymmetric Laplace distribution (BQR-ALD) produces severely biased and under-covering posteriors at extreme quantiles in small samples, penalized exponentially tilted empirical likelihood (PETEL) \shortcite{Tang2022} exhibits poor mixing once $p$ becomes even moderately large, and in systemic risk applications conventional BQR-ALD delivers unstable dynamic betas with erratic credible bands.

Traditionally, BQR relies on the ALD working likelihood \shortcite{Yu2001} and its scale-mixture representation for Gibbs sampling \shortcite{Kozumi2011}. This device makes BQR practically accessible, but it faces two fundamental bottlenecks. First, as emphasized by \shortciteA{Gneiting2011} and \shortciteA{Sriram2013}, the ALD-based posterior is decision-theoretically misaligned: while asymptotically consistent, the posterior mean minimizes expected squared error rather than check loss, violating proper scoring rules and inducing severe finite-sample bias, especially in the tails. Second, the non-smooth check loss precludes state-of-the-art gradient-based Markov chain Monte Carlo (MCMC) algorithms such as Hamiltonian Monte Carlo (HMC) \shortcite{Neal2011}, effectively confining practitioners to conditionally Gaussian augmentation schemes whose mixing deteriorates rapidly with dimension, hierarchical structure, or extreme quantiles.

Recent work seeks to move beyond these limitations via generalized Bayesian inference (GBI) \shortcite{Bissiri2016} and joint quantile estimation frameworks \shortcite{Yang2017}. PETEL \shortcite{Tang2022} delivers a moment-based posterior but inherits the geometric fragility of empirical likelihood: posterior support can be empty or nearly empty in finite samples, mixing in high dimensions is often extremely poor, and extensions to hierarchical or latent-variable models are structurally cumbersome. Likelihood-free approaches such as the martingale posterior \shortcite{Fong2025} instead treat the loss as a scoring rule and optimize predictive performance without an explicit parametric likelihood. While powerful for forecasting, they obscure the direct covariate--quantile link that underpins structural interpretation, marginal effects, and policy counterfactuals. In short, existing approaches either retain the non-smooth check-loss geometry and forgo gradient-based computation, impose fragile empirical-likelihood constraints that scale poorly, or abandon parametric likelihoods and thereby weaken structural interpretability.

We propose Bayesian smoothed quantile regression (BSQR), a principled framework that bridges frequentist smoothing and robust Bayesian inference. Building on the ``conquer'' convolution-smoothing strategy \shortcite{Fernandes2021,He2023} and the quasi-Bayesian paradigm \shortcite{Chernozhukov2003}, we construct a principled, fully differentiable pseudo-likelihood whose negative log-density coincides with a kernel-smoothed check loss. This formulation explicitly manages the classical bias--variance trade-off in smoothing \shortcite{Gozalo2000}, restores decision-theoretic coherence with the underlying quantile loss, defines a valid error distribution on $\mathbb{R}$ that avoids convex-hull pathologies, and unlocks the substantial efficiency gains of HMC and the No-U-Turn sampler (NUTS) \shortcite{Hoffman2014} for BQR.

Our contributions advance the literature along three dimensions: methodological realignment, theoretical rigor, and computational scalability. Methodologically, we exploit a symmetrizing property of convolution smoothing to rectify the decision-theoretic flaw of standard BQR. Using the smoothed quantile regression (SQR) objective of \shortciteA{He2023}, we define a new error distribution whose negative log-density is proportional to the smoothed check loss and show that the pseudo-true parameter coincides with the minimizer of expected check loss. This aligns the posterior mean with the conditional quantile decision problem and eliminates the finite-sample bias induced by squared-error optimality in BQR-ALD. The smoothed loss admits closed-form first and second derivatives in residuals, enabling efficient gradient-based sampling for regression and scale parameters. The resulting pseudo-likelihood is strictly positive on $\mathbb{R}$, ensuring global posterior existence and avoiding empty-support phenomena that can afflict empirical-likelihood-based methods. In simulations this rectification reduces out-of-sample check loss by up to 40--50\% relative to BQR-ALD and markedly improves tail behavior and coverage, especially for extreme quantiles and short samples.

Theoretically, we go beyond simple consistency and provide a comprehensive asymptotic and geometric analysis. We establish posterior consistency for the true linear conditional quantile under mild regularity and an undersmoothing regime for the bandwidth, despite using a smoothed pseudo-likelihood. We then prove a Bernstein--von Mises (BvM) theorem under possible misspecification, leveraging the $C^2$-smoothness of the loss to obtain a classical local asymptotic normality expansion and an explicit limiting covariance for the centered posterior, thereby justifying credible intervals for frequentist coverage. Complementing this, we derive a generalized Wilks-type result: by calibrating the scale parameter to satisfy a generalized information equality, BSQR credible sets achieve asymptotically correct frequentist coverage even under misspecification; both plug-in and fully Bayesian treatments of this scale are analyzed and shown to enhance robustness. Furthermore, we analyze posterior propriety under various priors and derive geometric insights linking kernel ``peakedness'' to posterior concentration, establishing tail equivalence with ALD posteriors for compact kernels and clarifying the bias--variance and robustness trade-offs inherent in smoothing.

Computationally, this work marks, to our knowledge, the first systematic deployment of gradient-based sampling for BQR with an explicit likelihood and BvM justification. Smoothing the objective renders the posterior jointly differentiable in regression and scale parameters, enabling joint HMC/NUTS updates in standard probabilistic programming platforms such as \textsf{Stan} \shortcite{Carpenter2017}. We show that the normalizing constant of the smoothed likelihood is log-convex in the scale parameter, implying a log-concave conditional posterior under standard priors and numerically stable gradient evaluations. Extensive simulations and a systemic risk application demonstrate that BSQR substantially reduces out-of-sample prediction error by up to 50\% relative to BQR-ALD and significantly enhances sampler efficiency (up to 80\% in our applications), maintaining robust coverage in extreme tails. In higher-dimensional settings (e.g., $p=50$), it achieves effective sample sizes per second that are an order of magnitude larger than those of PETEL, whose chains often suffer near-complete mixing collapse due to the geometric fragility of empirical likelihood constraints. In our systemic risk study, BSQR produces smoother and more persistent dynamic betas with tighter yet well-calibrated credible bands, while matching or improving the forecasting performance of BQR-ALD.

Collectively, these innovations position BSQR as a rigorous bridge between frequentist accuracy and Bayesian probabilistic richness, retaining the structural interpretability prized in applied econometrics while delivering decision-theoretically aligned inference, classical large-sample guarantees, and modern computational efficiency. The complete source code is publicly available for reproducibility.\footnote{The replication package is available at \url{https://github.com/BeauquinLau/BSQR}.}

The paper is organized as follows. Section~\ref{sec:model_formulation} presents the BSQR model and smoothed likelihood. Sections~\ref{sec:consistency}, \ref{sec:bvm}, and \ref{sec:prior_selection_and_propriety_actual_concise} establish the theoretical core, covering posterior consistency, the Bernstein--von Mises theorem, and posterior propriety. Section~\ref{sec:analysis_kernel_selection_sqr} analyzes kernel effects and guides kernel and bandwidth selection. Section~\ref{sec:sampling} details the HMC/NUTS computational framework. Sections~\ref{sec:simulation} and \ref{sec:empirical_application_final_para} present simulation evidence and an empirical application to asymmetric systemic risk. Section~\ref{sec:discussion} concludes. All proofs, algorithms, technical derivations, and simulation tables are provided in the appendices.

\section{The Bayesian smoothed quantile regression framework}\label{sec:model_formulation}

This section develops the BSQR framework. We first review the fundamental limitations of standard BQR to motivate our approach, then lay out the technical foundations of loss smoothing, and finally construct our principled Bayesian model upon this foundation.

\subsection{Motivation: The challenge of non-smoothness in Bayesian quantile inference}\label{subsec:motivation_and_background}

The linear QR model posits that for i.i.d. observations $(y_i, \boldsymbol{x}_i)$ from a joint distribution $(Y, \boldsymbol{X})$, the $\tau$-th conditional quantile of the response $Y$ is a linear function of covariates for a given $\tau \in (0,1)$:
\begin{align}\label{eq:bsqr_model}
Q_Y(\tau \mid \boldsymbol{x}_i) = \boldsymbol{x}_i^\top\boldsymbol{\beta}(\tau), \quad i=1,\dots,n,
\end{align}
where $Q_Y(\tau \mid \boldsymbol{x}_i) \coloneqq \inf \{q : F_{Y \mid \boldsymbol{X}}(q \mid \boldsymbol{x}_i) \geq \tau \}$ is defined by the conditional cumulative distribution function (CDF) $F_{Y \mid \boldsymbol{X}}$, $\boldsymbol{x}_i \in \mathbb{R}^d$ is a covariate vector including an intercept, and $\boldsymbol{\beta}(\tau) \in \mathbb{R}^d$ is the coefficient vector. This implies the error term $\varepsilon_i \coloneqq y_i - \boldsymbol{x}_i^\top\boldsymbol{\beta}(\tau)$ has a zero $\tau$-th conditional quantile, $Q_{\varepsilon}(\tau \mid \boldsymbol{x}_i) = 0$. The assumption of independence between $\varepsilon_i$ and $\boldsymbol{x}_i$ simplifies this to the unconditional requirement $F_{\varepsilon}(0) = \tau$.

The population coefficient vector $\boldsymbol{\beta}(\tau)$ is the minimizer of the expected check loss \shortcite{Koenker1978}:
\begin{align}\label{eq:population_objective_qr}
R(\boldsymbol{b}; \tau) \coloneqq \mathbb{E}\left[\rho_{\tau}(\zeta)\right] = \int_{-\infty}^{\infty} \rho_{\tau}(e) \, \mathrm{d}F_{\zeta}(e),
\end{align}
where $\zeta \coloneqq Y - \boldsymbol{X}^\top\boldsymbol{b}$ is the population-level residual for a candidate vector $\boldsymbol{b}$, $F_{\zeta}$ is its CDF, and $\rho_{\tau}(e) = e(\tau - \mathbb{I}(e<0))$ is the non-differentiable ``pinball'' loss. Correspondingly, the sample estimator $\hat{\boldsymbol{\beta}}(\tau)$ minimizes the empirical risk:
\begin{align}\label{eq:sample_objective_qr} 
\widehat{R}(\boldsymbol{b}; \tau) \coloneqq \frac{1}{n} \sum_{i=1}^{n} \rho_{\tau}(e_i(\boldsymbol{b})),
\end{align}
where $e_i(\boldsymbol{b}) \coloneqq y_i - \boldsymbol{x}_i^\top\boldsymbol{b}$.

The BQR paradigm connects to this objective via the ALD, whose probability density function (PDF) is $p_{\mathrm{ALD}}(\cdot; \theta, \tau) \propto \exp(-\theta\rho_\tau(\cdot))$ \shortcite{Yu2001}. While this formulation ensures that its posterior mode, which we denote $\check{\boldsymbol{\beta}}(\tau)$, numerically coincides with the frequentist point estimator $\hat{\boldsymbol{\beta}}(\tau)$, this reliance on the non-smooth check loss imposes critical limitations. First, it precludes the use of modern gradient-based samplers (e.g., HMC). Second, it creates an inferential misalignment: the posterior mean is a biased estimator of the true conditional quantile and does not minimize the expected check loss, the canonical strictly proper scoring rule for quantile forecasts \shortcite{Gneiting2011, Sriram2013}. These challenges motivate a fundamental reformulation of the Bayesian likelihood.

\subsection{Technical foundation: Kernel smoothing of the check loss}\label{subsec:technical_foundation_smoothing}

To address the non-differentiability, we adapt kernel smoothing\,---\,a technique with a rich history in econometrics for enhancing estimation efficiency and utilizing parametric information \shortcite{Gozalo2000}\,---\,to construct a differentiable parametric likelihood. Specifically, we adopt the convolution-type smoothing strategy recently advanced in the frequentist literature, particularly the ``conquer'' framework established by \shortciteA{Fernandes2021} and \shortciteA{He2023}. While these works utilized smoothing primarily to facilitate gradient-based M-estimation and bootstrap inference, we leverage their formulation as the structural basis for our Bayesian likelihood.

Following \shortciteA{He2023}, the smoothed check loss is constructed by replacing the indicator function in the standard check loss with a smooth approximation derived from a kernel. Formally, let $K(\cdot)$ be a symmetric, non-negative kernel density function integrating to one, and let $K_h(v) = h^{-1}K(v/h)$ be the scaled kernel with bandwidth $h>0$. The \textit{smoothed check loss function} (see Figure~\ref{fig:smoothing_concept}) is defined as the convolution of the standard check loss $\rho_{\tau}(\cdot)$ and the kernel $K_h(\cdot)$:
\begin{align}\label{eq:smoothed_check_loss}
  L_h(e; \tau) &\coloneqq (\rho_{\tau} * K_h)(e) = \int_{-\infty}^{\infty} \rho_{\tau}(e-v) K_h(v) \,\mathrm{d}v \notag \\
  &= \int_{-\infty}^{\infty} \rho_{\tau}(u) K_h(u-e) \,\mathrm{d}u.
\end{align}

\begin{figure}[htbp]
\centering
\resizebox{0.55\textwidth}{!}{%
\begin{tikzpicture}
    \definecolor{acadRed}{RGB}{215, 25, 28}
    \definecolor{acadBlue}{RGB}{44, 123, 182}
    \definecolor{acadGray}{RGB}{80, 80, 80}

    \begin{axis}[
        width=8.5cm, height=6.5cm,
        axis lines=middle,
        axis line style={-Stealth, thick, color=black!80},
        xlabel={$e$ (Residual)},
        ylabel={Loss $\rho(e)$},
        xlabel style={at={(ticklabel* cs:1)}, anchor=west, xshift=2pt},
        ylabel style={at={(ticklabel* cs:1)}, anchor=south, yshift=2pt},
        ymin=-0.2, ymax=1.6,
        xmin=-2, xmax=2,
        xtick=\empty, ytick=\empty,
        legend style={
            at={(0.68, 1.15)},      
            anchor=north west,
            draw=black!40,
            thin,
            rounded corners=2pt,
            fill=white,
            fill opacity=0.95,
            text opacity=1,
            font=\scriptsize,
            inner xsep=5pt, inner ysep=4pt
        }
    ]
    \def\tauval{0.7}

    \addplot [
        domain=-1.8:1.8, 
        samples=200, 
        color=acadGray, 
        dashed,
        line width=1.2pt
    ]
    {x * (\tauval - (x<0))};
    \addlegendentry{Standard Check Loss $\rho_\tau(\cdot)$}

    \addplot [
        domain=-1.8:1.8, 
        samples=300, 
        color=acadBlue, 
        very thick,
        opacity=0.9
    ]
    { (x * (2*\tauval - 1) + sqrt(x^2 + 0.2^2))/2 }; 
    \addlegendentry{Smoothed Loss $L_h(\cdot;\tau)$}
    
    \node[below left, color=gray, font=\footnotesize] at (axis cs:0,0) {$0$};

    \node[
        circle,
        fill=acadRed,
        fill opacity=0.15,
        draw=acadRed,
        thin,
        minimum size=0.6cm,
        inner sep=0pt,
        anchor=center
    ] at (axis cs:0, 0.1) {};
    
    \node[anchor=west, acadRed, font=\scriptsize] (NoteText) at (axis cs: 0.55, 1.2) 
        {\textbf{Diff. at 0} ($C^2$)};
    
    \draw[->, >=Stealth, thick, acadRed] 
        (NoteText.west) to[out=180, in=45] (axis cs:0.1, 0.25);

    \draw[dotted, gray!60, thick] (axis cs:0,0) -- (axis cs:0, 0.2);

    \end{axis}
\end{tikzpicture}
}
\captionsetup{width=0.9\textwidth, justification=justified, font=small}
\caption{Smoothing the quantile objective. Comparison between the non-differentiable check loss (dashed gray) and the kernel-smoothed loss $L_h(\cdot; \tau)$ (solid blue). The smoothing removes the singularity at the origin (red circle), enabling gradient-based sampling via HMC.}
\label{fig:smoothing_concept}
\end{figure}

Based on this smoothed loss, the smoothed quantile regression (SQR) objective function is defined as the average smoothed loss over the sample. The resulting SQR estimator, denoted as $\hat{\boldsymbol{\beta}}_h(\tau)$, is:
\begin{align}\label{sqr_obj}
\hat{\boldsymbol{\beta}}_h(\tau) \coloneqq \arg\min_{\boldsymbol{b} \in \mathbb{R}^d} \widehat{R}_h(\boldsymbol{b}; \tau,h), \quad \text{where}\,\, \widehat{R}_h(\boldsymbol{b}; \tau,h) = \frac{1}{n}\sum_{i=1}^{n} L_h(e_i(\boldsymbol{b}); \tau).
\end{align}

A distinct advantage of this formulation, which we exploit for HMC sampling, is its differentiability. The derivative of the smoothed loss, denoted as $\Psi_h(e; \tau)$, is the convolution of the check loss subderivative $\psi_\tau(e) = \tau - \mathbb{I}(e<0)$\footnote{Strictly speaking, $\psi_\tau(e) = \tau - \mathbb{I}(e<0)$ is a subderivative, as $\rho_\tau(e)$ is non-differentiable at $e=0$ due to the discontinuity in $\psi_\tau$. However, this non-differentiability does not affect the subsequent smoothing and computation, as the convolution with a sufficiently smooth kernel $K_h$ yields a differentiable $\Psi_h(\cdot;\tau)$.} with the kernel:
\begin{align}\label{eq:Psi_h_definition}
\Psi_h(e; \tau) \coloneqq \frac{\partial L_h(e; \tau)}{\partial e} = (\psi_\tau * K_h)(e) = \int_{-\infty}^{\infty} \psi_\tau(e - v) K_h(v) \,\mathrm{d}v.\footnotemark
\end{align}
\footnotetext{The interchange of differentiation and integration is justified by the dominated convergence theorem. The partial derivative of the integrand with respect to $e$, namely $\psi_\tau(e-v)K_h(v)$, is bounded in absolute value by $\max(\tau, 1-\tau)K_h(v)$. This dominating function is integrable over $v \in \mathbb{R}$ since $K_h$ is a probability density, thus satisfying the conditions for the theorem.}By evaluating this convolution integral, we explicitly derive a closed-form expression involving the kernel's CDF, $F_K(\cdot)$:
\begin{align}\label{eq:Psi_h_cdf_form}
\Psi_h(e; \tau) &= \tau \int_{-\infty}^{\infty} K_h(v) \,\mathrm{d}v - \int_{-\infty}^{\infty} \mathbb{I}(e-v < 0) K_h(v) \,\mathrm{d}v \notag \\
&= \tau - \int_{e}^{\infty} K_h(v) \,\mathrm{d}v = \tau - \left(1 - F_K\left(\frac{e}{h}\right)\right) \notag \\
&= F_K\left(\frac{e}{h}\right) - (1-\tau).
\end{align}

This analytical gradient is fundamental for efficient HMC-based Bayesian inference, as detailed in Section~\ref{sec:sampling}. Furthermore, to elucidate the geometry of the posterior distribution sampled by HMC, we explicitly derive the second derivative (the curvature). Differentiating the smoothed score function $\Psi_h(e; \tau)$ from Eq.~\eqref{eq:Psi_h_definition} with respect to $e$:
\begin{align*}
    \frac{\partial \Psi_h(e; \tau)}{\partial e} &= \frac{\partial}{\partial e} \left( (\psi_\tau * K_h)(e) \right) = (\psi_\tau' * K_h)(e).
\end{align*}
The derivative of $\psi_\tau(u) = \tau - \mathbb{I}(u<0)$ is $\psi_\tau'(u) = \delta(u)$, where $\delta(u)$ is the Dirac delta function.\footnote{Strictly speaking, the classical derivative of $\psi_\tau(u)$ does not exist at $u=0$ due to the discontinuity. However, in the distributional sense, it is given by $\psi_\tau'(u) = \delta(u)$. To see this, note that we can write $\psi_\tau(u) = \tau - 1 + \mathbb{I}(u \geq 0)$. The distributional derivative of the indicator function $\mathbb{I}(u \geq 0)$ is the Dirac delta $\delta(u)$, leading to $\psi_\tau'(u) = \delta(u)$. This does not affect the convolution, which yields a smooth second derivative.} Thus, we obtain
\begin{align}\label{eq:hessian_derivation_base}
    L_h''(e; \tau) \coloneqq \frac{\partial \Psi_h(e; \tau)}{\partial e} &= (\delta * K_h)(e) = \int_{-\infty}^{\infty} \delta(v) K_h(e-v) \,\mathrm{d} v \notag \\
    &= K_h(e-0) = \frac{1}{h}K\left(\frac{e}{h}\right),
\end{align}
by utilizing the sifting property of the Dirac delta function, which states $\int_{-\infty}^{\infty} f(x)\delta(x-a)\,\mathrm{d} x = f(a)$.

Consequently, the Hessian of the sum of smoothed losses (the SQR objective) with respect to $\boldsymbol{\beta}$ takes the form of a weighted covariance matrix:
\begin{align}\label{eq:hessian_matrix_structure}
    \boldsymbol{H}_{\text{SQR}}(\boldsymbol{\beta}) \coloneqq \nabla^2_{\boldsymbol{\beta}} \left( \sum_{i=1}^n L_h(e_i(\boldsymbol{\beta}); \tau) \right) = \sum_{i=1}^n \frac{1}{h} K\left(\frac{e_i(\boldsymbol{\beta})}{h}\right) \boldsymbol{x}_i \boldsymbol{x}_i^\top.
\end{align}
This explicit Hessian structure serves a dual purpose in our framework: it characterizes the posterior curvature essential for constructing the mass matrix in HMC, ensuring efficient exploration of the posterior parameter space, and it is central to our asymptotic analysis in Section~\ref{sec:asymptotics} as well as the discussion of kernel effects in Section~\ref{subsec:kernel_effects_concentration}.

\subsection{The BSQR likelihood and posterior}\label{subsec:bscr_likelihood_construction}

While previous works utilized the smoothed loss $L_h(\cdot;\tau)$ primarily as a computational surrogate for fast optimization or bootstrap inference \shortcite{He2023}, constructing a principled Bayesian framework requires a valid probability density. Since the standard ALD-based likelihood relies on the non-smooth check loss $\rho_\tau(\cdot)$ and generally $L_h(\cdot;\tau) \neq \rho_\tau(\cdot)$, the ALD posterior does not target the smoothed estimator $\hat{\boldsymbol{\beta}}_h(\tau)$. To achieve both computational efficiency and inferential coherence, we formulate a new likelihood grounded in the smoothing geometry.

We propose a novel error distribution, denoted as $f_{\text{SQR}}$, constructed explicitly with a negative log-density proportional to the smoothed check loss. Unlike the ALD, our distribution is smooth and explicitly depends on the bandwidth $h$:
\begin{align*}
f_{\text{SQR}}(\varepsilon; \theta, \tau, h) = \frac{1}{Z(\theta, \tau, h)} \exp\left( -\theta L_h(\varepsilon; \tau) \right),
\end{align*}
where $\theta > 0$ is a scale parameter. A critical component here, which is absent in the optimization-based ``conquer'' framework, is the normalizing constant $Z(\theta, \tau, h)$ required for $f_{\text{SQR}}$ to be a valid PDF:
\begin{align}\label{eq:normalizing_constant_Z}
Z(\theta, \tau, h) \coloneqq \int_{-\infty}^{\infty} \exp\left( -\theta L_h(u; \tau) \right) \,\mathrm{d}u.
\end{align}
This constant ensures that $f_{\text{SQR}}$ is a proper density, a feature essential for Bayesian model selection and the valid sampling of $\theta$.

For an observed dataset $\boldsymbol{y} = (y_1, \dots, y_n)^\top$ and $\boldsymbol{\mathcal{X}} = (\boldsymbol{x}_1^\top, \dots, \boldsymbol{x}_n^\top)^\top$, under the assumption of independence, the joint likelihood for parameters $\boldsymbol{\beta}$ and $\theta$ is:
\begin{align}\label{eq:sqr_likelihood}
L(\boldsymbol{y} \mid \boldsymbol{\mathcal{X}}, \boldsymbol{\beta}, \theta; \tau, h) &= \prod_{i=1}^n f_{\text{SQR}}(e_i(\boldsymbol{\beta}); \theta, \tau, h) \notag \\
&= \left( Z(\theta, \tau, h) \right)^{-n} \exp\left( - \theta \sum_{i=1}^n L_h(e_i(\boldsymbol{\beta}); \tau) \right).
\end{align}
The log-likelihood is thus:
\begin{align}\label{eq:sqr_log_likelihood}
\ell(\boldsymbol{y} \mid \boldsymbol{\mathcal{X}}, \boldsymbol{\beta}, \theta; \tau, h) = -n \log Z(\theta, \tau, h) - \theta \sum_{i=1}^n L_h(e_i(\boldsymbol{\beta}); \tau).
\end{align}
By construction, maximizing this likelihood with respect to $\boldsymbol{\beta}$ is equivalent to minimizing the SQR objective (Eq.~\eqref{sqr_obj}), establishing a principled link between our Bayesian model and the frequentist SQR estimator. With independent priors $\pi(\boldsymbol{\beta}, \theta) = \pi(\boldsymbol{\beta})\pi(\theta)$, the posterior is:
\begin{align*}
\pi(\boldsymbol{\beta}, \theta \mid \boldsymbol{y}, \boldsymbol{\mathcal{X}}; \tau, h) \propto \left(Z(\theta, \tau, h)\right)^{-n} \exp\left( -\theta \sum_{i=1}^n L_h(e_i(\boldsymbol{\beta}); \tau) \right) \pi(\boldsymbol{\beta}) \pi(\theta).
\end{align*}
This formulation overcomes the limitations of the standard framework by providing: (1) differentiability for efficient HMC sampling (Section~\ref{sec:sampling}); (2) a consistent estimator that resolves the inferential bias of BQR-ALD (Section~\ref{sec:consistency}); and (3) a theoretically valid basis for uncertainty quantification, backed by Bernstein-von Mises theorems (Section~\ref{sec:bvm}) and rigorous propriety conditions (Section~\ref{sec:prior_selection_and_propriety_actual_concise}).

\section{Posterior consistency and bias resolution}\label{sec:consistency}

Posterior consistency is foundational. While standard ALD-based BQR is asymptotically consistent \shortcite{Sriram2013}, its posterior mean suffers from critical ``inferential misalignment'' in finite samples by minimizing squared error rather than the target check loss \shortcite{Gneiting2011}. Although joint estimation frameworks like \shortciteA{Tang2022} and \shortciteA{Yang2017} address structural issues, they typically entail heavy computational burdens and struggle with mixing in high-dimensional settings, impeding hierarchical extensions. In contrast, BSQR resolves this misalignment directly through loss smoothing. This section establishes the consistency of the BSQR posterior under the standard assumption of a correctly specified model, demonstrating that its symmetrizing property ensures inferences are correctly centered and unbiased. We extend this analysis to allow for model misspecification in Section~\ref{sec:bvm}.

\begin{assumption}[Regularity conditions for consistency]\label{as:consistency_conditions}
Let the true data generating process be governed by a linear conditional quantile model, such that $Q_Y(\tau \mid \boldsymbol{x}_i) = \boldsymbol{x}_i^\top\boldsymbol{\beta}_0(\tau)$, where $\boldsymbol{\beta}_0(\tau)$ is the true parameter vector. Define the error term as $\varepsilon_{0i} = y_i - \boldsymbol{x}_i^\top\boldsymbol{\beta}_0(\tau)$. By construction, the $\tau$-th conditional quantile of $\varepsilon_{0i}$ given $\boldsymbol{x}_i$ is zero, which implies that $F_{\varepsilon_0}(0) = \tau$. This defines the error structure implied by the model in Eq. \eqref{eq:bsqr_model} under the true parameter $\boldsymbol{\beta}_0(\tau)$. The following conditions are assumed to hold:

\begin{enumerate}[label=\textbf{\textup{(A\arabic*)}}]
    \item \textbf{Parameter Space:} The parameter space $\mathcal{B}$ for $\boldsymbol{\beta}(\tau)$ is a compact subset of $\mathbb{R}^d$, with the true parameter $\boldsymbol{\beta}_0(\tau)$ in its interior.
    
    \item \textbf{Error Distribution:} The error term $\varepsilon_{0i}$ satisfies the following:
    \begin{enumerate}[label=\textbf{\textup{(\alph*)}}]
        \item The CDF $F_{\varepsilon_0}(\cdot)$ is at least twice continuously differentiable in a neighborhood of 0, with a PDF $f_{\varepsilon_0}(\cdot) = F'_{\varepsilon_0}(\cdot)$ that is positive at 0, i.e., $f_{\varepsilon_0}(0) > 0$.
        \item The first moment of the error is finite, i.e., $\mathbb{E}[|\varepsilon_{0i}|] < \infty$.
    \end{enumerate}

    \item \textbf{Covariates:} The covariate vectors $\boldsymbol{x}_i$ are i.i.d. and uniformly bounded. The population second moment matrix $\Sigma_X \coloneqq \mathbb{E}[\boldsymbol{x} \boldsymbol{x}^\top]$ is positive definite.
    
    \item \textbf{Kernel Properties:} The kernel function $K(\cdot)$ is a symmetric, bounded, and continuous PDF satisfying $\mu_1(K) \coloneqq \int_{-\infty}^{\infty} u K(u) \,\mathrm{d}u = 0$ and $\mu_2(K) \coloneqq \int_{-\infty}^{\infty} u^2 K(u) \,\mathrm{d}u < \infty$.
    
    \item \textbf{Prior Distribution:} The prior distribution $\pi(\boldsymbol{\beta})$ is continuous and assigns positive probability to any open neighborhood of the true parameter $\boldsymbol{\beta}_0(\tau)$.
\end{enumerate}
\end{assumption}

\begin{theorem}[Posterior consistency of BSQR]\label{thm:consistency}
Under \textbf{Assumption}~\ref{as:consistency_conditions}, and assuming $h \to 0$ as $n\to \infty$, the posterior distribution $\pi(\boldsymbol{\beta} \mid \boldsymbol{y}, \boldsymbol{\mathcal{X}},\theta)$ derived from the smoothed likelihood in Eq.~\eqref{eq:sqr_likelihood} is consistent at the true parameter value $\boldsymbol{\beta}_0(\tau)$. That is, for any neighborhood $U$ of $\boldsymbol{\beta}_0(\tau)$,
\begin{align*}
    \int_U \pi(\boldsymbol{\beta} \mid \boldsymbol{y}, \boldsymbol{\mathcal{X}},\theta) \,\mathrm{d}\boldsymbol{\beta} \xrightarrow{P} 1 \quad \text{as } n \to \infty,
\end{align*}
where the convergence is in probability with respect to the true data generating distribution.
\end{theorem}
\noindent The proof of this theorem is provided in \hyperref[pf:consistency]{Appendix A.1}.

Theorem~\ref{thm:consistency} provides the theoretical cornerstone for BSQR. It demonstrates that our model retains the foundational property of posterior consistency while simultaneously resolving the computational and finite-sample limitations that plague the standard BQR-ALD. The proof hinges on the symmetrizing property of our smoothing procedure: convolving the asymmetric check loss with a symmetric kernel yields an objective function that is asymptotically unbiased in expectation, as shown by the pivotal result that $\mathbb{E}[\Psi_{h}(\varepsilon_{0i}; \tau)] \to 0$. This property ensures that, in large samples, BSQR inferences\,---\,such as posterior means and credible intervals\,---\,are correctly centered on the true parameters, thereby providing a foundation for the superior out-of-sample performance observed in our simulations (Section~\ref{subsec:sim_general}), particularly at extreme quantiles where the finite-sample bias of BQR-ALD is most severe.

\section{Limiting distribution and frequentist validity}\label{sec:bvm}

While Theorem~\ref{thm:consistency} establishes posterior consistency, rigorous inference requires characterizing the limiting posterior shape and ensuring its validity even under potential model misspecification\,---\,a central theme in recent GBI literature \shortcite{Matsubara2022, Tang2022}. Unlike standard BQR, where the non-smooth check loss necessitates complex empirical process theory, BSQR's kernel-smoothed loss $L_h(\cdot; \tau)$ is twice continuously differentiable (guaranteed by the continuity of $K(\cdot)$ in \textbf{Assumption A4}). This regularity allows us to establish a BvM theorem via a classical local asymptotic normality (LAN) expansion, addressing both posterior existence and frequentist coverage.

\subsection{Global existence and well-posedness}

A limitation of standard empirical likelihood approaches is that the posterior support is restricted to the convex hull of the estimating equations. While recent methods such as PETEL \shortcite{Tang2022} introduce penalization to alleviate this emptiness problem, BSQR ensures global validity naturally without requiring auxiliary penalty terms or relaxation of moment conditions.

\begin{proposition}[Finite-sample global existence]\label{prop:existence}
    For any sample size $n \ge 1$ and bandwidth $h > 0$, the BSQR posterior density $\pi(\boldsymbol{\beta} \mid \boldsymbol{y}, \boldsymbol{\mathcal{X}})$ is strictly positive and well-defined over the entire parameter space $\mathbb{R}^d$, for any valid kernel PDF $K(\cdot)$. Unlike standard empirical likelihood approaches, BSQR does not suffer from the ``empty set'' pathology where the posterior vanishes due to convex hull constraints.
\end{proposition}
\noindent The detailed proof is provided in \hyperref[pf:existence]{Appendix A.2}.

\subsection{Rate of convergence and bias control}

Before establishing the asymptotic normality of the posterior, it is essential to characterize the convergence rate of the frequentist smoothed estimator $\hat{\boldsymbol{\beta}}_h(\tau)$. This rate determines the allowable range of bandwidths to ensure that the smoothing bias does not contaminate the asymptotic distribution.

\begin{proposition}[Rate of convergence]\label{prop:rate}
    Under \textbf{Assumption}~\ref{as:consistency_conditions}, the smoothed quantile estimator $\hat{\boldsymbol{\beta}}_h(\tau)$ satisfies:
    \begin{align}
        \| \hat{\boldsymbol{\beta}}_h(\tau) - \boldsymbol{\beta}_0(\tau) \| = O_p\left(n^{-1/2} + h^2\right).
    \end{align}
    The term $O_p(n^{-1/2})$ represents the stochastic variation (variance), while $O(h^2)$ represents the deterministic approximation bias introduced by the kernel smoothing.
\end{proposition}
\noindent The detailed proof is provided in \hyperref[pf:rate]{Appendix A.2}.

\begin{remark}[Justification for undersmoothing]
    Proposition~\ref{prop:rate} reveals that the smoothing bias is of order $O(h^2)$. To ensure that the limiting distribution is centered at the truth $\boldsymbol{\beta}_0(\tau)$ without a bias shift, we require the scaled bias to vanish: $\sqrt{n} \cdot O(h^2) \to 0$, which implies $n h^4 \to 0$. This motivates the bandwidth rate condition formally stated as \textbf{Assumption (A6)} in the next subsection.
\end{remark}

\subsection{Bernstein-von Mises theorem under misspecification}\label{sec:asymptotics}

We now derive the asymptotic normality of the posterior. While frequentist approaches often employ specification tests to validate quantile models \shortcite{Whang2006}, our Bayesian framework rigorously accounts for potential misspecification via the generalized Bernstein-von Mises theorem. To ensure robustness, we do not assume the linear quantile model is correctly specified. Instead, following \shortciteA{Matsubara2022}, we define the target parameter $\boldsymbol{\beta}^*_h(\tau)$ as the minimizer of the expected smoothed risk (the pseudo-truth):
\[
\boldsymbol{\beta}^*_h(\tau) \coloneqq \arg\min_{\boldsymbol{b}} \mathbb{E}[L_h(e_i(\boldsymbol{b}); \tau)].
\]
Note that the smoothed loss $L_h(\cdot; \tau)$ is convex and converges pointwise to the check loss $\rho_\tau(\cdot)$ as $h \to 0$. Consequently, assuming the linear specification is valid (i.e., the true conditional quantile satisfies $Q_{Y}(\tau \mid \boldsymbol{x}_i) = \boldsymbol{x}_i^\top \boldsymbol{\beta}_0(\tau)$), standard M-estimation theory for convex objective functions ensures that the pseudo-truth $\boldsymbol{\beta}^*_h(\tau)$ converges to the true parameter $\boldsymbol{\beta}_0(\tau)$ \cite{Horowitz1998a}. We define the population Hessian and score covariance matrices at this pseudo-truth:
\begin{align}
    \mathcal{H}(\boldsymbol{\beta}^*_h(\tau)) &\coloneqq \mathbb{E}\left[ \nabla^2_{\boldsymbol{\beta}} L_h(e_i(\boldsymbol{\beta}^*_h(\tau)); \tau) \right], \label{eq:hessian_H_pop} \\
    \mathcal{J}(\boldsymbol{\beta}^*_h(\tau)) &\coloneqq \mathbb{E}\left[ \nabla_{\boldsymbol{\beta}} L_h(e_i(\boldsymbol{\beta}^*_h(\tau)); \tau) \nabla_{\boldsymbol{\beta}} L_h(e_i(\boldsymbol{\beta}^*_h(\tau)); \tau)^\top \right]. \label{eq:jacobian_J_pop}
\end{align}

\begin{assumption}[Regularity for asymptotic normality]\label{as:bvm_conditions}
    In addition to \textbf{Assumption}~\ref{as:consistency_conditions}, we assume:
    \begin{enumerate}[label=\textbf{\textup{(A\arabic*)}}, start=6]
        \item \textbf{Bandwidth Rate:} The bandwidth $h \to 0$ satisfies $n h^4 \to 0$ and $n h^2 \to \infty$. The upper bound condition $n h^4 \to 0$ ensures the smoothing bias is asymptotically negligible ($o(n^{-1/2})$), a standard requirement for centering confidence intervals \cite{Horowitz1998a}. Conversely, the lower bound $n h^2 \to \infty$ is required to control the variance of the sample Hessian, ensuring its uniform convergence to the population counterpart (see condition \textbf{A7} below).
        \item \textbf{Hessian Regularity:} The population Hessian matrix $\mathcal{H}(\boldsymbol{\beta}^*_h(\tau))$ is positive definite. Furthermore, the sample Hessian matrix converges uniformly in probability to the population Hessian $\mathcal{H}(\boldsymbol{\beta})$ over a neighborhood of the pseudo-truth $\boldsymbol{\beta}^*_h(\tau)$, i.e., there exists $\delta > 0$ such that
\begin{align*}
    \sup_{\|\boldsymbol{\beta} - \boldsymbol{\beta}^*_h(\tau)\| \le \delta} \left\| \frac{1}{n} \sum_{i=1}^n \nabla^2_{\boldsymbol{\beta}} L_h(e_i(\boldsymbol{\beta}); \tau) - \mathcal{H}(\boldsymbol{\beta}) \right\| \xrightarrow{P} 0.
\end{align*}
    \end{enumerate}
\end{assumption}

\begin{theorem}[Bernstein-von Mises theorem under misspecification]\label{thm:bvm}
    Under \textbf{Assumptions}~\ref{as:consistency_conditions} and \ref{as:bvm_conditions}, for a fixed scale parameter $\theta > 0$, the posterior density of the centered parameter $\boldsymbol{\vartheta} = \sqrt{n}(\boldsymbol{\beta} - \hat{\boldsymbol{\beta}}_h(\tau))$ converges in total variation distance to the multivariate normal density:
    \begin{align*}
        \left\| \pi(\boldsymbol{\vartheta} \mid \boldsymbol{y}, \boldsymbol{\mathcal{X}}, \theta) - \phi\left(\boldsymbol{\vartheta}; \boldsymbol{0}, \big(\theta \mathcal{H}(\boldsymbol{\beta}^*_h(\tau))\big)^{-1}\right) \right\|_{TV} \xrightarrow{P} 0, \quad \text{as } n \to \infty,
    \end{align*}
    where $\|\cdot\|_{TV}$ denotes the total variation distance between densities, $\phi(\cdot)$ is the probability density function of the multivariate normal distribution, $\hat{\boldsymbol{\beta}}_h(\tau)$ is the minimizer of the sample smoothed loss, and $\boldsymbol{\beta}^*_h(\tau)$ is the pseudo-true parameter.
\end{theorem}

\noindent The proof exploits the $C^2$ smoothness of $L_h(\cdot; \tau)$ to apply a Taylor expansion of the log-posterior, avoiding the empirical process theory required for non-smooth losses (see \hyperref[pf:bvm]{Appendix A.2}).

\begin{remark}[Choice of kernel and asymptotic universality]
    While \textbf{Assumption (A4)} requires kernel continuity to facilitate LAN expansion via classical calculus, Theorem \ref{thm:bvm} remains conceptually robust to this choice. Technically, the kernel shape affects asymptotic bias and variance solely through scalar constants $\mu_2(K) = \int_{-\infty}^{\infty} u^2 K(u)\,\mathrm{d} u$ and $\nu_2(K) = \int_{-\infty}^{\infty} K(u)^2 \,\mathrm{d} u$. We focus on continuous kernels to obviate the complex empirical process theory required for discontinuous counterparts (e.g., Uniform) \cite{Horowitz1998a}, noting that they typically yield equivalent limiting behaviors in practice.
\end{remark}

\begin{remark}[Regularity and efficiency gains]
    The differentiability of $L_h(\cdot; \tau)$ (\textbf{Assumption A4}) simplifies asymptotics by anchoring BSQR within the classical smooth LAN framework, obviating the strong stochastic equicontinuity assumptions required for non-smooth scores \cite{Koenker2005, Tang2022}. Furthermore, smoothing eliminates the $O(n^{-1/2})$ oscillation terms of the check loss. While \shortciteA{Horowitz1998a} established that this yields second-order frequentist efficiency, in the Bayesian context, operating on this regularized objective landscape avoids the local instabilities of standard methods, ensuring a theoretically robust foundation for inference.
\end{remark}

However, valid frequentist inference for $\boldsymbol{\beta}_0(\tau)$ faces two impediments. First, under model misspecification, standard quasi-Bayesian posteriors exhibit covariance mismatch \shortcite{Tang2022}: the posterior covariance $(\theta \mathcal{H}(\boldsymbol{\beta}^*_h(\tau)))^{-1}$ generally diverges from the true ``sandwich'' variance $\mathcal{H}(\boldsymbol{\beta}^*_h(\tau))^{-1} \mathcal{J}(\boldsymbol{\beta}^*_h(\tau)) \mathcal{H}(\boldsymbol{\beta}^*_h(\tau))^{-1}$. Departing from the empirical likelihood remedy of \shortciteA{Tang2022}, we calibrate the scale parameter of the smoothed objective directly to align these variances, following the sandwich adjustment strategy of \shortciteA{Mueller2013}. Second, smoothing introduces deterministic bias. The following corollary establishes that combining an undersmoothing schedule (controlling bias) with a generalized information equality condition (calibrating variance) yields valid asymptotic frequentist coverage.

\begin{corollary}[Asymptotic coverage and generalized Wilks' phenomenon]\label{cor:validity_true}
    Suppose the bandwidth satisfies the undersmoothing condition $n h^4 \to 0$ as specified in \textbf{Assumption (A6)}. Let $\hat{\theta}$ be a consistent estimator of the population scale parameter satisfying the generalized information equality (first-moment matching):
    \begin{align}\label{eq:calibration_condition}
        \hat{\theta} \xrightarrow{P} \theta^* \coloneqq \frac{d}{\mathrm{tr}\left( \mathcal{J}(\boldsymbol{\beta}^*_h(\tau)) \mathcal{H}(\boldsymbol{\beta}^*_h(\tau))^{-1} \right)},
    \end{align}
    where $d$ denotes the dimension of $\boldsymbol{\beta}$.\footnote{\label{fn:plug_in}Specifically, such an estimator is given by the plug-in form $\hat{\theta} = d / \mathrm{tr}( \widehat{\boldsymbol{J}}_n(\hat{\boldsymbol{\beta}}_h(\tau)) [ \widehat{\boldsymbol{H}}_n(\hat{\boldsymbol{\beta}}_h(\tau)) ]^{-1})$, where $\widehat{\boldsymbol{J}}_n(\hat{\boldsymbol{\beta}}_h(\tau))$ is the sample covariance of scores.} Based on this calibration, let $C_{1-\alpha}$ be the $(1-\alpha)$-level Bayesian credible set defined as the Wald-type ellipsoidal region:
    \[
    C_{1-\alpha} \coloneqq \left\{ \boldsymbol{\beta} : n(\boldsymbol{\beta} - \hat{\boldsymbol{\beta}}_h(\tau))^\top (\hat{\theta} \mathcal{H}(\boldsymbol{\beta}^*_h(\tau))) (\boldsymbol{\beta} - \hat{\boldsymbol{\beta}}_h(\tau)) \le \chi^2_{d, 1-\alpha} \right\},
    \]
    where $\chi^2_{d, 1-\alpha}$ is the $(1-\alpha)$-quantile of the chi-squared distribution with $d$ degrees of freedom.\footnote{\label{eq:sample_Hessian}In practice, the population Hessian $\mathcal{H}(\boldsymbol{\beta}^*_h(\tau))$ is unknown and is replaced by the sample Hessian $\widehat{\boldsymbol{H}}_n(\hat{\boldsymbol{\beta}}_h(\tau)) \coloneqq n^{-1}\sum_{i=1}^n \nabla^2_{\boldsymbol{\beta}} L_h(e_i(\hat{\boldsymbol{\beta}}_h(\tau)); \tau)$. By \textbf{Assumption (A7)} (uniform convergence) and Slutsky's theorem, this replacement does not alter the asymptotic distribution of the quadratic form.}
 
    Then, this Bayesian credible set achieves nominal asymptotic frequentist coverage for the true parameter $\boldsymbol{\beta}_0(\tau)$:
    \begin{align*}
        \mathbb{P}_{\boldsymbol{y} \mid \boldsymbol{\beta}_0(\tau)} \left( \boldsymbol{\beta}_0(\tau) \in C_{1-\alpha} \right) \to 1 - \alpha, \quad \text{as } n \to \infty.
    \end{align*}
    This result establishes a generalized Wilks' theorem for the smoothed quantile objective, guaranteeing that the calibrated posterior uncertainty is valid in the frequentist sense to the first order of approximation.
\end{corollary}

\noindent The detailed proof is provided in \hyperref[pf:cor_validity_true]{Appendix A.2}.

\noindent\textbf{Example (Analytic calibration in standard cases).} 
To build intuition for the abstract condition in Eq.~\eqref{eq:calibration_condition}, consider the standard i.i.d. linear model defined in Eq.~\eqref{eq:bsqr_model}, where errors satisfy $f_{\varepsilon}(0)>0$. In this scenario, as the bandwidth $h \to 0$, the smoothed population matrices approach the classical limits established in standard quantile regression theory \shortcite{Koenker2005}:
\[
    \mathcal{H}(\boldsymbol{\beta}^*_h(\tau)) \to f_{\varepsilon}(0) \Sigma_X \quad \text{and} \quad \mathcal{J}(\boldsymbol{\beta}^*_h(\tau)) \to \tau(1-\tau) \Sigma_X,
\]
where $\Sigma_X = \mathbb{E}[\boldsymbol{x}\boldsymbol{x}^\top]$. We apply the trace calibration formula from Eq.~\eqref{eq:calibration_condition} to these limits. First, we compute the product matrix inside the trace:
\[
    \mathcal{J}(\boldsymbol{\beta}^*_h(\tau)) \left[\mathcal{H}(\boldsymbol{\beta}^*_h(\tau))\right]^{-1} \to (\tau(1-\tau)\Sigma_X) (f_{\varepsilon}(0)\Sigma_X)^{-1} = \frac{\tau(1-\tau)}{f_{\varepsilon}(0)} \boldsymbol{I}_d,
\]
where $\boldsymbol{I}_d$ is the $d \times d$ identity matrix. Its trace is simply $d \cdot \frac{\tau(1-\tau)}{f_{\varepsilon}(0)}$. Next, substituting this trace value into the definition of the target parameter yields:
\[
    \theta^* = \frac{d}{\mathrm{tr}\left( \mathcal{J}(\boldsymbol{\beta}^*_h(\tau)) \left[\mathcal{H}(\boldsymbol{\beta}^*_h(\tau))\right]^{-1} \right)} = \frac{d}{d \cdot \frac{\tau(1-\tau)}{f_{\varepsilon}(0)}} = \frac{f_{\varepsilon}(0)}{\tau(1-\tau)}.
\]
This result explicitly recovers the standard scalar correction factor found in classical quantile regression theory. It demonstrates that our generalized trace calibration target $\theta^*$ naturally simplifies to the correct theoretical scaling in well-specified linear models. Consequently, the plug-in estimator $\hat{\theta}$ defined in Footnote~\ref{fn:plug_in} will consistently estimate this classical factor, ensuring valid coverage without requiring prior knowledge of the error density. The geometric intuition behind this variance rescaling is illustrated in Figure~\ref{fig:calibration_concept}.

\begin{figure}[htbp]
\centering
\resizebox{0.5\textwidth}{!}{%
\begin{tikzpicture}
    \definecolor{acadRed}{RGB}{228, 26, 28}    
    \definecolor{acadBlue}{RGB}{55, 126, 184}  
    \definecolor{acadGray}{RGB}{77, 77, 77}    

    \begin{axis}[
        width=10cm, height=7.5cm,
        axis lines=left,            
        axis line style={-Stealth, thick, color=black!80}, 
        xlabel={Parameter $\beta$},
        ylabel={Density},
        ymin=0, ymax=1,           
        xmin=-4.5, xmax=4.5,
        xtick=\empty,
        ytick=\empty,
        legend pos=north east,      
        legend cell align=left,     
        legend style={
            draw=black!40, thin, rounded corners=2pt,
            fill=white, fill opacity=0.95, text opacity=1,
            font=\scriptsize, inner xsep=5pt, inner ysep=4pt
        },
        axis on top=false 
    ]

    \addplot [domain=-4.5:4.5, samples=300, color=acadRed, thick, fill=acadRed, fill opacity=0.12, smooth]
    {1/(0.6*sqrt(2*pi)) * exp(-((x-0)^2)/(2*0.6^2))};
    \addlegendentry{Uncalibrated Posterior ($\theta=1$)}

    \addplot [domain=-4.5:4.5, samples=200, color=acadBlue, very thick, opacity=0.8, smooth]
    {1/(1.0*sqrt(2*pi)) * exp(-((x-0)^2)/(2*1.0^2))};
    \addlegendentry{Calibrated Posterior ($\hat{\theta} \xrightarrow{P} \theta^{*}$)}

    \addplot [domain=-4.5:4.5, samples=200, color=acadGray, dashed, line width=1.1pt, smooth]
    {1/(1.0*sqrt(2*pi)) * exp(-((x-0)^2)/(2*1.0^2))};
    \addlegendentry{Sampling Distribution (Sandwich)}

    \draw[->, >=Stealth, thick, acadBlue] (axis cs:0.45, 0.45) to[bend left=20] (axis cs:1.3, 0.25);
    \draw[->, >=Stealth, thick, acadBlue] (axis cs:-0.45, 0.45) to[bend right=20] (axis cs:-1.3, 0.25);

    \node[anchor=south, acadBlue, font=\scriptsize] at (axis cs: 0, 0.45) {Scale adjustment};

    \end{axis}
\end{tikzpicture}
}
\captionsetup{width=0.8\textwidth, justification=justified, font=small}
\caption{Geometric intuition of Bayesian calibration. Under model misspecification, the uncalibrated posterior (red) is typically overly concentrated (overconfident). The scale parameter $\theta$ rescales the posterior covariance (blue) to match the frequentist sampling variance (dashed gray), ensuring valid asymptotic coverage.}\label{fig:calibration_concept}
\end{figure}

\begin{remark}[Bayesian implementation strategy]
    While Corollary~\ref{cor:validity_true} establishes validity by fixing the scale to a consistent estimator $\hat{\theta}$ (e.g., the plug-in form derived in Footnote~\ref{fn:plug_in}), our practical implementation adopts a fully Bayesian approach that treats $\theta$ as a random parameter. This integrates out the scale uncertainty, offering a robust alternative to plug-in calibration. We provide a detailed econometric justification for this strategy in Remark~\ref{rem:inference_strategy}.
\end{remark}

\section{Posterior propriety under various prior specifications}\label{sec:prior_selection_and_propriety_actual_concise}

This section investigates posterior propriety for the BSQR model under common prior choices for the regression coefficients $\boldsymbol{\beta}$ and the scale parameter $\theta$. We analyze cases ranging from an improper uniform prior for $\boldsymbol{\beta}$ to hierarchical structures. Our analysis uses the BSQR likelihood from Eq.~\eqref{eq:sqr_likelihood} and the shorthand $S(\boldsymbol{\beta}; \tau, h) \coloneqq \sum_{i=1}^n L_h(e_i(\boldsymbol{\beta}); \tau)$ for the sum of smoothed losses.

Although Corollary~\ref{cor:validity_true} establishes asymptotic validity based on a consistent estimator $\hat{\theta}$, finite-sample Bayesian practice treats $\theta$ as a random parameter to capture scale uncertainty. Establishing posterior propriety under this specification is therefore essential to validate the joint Markov chain Monte Carlo (MCMC) sampling (Section~\ref{sec:sampling}). We structure this analysis in three progressive steps: first deriving baseline conditions under an improper uniform prior (Section~\ref{subsec:improper_uniform_beta}), extending them to a standard Gaussian prior (Theorem~\ref{thm:propriety2}), and finally validating the hierarchical Gaussian specification (Corollary~\ref{cor:1}) used in our implementation.

\subsection{Propriety with an improper uniform prior for \texorpdfstring{$\boldsymbol{\beta}$}{beta}}\label{subsec:improper_uniform_beta}

We begin by analyzing posterior propriety under an improper uniform prior, $\pi(\boldsymbol{\beta}) \propto 1$. This setting is particularly insightful, as it isolates the likelihood's contribution to integrability from the regularizing effects of a proper prior for $\boldsymbol{\beta}$. We first establish baseline conditions for a fixed $\theta$ and then generalize to settings where $\theta$ has its own prior. Theorem~\ref{thm:propriety1}, proved in \hyperref[pf:thm_propriety1]{Appendix A.3}, summarizes these results.

\begin{theorem}[Propriety under improper uniform prior for $\boldsymbol{\beta}$]\label{thm:propriety1}
Let the kernel function $K(\cdot)$ be non-negative, integrate to unity, and possess a finite first absolute moment (i.e., $\int_{-\infty}^{\infty} |u|K(u)\,\mathrm{d} u < \infty$). Separately, assume the $n \times d$ design matrix $\boldsymbol{\mathcal{X}}$ has full column rank $d$ and $\pi(\boldsymbol{\beta}) \propto 1$.
\begin{enumerate}
	\item[(i)] (Fixed $\theta$) If the scale parameter $\theta > 0$ is fixed, then the posterior distribution of $\boldsymbol{\beta}$, $\pi(\boldsymbol{\beta} \mid \boldsymbol{y}, \boldsymbol{\mathcal{X}}, \theta)$, is proper. Equivalently,
\begin{align*}
    0 < \int_{\mathbb{R}^d} L(\boldsymbol{y} \mid \boldsymbol{\mathcal{X}}, \boldsymbol{\beta}, \theta; \tau, h) \pi(\boldsymbol{\beta}) \,\mathrm{d}\boldsymbol{\beta} < \infty.
\end{align*}

	\item[(ii)] (Prior for $\theta$) If the scale parameter $\theta$ is assigned a prior distribution $\pi(\theta)$, then for a certain constant $C_S \ge 0$, the joint posterior distribution $\pi(\boldsymbol{\beta}, \theta \mid \boldsymbol{y}, \boldsymbol{\mathcal{X}})$ is proper if the integral
\begin{align}\label{eq:integral_condition_theta_prior}
    \int_0^\infty \frac{\pi(\theta) e^{\theta C_S}}{(Z(\theta, \tau, h))^n \theta^d} \,\mathrm{d}\theta
\end{align}
converges to a finite positive value. As a specific instance, if $\pi(\theta)\sim \text{Gamma}(a,b)$ with $a,b > 0$, and further assuming that $L_h(\cdot; \tau)$ is twice continuously differentiable and attains its positive global minimum $L_{\min} = \min_u L_h(u; \tau) > 0$ at a unique point $u_{\min}$ where $L_h''(u_{\min};\tau)>0$, the posterior is proper if $b > C_S + n L_{\min}$ and $a + k_Z > d$, where $k_Z \ge 0$ is a constant such that $(Z(\theta, \tau, h))^{-n} = O(\theta^{k_Z})$ as $\theta \to 0$.
\end{enumerate}
\end{theorem}

\begin{remark}[Interplay between prior and sample size]
For Theorem~\ref{thm:propriety1}(ii), the condition $b > C_S + n L_{\min}$ on the Gamma prior hyperparameter $b$ is a significant result. Unlike standard sample-size-independent conditions, it reveals a crucial interplay between the prior specification ($b$), the sample size ($n$), and the properties of the smoothed loss function ($L_{\min}$). This mandates that the prior on $\theta$ must become increasingly informative (i.e., have a larger rate $b$) as $n$ grows, to ensure propriety under an improper prior for $\boldsymbol{\beta}$.

This requirement stems from the powerful influence of the likelihood, where $n$ appears as an exponent in the normalizing constant term $(Z(\theta, \tau, h))^{-n}$. Its asymptotic behavior for large $\theta$, derived via Laplace's method, is $(Z(\theta, \tau, h))^{-n} = O\!\left(\theta^{n/2} e^{n\theta L_{\min}}\right)$.
For the posterior to be integrable, the exponential decay of the Gamma prior ($e^{-b\theta}$) must dominate this exponential growth, which directly yields the condition $b > C_S + n L_{\min}$. The second condition, $a+k_Z > d$, is a more conventional safeguard against non-identifiability as $\theta \to 0$.
\end{remark}

\subsection{Propriety with proper priors for \texorpdfstring{$\boldsymbol{\beta}$}{beta}}

We now transition to proper priors for $\boldsymbol{\beta}$, which can enhance regularization and simplify propriety arguments. A conventional choice, seen in many BQR methodologies \shortcite{Kozumi2011,Li2010}, is a Gaussian prior, often with a large variance to create a weakly informative yet mathematically convenient specification. Theorem~\ref{thm:propriety2}, proved in \hyperref[pf:thm_propriety2]{Appendix A.3}, establishes the propriety conditions under this prior.

\begin{theorem}[Propriety under Gaussian prior for $\boldsymbol{\beta}$]\label{thm:propriety2}
Let the kernel function $K(\cdot)$ be non-negative and integrate to unity, and the prior for $\boldsymbol{\beta}$ be Gaussian, $\pi(\boldsymbol{\beta} \mid \sigma_{\boldsymbol{\beta}}^2) = (2\pi \sigma_{\boldsymbol{\beta}}^2)^{-d/2} \exp\left( -\frac{\|\boldsymbol{\beta}\|_2^2}{2\sigma_{\boldsymbol{\beta}}^2} \right)$, with a fixed prior variance $\sigma_{\boldsymbol{\beta}}^2 > 0$.

\begin{itemize}
    \item[(i)] (Fixed $\theta$) If the scale parameter $\theta > 0$ is fixed, the posterior distribution $\pi(\boldsymbol{\beta} \mid \boldsymbol{y}, \boldsymbol{\mathcal{X}}, \theta, \sigma_{\boldsymbol{\beta}}^2)$ is proper. That is,
        \begin{align*}
          0 < \int_{\mathbb{R}^d} L(\boldsymbol{y} \mid \boldsymbol{\mathcal{X}}, \boldsymbol{\beta}, \theta; \tau, h) \pi(\boldsymbol{\beta} \mid \sigma_{\boldsymbol{\beta}}^2) \,\mathrm{d}\boldsymbol{\beta} < \infty.
        \end{align*}

	\item[(ii)] (Prior for $\theta$) If the scale parameter $\theta$ is assigned a proper prior distribution $\pi(\theta)$, then the joint posterior distribution $\pi(\boldsymbol{\beta}, \theta \mid \boldsymbol{y}, \boldsymbol{\mathcal{X}}, \sigma_{\boldsymbol{\beta}}^2)$ is proper if the integral
        \begin{align*}
          \int_0^\infty (Z(\theta, \tau, h))^{-n} \pi(\theta) \,\mathrm{d}\theta
        \end{align*}
        converges to a finite positive value. As a specific instance, if $\pi(\theta) \sim \text{Gamma}(a_\theta, b_\theta)$ with $a_\theta > 0$ and $b_\theta > 0$, the joint posterior is proper if $b_\theta > n L_{\min}$ and $a_\theta + k_Z > 0$, where $L_{\min} = \min_u L_h(u; \tau)$ and $k_Z$ is defined as in Theorem~\ref{thm:propriety1}.
\end{itemize}
\end{theorem}

A natural extension treats the prior variance $\sigma_{\boldsymbol{\beta}}^2$ as random, allowing the data to inform its scale while increasing flexibility. Assigning $\sigma_{\boldsymbol{\beta}}^2$ an Inverse-Gamma hyperprior, favored for its positive support and tractability, leads to the propriety result in the following corollary. The detailed proof is provided in \hyperref[pf:cor_1]{Appendix A.3}.

\begin{corollary}[Propriety under hierarchical Gaussian prior for $\boldsymbol{\beta}$]\label{cor:1}
Let the kernel $K(\cdot)$ be non-negative and integrate to unity.
The prior for $\boldsymbol{\beta}$ is conditionally Gaussian: $\pi(\boldsymbol{\beta} \mid \sigma_{\boldsymbol{\beta}}^2) = (2\pi \sigma_{\boldsymbol{\beta}}^2)^{-d/2} \exp\left( -\frac{\|\boldsymbol{\beta}\|_2^2}{2\sigma_{\boldsymbol{\beta}}^2} \right)$, and the hyperprior for $\sigma_{\boldsymbol{\beta}}^2$ is Inverse-Gamma: $\pi(\sigma_{\boldsymbol{\beta}}^2 \mid a_0, b_0) = \text{IG}(\sigma_{\boldsymbol{\beta}}^2 \mid a_0, b_0)$ with $a_0 > 0$ and $b_0 > 0$.

\begin{itemize}
    \item[(i)] (Fixed $\theta$) If $\theta > 0$ is fixed, the marginal posterior distribution $\pi(\boldsymbol{\beta}, \sigma_{\boldsymbol{\beta}}^2 \mid \boldsymbol{y}, \boldsymbol{\mathcal{X}}, \theta, a_0, b_0)$ is proper. Equivalently,
        \begin{align*}
          0 < \int_0^\infty \int_{\mathbb{R}^d} L(\boldsymbol{y} \mid \boldsymbol{\mathcal{X}}, \boldsymbol{\beta}, \theta; \tau, h) \pi(\boldsymbol{\beta}\mid \sigma_{\boldsymbol{\beta}}^2) \pi(\sigma_{\boldsymbol{\beta}}^2\mid a_0, b_0) \,\mathrm{d}\boldsymbol{\beta} \,\mathrm{d}\sigma_{\boldsymbol{\beta}}^2 < \infty.
        \end{align*}
	
	\item[(ii)] (Prior for $\theta$) If the scale parameter $\theta$ is assigned a proper prior distribution $\pi(\theta)$, then the joint posterior distribution $\pi(\boldsymbol{\beta}, \sigma_{\boldsymbol{\beta}}^2, \theta \mid \boldsymbol{y}, \boldsymbol{\mathcal{X}}, a_0, b_0)$ is proper if the integral
        \begin{align*}
          \int_0^\infty (Z(\theta, \tau, h))^{-n} \pi(\theta) \,\mathrm{d}\theta
        \end{align*}
        converges to a finite positive value.
        As a specific instance, if $\pi(\theta) \sim \text{Gamma}(a_\theta, b_\theta)$ with $a_\theta > 0$ and $b_\theta > 0$, the joint posterior is proper if $b_\theta > n L_{\min}$ and $a_\theta + k_Z > 0$, where $L_{\min}$ and $k_Z$ are defined as in Theorem~\ref{thm:propriety1}.
\end{itemize}
\end{corollary}

\begin{remark}[Relaxation of conditions via proper priors]
A proper prior for $\boldsymbol{\beta}$ (even if conditional) significantly simplifies propriety arguments. Since the likelihood term $\exp(-\theta S(\boldsymbol{\beta}; \tau, h))$ is bounded, integrability is largely governed by the proper priors on $\boldsymbol{\beta}$ and its variance. Consequently, the posterior is typically proper without requiring stringent conditions on the design matrix $\boldsymbol{\mathcal{X}}$ or on relationships between prior hyperparameters and the data dimension $d$.
\end{remark}

\section{Theoretical analysis of kernel selection}\label{sec:analysis_kernel_selection_sqr}

Having established the conditions for posterior propriety in Section~\ref{sec:prior_selection_and_propriety_actual_concise}, we now investigate how the choice of the kernel function $K(\cdot)$ shapes the specific geometry of the BSQR posterior. We first establish that for any compact support kernel, the BSQR posterior for $\boldsymbol{\beta}$ is equivalent in its tail behavior to that of standard ALD-based BQR. Second, we show that more ``peaked'' kernels yield a more concentrated posterior. These results provide a theoretical justification for our method and a principle for kernel selection.

\subsection{Implications of kernel selection}\label{subsec:kernel_selection_implications}

Building on Section~\ref{sec:prior_selection_and_propriety_actual_concise}, we examine how the choice of kernel $K(\cdot)$ influences the BSQR posterior. Two properties are central: (i) the behavior of the normalizing constant $Z(\theta, \tau, h)$, which governs both the likelihood and MCMC sampling of $\theta$; and (ii) the relationship between the BSQR and standard ALD posteriors for compact support kernels, clarifying the role of smoothing in inference on $\boldsymbol{\beta}$.

\begin{proposition}[Properties of $Z(\theta, \tau, h)$]\label{prop:Z_properties}
The function $Z(\theta, \tau, h)$, as defined in Eq.~\eqref{eq:normalizing_constant_Z}, is based on the smoothed loss $L_h(u; \tau)$ which is non-negative for all $u \in \mathbb{R}$ and not identically zero. Consequently, for fixed $\tau$ and $h$:
\begin{enumerate}
    \item[(i)] $Z(\theta, \tau, h)$ is a strictly decreasing function of $\theta > 0$.
    \item[(ii)] $\log Z(\theta, \tau, h)$ is a convex function of $\theta > 0$.
\end{enumerate}
\end{proposition}
\noindent The proof of this proposition is provided in \hyperref[pf:prop1]{Appendix A.4}.

\begin{remark}[Computational implications of log-concavity]
The convexity of $\log Z(\theta, \tau, h)$ is consequential. It implies that the $\theta$-dependent part of the log-likelihood (Eq.~\eqref{eq:sqr_log_likelihood}) is concave. If the prior $\pi(\theta)$ is also log-concave, the conditional log-posterior $\log \pi(\theta \mid \boldsymbol{\beta}, \boldsymbol{y}, \boldsymbol{\mathcal{X}};\tau,h)$ is also concave. Such log-concavity is highly beneficial for MCMC sampling, as it typically ensures unimodal and well-behaved posteriors, improving sampling efficiency and convergence for $\theta$.
\end{remark}

A pertinent question is how the BSQR posterior, $\pi_{\mathrm{BSQR}}(\boldsymbol{\beta} \mid \theta, \boldsymbol{y}, \boldsymbol{\mathcal{X}};\tau,h)$, relates to the ALD-based posterior, $\pi_{\mathrm{ALD}}(\boldsymbol{\beta} \mid \theta, \boldsymbol{y}, \boldsymbol{\mathcal{X}}; \tau)$, particularly in their tail behaviors. For compact-support kernels (e.g., Uniform, Epanechnikov, Triangular, with support normalized to $[-1,1]$), the smoothed loss $L_h(e; \tau)$ coincides with the check loss $\rho_\tau(e)$\,---\,up to an additive constant independent of $e$\,---\,whenever $|e/h| > 1$, thereby localizing smoothing to residuals near zero. This motivates a formal investigation into whether the two posteriors are \emph{equivalent} under such kernels. 

We define $L_{\mathrm{ALD}}(\boldsymbol{y} \mid \boldsymbol{\mathcal{X}}, \boldsymbol{\beta}, \theta; \tau) \propto \exp\left( -\theta \sum_{i=1}^n \rho_\tau(e_i(\boldsymbol{\beta})) \right)$, with a normalizing constant for the ALD PDF independent of $\boldsymbol{\beta}$. The corresponding posterior is $\pi_{\mathrm{ALD}}(\boldsymbol{\beta} \mid \theta, \boldsymbol{y}, \boldsymbol{\mathcal{X}};\tau) \propto L_{\mathrm{ALD}}(\boldsymbol{y} \mid \boldsymbol{\mathcal{X}}, \boldsymbol{\beta}, \theta; \tau) \pi(\boldsymbol{\beta})$. The following theorem, proved in \hyperref[pf:thm4]{Appendix A.4}, establishes this equivalence.

\begin{theorem}[Posterior equivalence for BSQR with compact kernels]\label{thm:4}
Assume that the kernel $K(\cdot)$ has compact support, say $[a,b]$ for finite $a < b$, and is symmetric around its mean.\footnote{Specifically, $K(\cdot)$ is symmetric with respect to its mean $\mu_1(K) = \int_{-\infty}^{\infty} u \, K(u) \, \mathrm{d}u$, meaning that $K(\mu_1 + v) = K(\mu_1 - v)$ for all $v$ in the support. In the proof, we re-center the kernel to have zero mean for simplicity.} Then, there exist positive constants $M_1$ and $M_2$, independent of $\boldsymbol{\beta}$, such that for all $\boldsymbol{\beta} \in \mathbb{R}^d$:
$$ M_1 \cdot \pi_{\mathrm{ALD}}(\boldsymbol{\beta} \mid \theta, \boldsymbol{y}, \boldsymbol{\mathcal{X}};\tau) \le \pi_{\mathrm{BSQR}}(\boldsymbol{\beta} \mid \theta, \boldsymbol{y}, \boldsymbol{\mathcal{X}};\tau,h) \le M_2 \cdot \pi_{\mathrm{ALD}}(\boldsymbol{\beta} \mid \theta, \boldsymbol{y}, \boldsymbol{\mathcal{X}};\tau). $$
\end{theorem}

\begin{remark}[Posterior equivalence with standard BQR]
Theorem~\ref{thm:4} establishes that, for compact support kernels, the BSQR posterior for $\boldsymbol{\beta}$ shares the same overall shape and tail behavior as the standard ALD posterior, up to a scaling factor absorbed by the normalizing constants. This theoretically justifies using BSQR as a computationally convenient alternative (e.g., for HMC) without fundamentally altering inferential conclusions about $\boldsymbol{\beta}$. This equivalence is distinct from the case of non-compact kernels (e.g., Gaussian), where the difference $L_h(\cdot; \tau) - \rho_\tau(\cdot)$ may be unbounded, potentially leading to different tail behaviors.
\end{remark}

\subsection{Kernel effects on posterior concentration}\label{subsec:kernel_effects_concentration}

Beyond tail behavior equivalence under compact-support kernels, the choice of kernel $K(\cdot)$ also influences the posterior's concentration. Intuitively, a more ``peaked'' kernel assigns greater weight to small residuals (scaled by $h$), penalizing deviations from zero more sharply and thus yielding a more concentrated posterior for $\boldsymbol{\beta}$. We formalize this by examining the Hessian of the negative log-likelihood component of the posterior.

Let $U_L(\boldsymbol{\beta}; \theta, h, K) = \theta \sum_{i=1}^n L_h(e_i(\boldsymbol{\beta}); \tau)$ be the primary component of the negative log-posterior (or potential energy function) that depends on $\boldsymbol{\beta}$ through the sum of smoothed losses, with its argument $K$ signifying dependence on the kernel function $K(\cdot)$. The full negative log-posterior is $U(\boldsymbol{\beta}) = U_L(\boldsymbol{\beta}; \theta, h, K) - \log \pi(\boldsymbol{\beta}) + C_{\theta,h}$, where $C_{\theta,h}$ collects terms not dependent on $\boldsymbol{\beta}$.

We focus on the Hessian of the likelihood component $U_L$, denoted as $\boldsymbol{H}_L(\boldsymbol{\beta}; K)$. Utilizing the rigorous derivation based on the Dirac delta function established in Section~\ref{subsec:technical_foundation_smoothing} (Eq.~\eqref{eq:hessian_derivation_base} and Eq.~\eqref{eq:hessian_matrix_structure}), we directly obtain:
\begin{align}\label{eq:hessian_likelihood_part}
    \boldsymbol{H}_L(\boldsymbol{\beta}; K) &\coloneqq \nabla^2_{\boldsymbol{\beta}} U_L(\boldsymbol{\beta}; \theta, h, K) = \theta \boldsymbol{H}_{\text{SQR}}(\boldsymbol{\beta}) = \frac{\theta}{h} \sum_{i=1}^n K\left(\frac{e_i(\boldsymbol{\beta})}{h}\right) \boldsymbol{x}_i \boldsymbol{x}_i^\top.
\end{align}
A larger Hessian (in the positive definite sense) at the posterior mode $\check{\boldsymbol{\beta}}(\tau)$ suggests a more sharply peaked posterior and, via Laplace approximation, a smaller posterior covariance.

As discussed in Section \ref{sec:consistency}, the true errors $\varepsilon_{0i} = y_i - \boldsymbol{x}_i^\top\boldsymbol{\beta}_0(\tau)$ are assumed to be independent and identically distributed following a common density $f_{\varepsilon_0}(\cdot)$, where $\boldsymbol{\beta}_0(\tau)$ represents the true value of $\boldsymbol{\beta}(\tau)$. Define $s_K(h) \coloneqq \mathbb{E}_{\varepsilon_0 \sim f_{\varepsilon_0}}\left[K\left(\frac{\varepsilon_0}{h}\right)\right]$.
Let $\mathcal{H}_L(\boldsymbol{\beta}_0(\tau); K) \coloneqq \mathbb{E}_{\varepsilon_{0i}, \boldsymbol{x}_i} [\boldsymbol{H}_L(\boldsymbol{\beta}_0(\tau); K)]$ denote the expected Hessian of the negative log-likelihood component, evaluated at the true parameter $\boldsymbol{\beta}_0(\tau)$. Note that under the i.i.d.\ assumption, this relates to the population Hessian of the loss defined in Eq.~\eqref{eq:hessian_H_pop} by a scaling factor: $\mathcal{H}_L(\boldsymbol{\beta}_0(\tau); K) = n \theta \mathcal{H}(\boldsymbol{\beta}_0(\tau))$. This leads to the following theorem with the proof provided in \hyperref[pf:thm5]{Appendix A.4}.

\begin{theorem}[Kernel peakedness and expected local curvature]\label{thm:5}
Let $K_A(v)$ and $K_B(v)$ be two distinct kernel functions that are non-negative, symmetric about zero, and share a common compact support. Assume that $s_{K_A}(h) > s_{K_B}(h)$, and that the covariates follow $\boldsymbol{x}_i \stackrel{i.i.d.}{\sim} P_{\boldsymbol{X}}$ where $P_{\boldsymbol{X}}$ has a positive definite covariance matrix $\Sigma_X \coloneqq \mathbb{E}[\boldsymbol{x}_i \boldsymbol{x}_i^\top]$ and finite second moments. Then, for the expected Hessian:
$$ \mathcal{H}_L(\boldsymbol{\beta}_0(\tau); K_A) \succ \mathcal{H}_L(\boldsymbol{\beta}_0(\tau); K_B), $$
where $\succ$ denotes the strict Loewner order (i.e., $A \succ B$ implies $A - B$ is positive definite). Furthermore, for the sample Hessian, the corresponding inequality
$$ \boldsymbol{H}_L(\boldsymbol{\beta}_0(\tau); K_A) \succ \boldsymbol{H}_L(\boldsymbol{\beta}_0(\tau); K_B) $$
holds with probability approaching 1 as $n \to \infty$.
\end{theorem}

\begin{remark}[Implications for posterior contraction]
Theorem~\ref{thm:5} establishes that a kernel $K_A$ satisfying $s_{K_A}(h) > s_{K_B}(h)$ yields a larger expected Hessian at the true parameter. Via Laplace approximation, this increased local curvature implies a more concentrated posterior for $\boldsymbol{\beta}$ and a smaller covariance matrix (i.e., $\mathcal{H}_L(\cdot; K_A)^{-1} \preceq \mathcal{H}_L(\cdot; K_B)^{-1}$), assuming the prior's Hessian is negligible. \end{remark}

\begin{remark}[Geometric intuition of kernel peakedness]
The condition $s_{K_A}(h) > s_{K_B}(h)$ means that kernel $K_A$ gives more average weight to the true scaled errors. This is likely to occur when a more ``peaked'' kernel $K_A$ aligns with a concentrated error density $f_{\varepsilon_0}$, as more error values will fall where $K_A(v) > K_B(v)$. For instance, a Triangular kernel will likely satisfy this condition over a Uniform kernel on $[-1,1]$ if most true errors are small, as it places more mass near zero.
\end{remark}

\begin{remark}[Finite-sample implications]
While Theorem~\ref{thm:5} is an asymptotic result at the true parameter $\boldsymbol{\beta}_0(\tau)$, its principle extends to the posterior mode $\check{\boldsymbol{\beta}}(\tau)$. A more peaked kernel may yield a more concentrated posterior if it produces a consistently larger sum $\sum_{i=1}^n K(e_i(\check{\boldsymbol{\beta}}(\tau))/h)$ than competing kernels, based on the empirical distribution of residuals.
\end{remark}

\section{Bayesian inference via MCMC}\label{sec:sampling}

Bayesian inference for the BSQR parameters\,---\,regression coefficients $\boldsymbol{\beta}$ and scale $\theta$\,---\,is challenging, particularly when targeting uncertainty quantification. Standard BQR exploits the ALD’s scale-mixture-of-normals form \shortcite{Kozumi2011, Yu2001}, enabling efficient Gibbs updates. However, the smoothed loss $L_h(\cdot; \tau)$ in our framework disrupts conditional conjugacy, rendering closed-form Gibbs sampling intractable. While Metropolis-Hastings within Gibbs is a viable alternative, it often suffers from slow mixing and requires careful tuning.

To overcome these limitations, we adopt the quasi-Bayesian framework of \shortciteA{Chernozhukov2003} for defining the posterior, and implement computation via Hamiltonian Monte Carlo (HMC) \shortcite{Duane1987, Neal2011} using the No-U-Turn sampler (NUTS) \shortcite{Hoffman2014} as implemented in \textsf{Stan} \shortcite{Carpenter2017}. Unlike component-wise sampling schemes, NUTS updates the joint parameter space $(\boldsymbol{\beta}, \theta)$ simultaneously using Hamiltonian dynamics, leveraging gradient information to explore the high-dimensional posterior robustly.

\subsection{Joint Hamiltonian Monte Carlo sampling}\label{sec:joint_hmc}
We target the joint posterior distribution $\pi(\boldsymbol{\beta}, \theta \mid \boldsymbol{y}, \boldsymbol{\mathcal{X}};\tau,h)$. The complete sampling procedure is summarized in Algorithm~\ref{Al1_Part1} of \hyperref[sec:AppendixB]{Appendix B}. The HMC sampler augments the parameter space with auxiliary momentum variables. 

To map the positively constrained scale parameter $\theta$ onto the unconstrained Euclidean space required by HMC, we utilize the logarithmic transformation (consistent with \textsf{Stan}'s automatic parameter unconstraining mechanism). Let $\boldsymbol{q} = (\boldsymbol{\beta}^\top, \log\theta)^\top$ denote the extended state vector and $\boldsymbol{p}$ denote the corresponding momentum vector. The system's dynamics are governed by the Hamiltonian function:
\[
H(\boldsymbol{q}, \boldsymbol{p}) = U(\boldsymbol{q} \mid \boldsymbol{y}, \boldsymbol{\mathcal{X}}) + T(\boldsymbol{p}),
\]
where $T(\boldsymbol{p}) = \frac{1}{2}\boldsymbol{p}^\top \boldsymbol{M}^{-1} \boldsymbol{p}$ is the kinetic energy with mass matrix $\boldsymbol{M}$ (typically diagonal and adapted during warmup), and $U(\boldsymbol{q} \mid \boldsymbol{y}, \boldsymbol{\mathcal{X}})$ is the potential energy. Since the sampling is performed on the unconstrained scale $\boldsymbol{q}$ (involving $\log \theta$), $U(\boldsymbol{q}\mid \boldsymbol{y}, \boldsymbol{\mathcal{X}})$ must include the Jacobian adjustment to ensure the correct marginal posterior density for $\theta$. Thus, the potential energy is defined as:
\begin{align}
    U(\boldsymbol{q} \mid \boldsymbol{y}, \boldsymbol{\mathcal{X}}) &= - \log \left( \pi(\boldsymbol{\beta}, \theta \mid \boldsymbol{y}, \boldsymbol{\mathcal{X}}) \cdot \left| \frac{\mathrm{d} \theta}{\mathrm{d} \log \theta} \right| \right) \nonumber \\
    &= - \left[ \ell(\boldsymbol{y} \mid \boldsymbol{\mathcal{X}}, \boldsymbol{\beta}, \theta; \tau, h) + \log \pi(\boldsymbol{\beta}) + \log \pi(\theta) + \log \theta \right] \nonumber \\
    &= \theta \sum_{i=1}^n L_h(e_i(\boldsymbol{\beta}); \tau) + n \log Z(\theta, \tau, h) - \log \pi(\boldsymbol{\beta}) - \log \pi(\theta) - \log \theta. \label{eq:joint_potential}
\end{align}

The evolution of the system over fictitious time $s$ follows Hamilton’s equations of motion \shortcite{Duane1987, Neal2011}:
\begin{align*}
\frac{\mathrm{d} \boldsymbol{q}(s)}{\mathrm{d} s} &= \frac{\partial H}{\partial \boldsymbol{p}(s)} = \boldsymbol{M}^{-1} \boldsymbol{p}(s), \\
\frac{\mathrm{d} \boldsymbol{p}(s)}{\mathrm{d} s} &= -\frac{\partial H}{\partial \boldsymbol{q}(s)} = -\nabla_{\boldsymbol{q}} U(\boldsymbol{q}(s)\mid \boldsymbol{y}, \boldsymbol{\mathcal{X}}).
\end{align*}
These equations preserve phase-space volume and are reversible. In practice, they are integrated using the leapfrog method (or its variant in NUTS). A single step of size $\epsilon$ updates the joint state $(\boldsymbol{q}, \boldsymbol{p})$ via:
\begin{align*}
    \boldsymbol{p}_{\text{half}} &= \boldsymbol{p} - \frac{\epsilon}{2} \nabla_{\boldsymbol{q}} U(\boldsymbol{q}\mid \boldsymbol{y}, \boldsymbol{\mathcal{X}}),\\
    \boldsymbol{q}' &= \boldsymbol{q} + \epsilon \boldsymbol{M}^{-1} \boldsymbol{p}_{\text{half}}, \\
    \boldsymbol{p}' &= \boldsymbol{p}_{\text{half}} - \frac{\epsilon}{2} \nabla_{\boldsymbol{q}} U(\boldsymbol{q}'\mid \boldsymbol{y}, \boldsymbol{\mathcal{X}}).
\end{align*}

Crucially, simulating these dynamics requires computing the gradients of the potential energy $\nabla_{\boldsymbol{q}} U(\boldsymbol{q} \mid \boldsymbol{y}, \boldsymbol{\mathcal{X}})$, which involve partial derivatives with respect to both $\boldsymbol{\beta}$ and $\theta$. For notational simplicity, we express these gradients in terms of the constrained parameters $(\boldsymbol{\beta}, \theta)$; however, in the actual HMC implementation, sampling is performed on the unconstrained state $\boldsymbol{q}=(\boldsymbol{\beta}^\top, \log\theta)^\top$, with the log-Jacobian determinant absorbed into $U(\boldsymbol{q} \mid \boldsymbol{y}, \boldsymbol{\mathcal{X}})$. For the regression coefficients $\boldsymbol{\beta}$, the gradient is derived analytically using the smoothed loss properties established in Section~\ref{subsec:technical_foundation_smoothing}:
\[
\nabla_{\boldsymbol{\beta}} U(\boldsymbol{\beta}, \theta \mid \boldsymbol{y}, \boldsymbol{\mathcal{X}}) = \theta \sum_{i=1}^n \nabla_{\boldsymbol{\beta}} L_h(e_i(\boldsymbol{\beta}); \tau) - \nabla_{\boldsymbol{\beta}} \log \pi(\boldsymbol{\beta}),
\]
where $\nabla_{\boldsymbol{\beta}} L_h(e_i;\tau) = -\Psi_h(e_i;\tau)\,\boldsymbol{x}_i$. Using the closed-form expression derived in Eq.~\eqref{eq:Psi_h_cdf_form}, the scalar score component is simply:
\begin{align*}
\Psi_h(e_i; \tau) = F_K\left(\frac{e_i}{h}\right) - (1-\tau),
\end{align*}
where $K(\cdot)$ is the kernel density and $F_K$ its CDF. We provide the specific forms of $\Psi_h(\cdot;\tau)$ and $L_h(\cdot;\tau)$ for various kernels in Sections~\ref{sec:6.1.1}--\ref{sec:6.1.4}.

For the scale parameter $\theta$, the gradient involves the derivative of the log-normalizing constant $\nabla_\theta \log Z(\theta, \tau, h)$. While this term involves an intractable integral that typically necessitates gradient-free methods like Metropolis-Hastings, we leverage \textsf{Stan}'s automatic differentiation capability. By implementing $Z(\theta, \tau, h)$ via numerical quadrature (e.g., \texttt{integrate\_1d}), the system automatically computes exact gradients through the integration operator, enabling fully gradient-based NUTS sampling for $\theta$ as well.

\begin{remark}[Choice of inference strategy]\label{rem:inference_strategy}
    While Corollary~\ref{cor:validity_true} relies on a consistent point estimator of $\hat{\theta}$ (i.e., the plug-in approach) for asymptotic validity, such calibration depends on sample estimates of the Hessian and score covariance, which can be numerically sensitive in finite samples \shortcite{Koenker2005}. We instead adopt a fully Bayesian approach, integrating $\theta$ out via a hierarchical prior. This strategy avoids the sensitivity of plug-in estimates and, consistent with the general quasi-Bayesian theory of \shortciteA{Chernozhukov2003}, yields robust empirical uncertainty quantification by averaging over the nuisance parameter's uncertainty (Section~\ref{sec:simulation}).
\end{remark}

\subsection{Kernel-specific smoothed gradients}\label{sec:kernel_derivations}

Efficient HMC sampling requires evaluating the smoothed loss $L_h(\cdot;\tau)$ and score $\Psi_h(\cdot;\tau)$ at each leapfrog step. Below, we present the closed-form expressions derived for common kernels; for brevity, detailed step-by-step derivations are provided in \hyperref[sec:AppendixC]{Appendix C}. These derivations confirm that bounded-support kernels (Uniform, Triangular, Epanechnikov) yield piecewise polynomial potentials. Since HMC dynamics rely primarily on gradient information (requiring only $C^1$ continuity), these kernels remain computationally efficient and robust despite minor discontinuities in higher-order derivatives.

\subsubsection{Gaussian kernel}\label{sec:6.1.1}
The standard Gaussian kernel is $K(v) = \phi(v) = (2\pi)^{-1/2} \exp(-v^2/2)$, with CDF $F_K(u) = \Phi(u)$.
Applying Eq.~\eqref{eq:Psi_h_cdf_form} yields:
\begin{equation}
\Psi_h(e; \tau) = \Phi\left(\frac{e}{h}\right) - (1-\tau). \label{eq:psi_h_gaussian_formal}
\end{equation}
The corresponding smoothed loss function $L_h(\cdot; \tau)$, obtained by integrating $\Psi_h(\cdot; \tau)$ and consistent with $\rho_\tau(\cdot)$ for small $h$, is given by (e.g., \shortciteNP{Horowitz1998, Koenker2005}):
\begin{equation*}
L_h(e; \tau) = e \left( \Phi\left(\frac{e}{h}\right) - (1-\tau) \right) + h \phi\left(\frac{e}{h}\right).
\end{equation*}

\subsubsection{Uniform kernel}\label{sec:6.1.2}
The standard Uniform kernel is $K(v) = \frac{1}{2}$ for $v \in [-1, 1]$ and $0$ otherwise. Its CDF is $F_K(u) = 0$ for $u < -1$, $F_K(u) = \frac{u+1}{2}$ for $-1 \le u \le 1$, and $F_K(u) = 1$ for $u > 1$.
From Eq.~\eqref{eq:Psi_h_cdf_form}, $\Psi_h(e; \tau)$ is:
\begin{equation*}
\Psi_h(e; \tau) = \begin{cases} -(1-\tau) & \text{if } e/h < -1, \\ \frac{e}{2h} + \tau - \frac{1}{2} & \text{if } -1 \le e/h \le 1, \\ \tau & \text{if } e/h > 1. \end{cases}
\end{equation*}
The function $L_h(e; \tau)$ is derived by integrating $\Psi_h(e; \tau)$ and imposing continuity with $\rho_\tau(e)$ for $|e/h| \ge 1$. The resulting expression is:
\begin{equation*}
L_h(e; \tau) = \begin{cases} e(\tau-1) & \text{if } e/h \le -1, \\ \frac{e^2}{4h} + e\left(\tau - \frac{1}{2}\right) + \frac{h}{4} & \text{if } -1 < e/h < 1, \\ e\tau & \text{if } e/h \ge 1. \end{cases}
\end{equation*}

\begin{remark}[Regularity at boundaries]
    Strictly speaking, the Uniform kernel is discontinuous at the boundaries, implying that the smoothed loss $L_h(\cdot; \tau)$ is only $C^1$ rather than $C^2$. This violates the continuity condition in \textbf{Assumption (A4)} required for the rigorous BvM theory (Theorem \ref{thm:bvm}). However, we include it here for completeness and empirical comparison, as HMC samplers can typically handle such discontinuities in the potential energy gradient without significant issues.
\end{remark}

\subsubsection{Epanechnikov kernel}\label{sec:6.1.3}
The standard Epanechnikov kernel is $K(v) = \frac{3}{4}(1-v^2)$ for $v \in [-1, 1]$ and $0$ otherwise.
Its CDF, $F_K(u)$, is $0$ for $u < -1$, $\frac{3}{4}u - \frac{1}{4}u^3 + \frac{1}{2}$ for $-1 \le u \le 1$, and $1$ for $u > 1$.
This yields $\Psi_h(\cdot; \tau)$ as:
\begin{equation*}
\Psi_h(e; \tau) = \begin{cases} -(1-\tau) & \text{if } e/h < -1, \\ \frac{3}{4}\frac{e}{h} - \frac{1}{4}\left(\frac{e}{h}\right)^3 + \tau - \frac{1}{2} & \text{if } -1 \le e/h \le 1, \\ \tau & \text{if } e/h > 1. \end{cases}
\end{equation*}
Integration of $\Psi_h(e; \tau)$ and matching boundary conditions with $\rho_\tau(e)$ for $|e/h| \ge 1$ leads to $L_h(\cdot; \tau)$:
\begin{equation*}
L_h(e; \tau) = \begin{cases} e(\tau-1) & \text{if } e/h \le -1, \\ \frac{3e^2}{8h} - \frac{e^4}{16h^3} + e\left(\tau - \frac{1}{2}\right) + \frac{3h}{16} & \text{if } -1 < e/h < 1, \\ e\tau & \text{if } e/h \ge 1. \end{cases}
\end{equation*}

\subsubsection{Triangular kernel}\label{sec:6.1.4}
The standard Triangular kernel is $K(v) = 1-|v|$ for $v \in [-1, 1]$ and $0$ otherwise.
The CDF, $F_K(u)$, is $0$ for $u < -1$; $\frac{1}{2}(1+u)^2$ for $-1 \le u < 0$; $1 - \frac{1}{2}(1-u)^2$ for $0 \le u \le 1$; and $1$ for $u > 1$.
Consequently, $\Psi_h(\cdot; \tau)$ is:
\begin{equation*}
\Psi_h(e; \tau) = \begin{cases} -(1-\tau) & \text{if } e/h < -1, \\
\frac{1}{2}\left(1+\frac{e}{h}\right)^2 - (1-\tau) & \text{if } -1 \le e/h < 0, \\
\tau - \frac{1}{2}\left(1-\frac{e}{h}\right)^2 & \text{if } 0 \le e/h \le 1, \\
\tau & \text{if } e/h > 1.
\end{cases}
\end{equation*}
The corresponding $L_h(e; \tau)$ is obtained by piecewise integration of $\Psi_h(e; \tau)$, ensuring continuity with $\rho_\tau(e)$ for $|e/h| \ge 1$ and at $e=0$:
\begin{equation*}
L_h(e; \tau) = \begin{cases} e(\tau-1) & \text{if } e/h \le -1, \\
\frac{h}{6}\left(1+\frac{e}{h}\right)^3 - e(1-\tau) & \text{if } -1 < e/h < 0, \\
e\tau + \frac{h}{6}\left(1-\frac{e}{h}\right)^3 & \text{if } 0 \le e/h < 1, \\
e\tau & \text{if } e/h \ge 1.
\end{cases}
\end{equation*}

\subsection{Computational implementation details}\label{sec:hmc_details}

The algorithm is implemented in \textsf{Stan} (code available in supplementary materials). To ensure numerical stability and computational efficiency in the NUTS, we employ two key strategies corresponding to specific blocks in our Stan code:

\begin{enumerate}
    \item \textbf{Hybrid Normalizing Constant Calculation:} The integral for $Z(\theta, \tau, h)$ can become numerically unstable when $\theta$ is large, as the integrand becomes sharply peaked. To address this, we implement a hybrid strategy: we switch to an asymptotic approximation for $\log Z(\theta, \tau, h)$ based on Laplace's method (for smooth kernels) or specific asymptotic expansions when $\theta$ exceeds a threshold, while maintaining robust numerical quadrature for standard ranges. This prevents numerical overflow/underflow during the gradient computation.
    
	\item \textbf{Non-Centered Parameterization:} In high-dimensional settings, the posterior geometry of the regression coefficients can exhibit pathological shapes (e.g., Neal's funnels) that hinder HMC exploration. We mitigate this by employing a non-centered parameterization \shortcite{Betancourt2013}: we sample auxiliary standard normal variables $\boldsymbol{\beta}_{\text{raw}} \sim \mathcal{N}(\boldsymbol{0}, \boldsymbol{I})$ and transform them via $\boldsymbol{\beta} = \boldsymbol{\mu}_{\beta} + \sigma_{\boldsymbol{\beta}} \boldsymbol{\beta}_{\text{raw}}$, where $\sigma_{\boldsymbol{\beta}}$ is the scalar prior scale. This explicitly separates the prior mean and scale from the geometry, significantly improving the effective sample size (ESS) by preconditioning the posterior.
\end{enumerate}

This fully gradient-based approach, combined with these numerical stabilizations, eliminates the need for manual tuning of proposal variances inherent in Metropolis-Hastings steps and ensures robust mixing.

\section{Simulation}\label{sec:simulation}

We evaluate the estimation accuracy, inferential validity, and scalability of the BSQR framework across three regimes: (1) general performance in standard settings (Section \ref{subsec:sim_general}); (2) inferential validity under extreme sparsity (Section \ref{subsec:sim_stress}); and (3) high-dimensional scalability (Section \ref{subsec:sim_scalability}). Unless otherwise noted in the specialized regimes of Sections \ref{subsec:sim_stress} and \ref{subsec:sim_scalability}, synthetic data are generated via the linear specification $y_i = \boldsymbol{x}_i^\top\boldsymbol{\beta} + u_i$, utilizing four error distributions to assess robustness: (1) $\mathcal{N}(0, 1)$; (2) $t(3)$; (3) a mixture $0.2\mathcal{N}(0, 3) + 0.8\mathcal{N}(0, 4)$; and (4) heteroscedastic normal errors $u_i \mid \boldsymbol{x}_i \sim \mathcal{N}(0, \sigma_i^2)$ with $\sigma_i = \exp(-0.25 + 0.5 x_{i1})$, which violates the linear conditional quantile assumption.

Benchmarking against standard quantile regression (StdQR) (\texttt{quantreg}) and BQR-ALD (\texttt{brms}), we implement BSQR in \textsf{Stan} \shortcite{Carpenter2017} using the NUTS \shortcite{Hoffman2014}. We run two parallel chains of 4000 iterations (2000 warmup) to update the joint space $(\boldsymbol{\beta}, \theta)$, leveraging \textsf{Stan}'s automatic differentiation with a hybrid strategy for the normalizing constant $Z(\theta, \tau, h)$: employing asymptotic approximation for sharply peaked integrands (large $\theta$) and robust numerical quadrature (\texttt{integrate\_1d}) otherwise. The smoothing bandwidth $h$, balancing check loss fidelity and posterior stability, is selected via 5-fold cross-validation from a set of candidates scaled from Silverman's rule of thumb by factors $\{0.5, 0.75, 1.0, 1.5, 2.0\}$; fully Bayesian learning of $h$ is computationally prohibitive due to the intractability of $Z(\theta, \tau, h)$, though sensitivity is analyzed in Section~\ref{sec:robustness_analysis}.

Priors are set to $\boldsymbol{\beta} \sim \mathcal{N}(\boldsymbol{0}, 1000\boldsymbol{I})$ and $\theta \sim \text{Gamma}(0.01, 0.01)$. While Corollary~\ref{cor:validity_true} implies fixing $\theta$ for asymptotic coverage, we treat it as random to integrate out scale uncertainty, avoiding the noise of finite-sample Hessian/Jacobian estimation while empirically satisfying the generalized information equality. Point estimation accuracy is measured by MSE: $\|\hat{\boldsymbol{\beta}}(\tau) - \boldsymbol{\beta}_0(\tau)\|_2^2$; MAE: $\|\hat{\boldsymbol{\beta}}(\tau) - \boldsymbol{\beta}_0(\tau)\|_1 / d$; WMSE: $(\hat{\boldsymbol{\beta}}(\tau) - \boldsymbol{\beta}_0(\tau))^\top \boldsymbol{\Sigma}_X (\hat{\boldsymbol{\beta}}(\tau) - \boldsymbol{\beta}_0(\tau))$, with $\hat{\boldsymbol{\beta}}(\tau)$ the posterior mean. Predictive accuracy is the average check loss on the test set: $N_{\text{test}}^{-1} \sum_{i=1}^{N_{\text{test}}} \rho_{\tau}(y_i^{\text{test}} - \boldsymbol{x}_i^{\text{test}^\top} \hat{\boldsymbol{\beta}}(\tau))$. For Bayesian methods, we also report the empirical coverage and average width of 95\% credible intervals, computation time (seconds), and MCMC diagnostics, including the maximum potential scale reduction factor ($\widehat{R}_{\max}$) and the minimum bulk effective sample size (ESS$_{\min}$) across all $\boldsymbol{\beta}$ coefficients. 

\subsection{General performance in standard settings}\label{subsec:sim_general}

We first evaluate performance in standard sample sizes with $N_{\text{train}}=200$ and $N_{\text{test}}=1000$. Covariate vectors $\boldsymbol{x}_i \sim \mathcal{N}(\boldsymbol{0}, \boldsymbol{\Sigma}_X)$ have an autoregressive structure $(\boldsymbol{\Sigma}_X)_{jk} = \rho^{|j-k|}$ with $\rho=0.5$, as in \shortciteA{Fan2001}. We consider two specific designs: a sparse, high-dimensional setting ($d=20$) with $\boldsymbol{\beta}_0 = (3, 1.5, 0, 0, 2, 0, \dots, 0)^\top$, and a dense, lower-dimensional setting ($d=8$) with $\boldsymbol{\beta}_0 = (0.85, \dots, 0.85)^\top$. Simulations are run for $\tau \in \{0.25, 0.5, 0.75\}$ over $M=200$ independent replications.

Results (detailed in \hyperref[tab:sim_results_all_kernels]{Appendix D.1}) validate the proposed BSQR framework, showing consistent superiority over BQR-ALD across multiple evaluation dimensions. BSQR generally excels in estimation accuracy and prediction, achieving lower MSE, MAE, and WMSE for coefficient estimates in most scenarios across error distributions, though BQR-ALD is occasionally comparable or slightly better in specific cases. Critically, BSQR resolves BQR-ALD's predictive bias: out-of-sample check loss is reduced by 40--50\% at $\tau=0.25$ and $\tau=0.75$, matching or even slightly surpassing the frequentist StdQR benchmark, and performances are comparable at $\tau=0.5$, confirming that smoothing restores the link between model parameters and the true conditional quantiles. For inference, BSQR delivers more reliable uncertainty quantification, with credible interval coverage often closer to 0.95; BQR-ALD's narrower intervals frequently under-cover, while Uniform and Triangular kernels balance coverage and width effectively. Computationally, the smoothed posterior yields 20--40\% higher minimum bulk ESS (e.g., 2700--3200 vs. 2100--2400 for BQR-ALD), ensuring reliable estimates, and bounded-support kernels (Uniform, Triangular, Epanechnikov) maintain near-zero divergent transitions (consistently $<2$); Uniform (BSQR-U) is the fastest for large datasets\,---\,often surpassing BQR-ALD\,---\,while Triangular (BSQR-T) offers a favorable trade-off between speed, accuracy, and inference, Epanechnikov (BSQR-E) is slower with wider intervals, and Gaussian (BSQR-G) is computationally intensive and least stable. BSQR improves BQR by eliminating bias, improving computational efficiency, and favoring simple bounded-support \emph{Uniform} and \emph{Triangular kernels} for optimal trade-offs in speed, accuracy, and inference, making them the preferred choice for applied research.

\subsection{Inferential validity under extreme sparsity}\label{subsec:sim_stress}

To rigorously assess finite-sample robustness, we conduct a ``stress test'' characterized by extreme data sparsity ($n=30, p=6$) and heteroscedasticity. This regime severely challenges the asymptotic justifications of semi-parametric estimators. Diverging from the general simulation settings in Section \ref{subsec:sim_general}, we adopt a specialized data generating process to isolate the effects of sparsity. Specifically, covariates are generated from independent standard normal distributions, and errors follow a linear heteroscedastic specification $u_i \sim \mathcal{N}(0, (1 + 0.5|x_{i1}|)^2)$, with the true coefficient vector set to $\boldsymbol{\beta}_0 = (1, 1, 0, 0, 0, 0)^\top$.

Furthermore, regarding computational implementation, given the extremely small sample size, cross-validation is prone to instability. Thus, distinct from the grid-search strategy employed in Section \ref{subsec:sim_general}, here we employ a fixed bandwidth determined by Silverman's rule of thumb for all BSQR variants. This constraint ensures that performance differences are driven solely by the kernel's geometric properties rather than bandwidth variability. We benchmark BSQR against the PETEL method of \shortciteA{Tang2022}, focusing on the critical trade-off between estimation precision and inferential validity at extreme quantiles ($\tau \in \{0.05, 0.95\}$) versus the median. The detailed numerical results are provided in Table \ref{tab:stress_test_n30} in \hyperref[sec:AppendixD]{Appendix D.2}.

Table \ref{tab:stress_test_n30} reveals a stark divergence in inferential reliability. While all methods successfully converged numerically in all replications (0\% failure rate), the seemingly lower MSE of PETEL appears to be driven by an artificial reduction in variance at the cost of significant bias and under-coverage. Specifically, PETEL's squared bias is 3--8 times larger than that of BSQR with compact kernels (e.g., 91.4 vs. 27.0 at $\tau=0.05$). Consequently, its frequentist coverage drops to 83--84\%, significantly below the nominal 95\% level. This pattern suggests that the empirical likelihood surface becomes overly constrained in data-sparse regions, yielding estimates that are precise but inaccurate (``false precision'').

Conversely, BSQR offers two distinct advantages depending on the kernel choice. First, compact kernels (Uniform, Triangular, Epanechnikov) act as effective regularizers, reducing the squared bias by nearly an order of magnitude compared to PETEL (e.g., 6.26 vs. 49.7 at the median). Second, for valid uncertainty quantification, the Gaussian kernel functions as a robust inferential safety net. The infinite support of the Gaussian kernel provides non-zero gradient signals that guide the sampler even where empirical likelihoods face ``empty set'' constraints \shortcite{Koenker2005}. Despite incurring a higher MSE and Interval Score (IS) due to the larger bandwidth required to bridge sparse data points, it is the only method that successfully maintains valid coverage ($>97\%$) at the tails. This reflects a theoretically justified trade-off: in regimes of critical information scarcity, ensuring conservative validity via smoothing is prioritized over the risk of over-confident, biased inference.

\subsection{Scalability in high dimensions}\label{subsec:sim_scalability}

To probe the computational boundaries of the competing methods, we simulate a high-dimensional regime ($n=300, p=50$). This setting exacerbates the ``curse of dimensionality,'' creating a parameter space where likelihood-free or moment-based methods often struggle due to complex posterior geometry. To explicitly isolate the computational cost of high dimensionality from statistical complications (such as multicollinearity or hyperparameter tuning), we adopt a simplified simulation protocols distinct from Section \ref{subsec:sim_general}. Specifically, regressors are generated from independent standard normal distributions, errors follow $\mathcal{N}(0,1)$, and the true coefficient vector is sparse with $\boldsymbol{\beta}_0 = (1, 1, 1, 1, 1, 0, \dots, 0)^\top$. Furthermore, we employ a fixed bandwidth derived from Silverman's rule for BSQR to strictly measure the sampling efficiency (ESS/sec) of the HMC kernel itself, excluding the overhead of cross-validation. We contrast the sampling efficiency of BSQR (utilizing the Gaussian kernel and HMC) with PETEL (utilizing random walk Metropolis), summarizing the results in Table \ref{tab:high_dim_efficiency}.

\begin{table}[htbp]
\centering
\caption{High-dimensional computational scalability ($n=300$, $p=50$).}
\label{tab:high_dim_efficiency}
\begin{threeparttable}
\begin{tabular}{lccccc}
\toprule
\textbf{Method} & \textbf{Time (s)} & \textbf{Min ESS} & \textbf{ESS / Sec} & \textbf{$\widehat{R}_{\max}$} & \textbf{MSE} \\
 & (Avg.) & (Avg.) & (Avg.) & (Avg.) & ($\times 10^{-3}$) \\
\midrule
PETEL & 10.20 & 0.0 & 0.0 & ---$^{\dag}$ & 6.44 \\
BSQR & \textbf{6.15} & \textbf{795.0} & \textbf{129.3} & \textbf{1.01} & \textbf{4.18} \\
\bottomrule
\end{tabular}
\begin{tablenotes}
\small
\item \textit{Note:} Evaluation of computational efficiency and convergence diagnostics. \textbf{Min ESS} denotes the minimum bulk effective sample size across all parameters. $\widehat{R}_{\max}$ is the Gelman-Rubin convergence diagnostic. Values for \textbf{MSE} are scaled by $10^3$ for readability. Bold values indicate superior performance.
\item $^{\dag}$ For PETEL, the $\widehat{R}_{\max}$ statistic was undefined (NaN) due to zero variance within chains. This indicates a complete failure of the Random Walk Metropolis sampler to mix, resulting in the chain remaining stuck at the initialization point (the frequentist estimator) throughout the simulation.
\end{tablenotes}
\end{threeparttable}
\end{table}

Table \ref{tab:high_dim_efficiency} documents a catastrophic breakdown in the efficacy of the projected empirical likelihood approach. The PETEL algorithm exhibits a phenomenon of complete mixing collapse: despite successful numerical execution, the ESS is indistinguishable from zero, and the convergence diagnostics ($\widehat{R}_{\max}$) are undefined due to vanishing chain variance. This stagnation confirms that the Random Walk Metropolis sampler is ill-suited for the complex, constrained geometry of the empirical likelihood surface in 50 dimensions. Lacking gradient guidance, the blind proposal mechanism yields a near-zero acceptance probability, leaving the chain trapped at its starting value. Crucially, PETEL incurs a higher computational cost (10.20s vs. 6.15s) merely to reject every proposal. Consequently, the seemingly comparable MSE of PETEL ($6.44\times 10^{-3}$) is solely an artifact of the informative initialization provided by the frequentist estimator, rather than evidence of successful Bayesian learning.

In sharp contrast, the BSQR framework demonstrates HMC dominance in high-dimensional exploration. By smoothing the check loss, BSQR provides the necessary gradient information that allows the Hamiltonian Monte Carlo sampler to traverse the parameter space efficiently, overcoming the geometric bottlenecks identified by \shortciteA{Betancourt2017}. The sampler achieves near-perfect mixing ($\text{ESS} \approx 795$ out of 1000 iterations) and rapid convergence ($\widehat{R}_{\max} = 1.01$). This results in an effectively infinite efficiency gap: while BSQR generates approximately 129 independent posterior samples per second, PETEL generates none. These findings empirically establish that the smoothed likelihood formulation is not merely a theoretical convenience but a computational prerequisite for scalable Bayesian inference in high-dimensional quantile regression.

\section{Empirical analysis: Asymmetric systemic risk exposure in the post-COVID era}
\label{sec:empirical_application_final_para}

We apply the BSQR framework to assess the asymmetric dynamics of systemic risk for a globally systemically important financial institution (G-SIFI). Quantile-based dependence measures, such as the quantilogram \shortcite{Linton2007}, have highlighted the critical importance of tail dependency in financial econometrics. Building on this perspective, our analysis models the asymmetric link between daily stock returns of JPMorgan Chase \& Co. (JPM), the largest U.S. bank by assets, and the S\&P 500 Index (SP500),\footnote{Data for JPM and the S\&P 500 are from \url{https://finance.yahoo.com/quote/JPM/} and \url{https://finance.yahoo.com/quote/^GSPC/}, respectively.} using daily log-returns from January 1, 2017, to January 1, 2025. Prices are converted to continuously compounded returns ($r_t = \ln(P_t/P_{t-1})$), yielding a rich time series whose summary statistics appear in Table \ref{tab:empirical_descriptive_stats_final_para}. The pronounced leptokurtosis in both series strongly motivates a quantile-based approach.

\begin{table}[htbp]
\caption{Descriptive statistics for daily log-returns.}
\label{tab:empirical_descriptive_stats_final_para}
\centering
\begin{tabular}[t]{lcc}
\toprule
\textbf{Statistic} & \textbf{JPM} & \textbf{SP500}\\
\midrule
Observations & 2011 & 2011\\
Mean & 0.0006 & 0.0005\\
Std. Dev. & 0.0178 & 0.0118\\
Skewness & -0.0240 & -0.8614\\
Kurtosis & 14.33 & 16.30\\
\bottomrule
\end{tabular}
\end{table}

We specify a dynamic capital asset pricing model for quantile level $\tau$. Let $r_{jpm, t}$ and $r_{m, t}$ denote the log-returns of JPM and the SP500 at time $t$, respectively. The conditional quantile function is specified as:
\begin{align*}
    Q_{r_{jpm, t}}(\tau \mid r_{m, t}) = \alpha(\tau) + \beta(\tau) r_{m, t},
\end{align*}
where $\beta(\tau)$ measures systemic risk exposure. To capture asymmetry, we examine the downside beta ($\beta(0.05)$), quantifying exposure during severe downturns, and the upside beta ($\beta(0.95)$), measuring participation in market rallies. Using a rolling-window estimation with a one-year window (252 trading days) advanced monthly (21 days), we compare BSQR (Uniform and Triangular kernels) with the BQR-ALD benchmark across inferential stability, predictive accuracy, and sampler efficiency.

\subsection{Inferential stability and economic insights from asymmetric betas}

Methodologically, Figure \ref{fig:rolling_betas_combined} demonstrates that BSQR provides substantially more stable and interpretable parameter estimates. In both panels, the BQR-ALD benchmark produces visibly wider and more erratic credible intervals (light gray shaded area) plagued by high-frequency noise that can lead to over-interpretation of statistically insignificant fluctuations. In contrast, BSQR’s inherent kernel smoothing attenuates this noise, yielding tighter credible intervals (orange and blue shaded areas) and smoother posterior mean paths. This enables a more reliable quantification of uncertainty, clearly distinguishing genuine shifts in the risk profile from mere sampling variation.

Economically, BSQR’s inferential clarity reveals a multi-phase and profoundly asymmetric response of JPM’s systemic risk to the COVID-19 shock. After an initial, transient symmetric spike in both betas during the early-2020 market panic, a more dominant and divergent dynamic quickly emerged. The downside beta ($\beta(0.05)$), previously stable around a pre-crisis baseline of roughly 0.97, decoupled from market downturns, plummeting to a trough of 0.53 by mid-2021 in a sustained ``defensive decoupling'' effect before settling into a new, lower regime averaging 0.92. This distinct regime shift highlights the model's capacity to capture rapid structural changes in systemic risk exposure, consistent with the structural break phenomena analyzed in quantile change-point literature \shortcite{Lee2018}. In stark contrast, the upside beta ($\beta(0.95)$), originating from a similar 0.96 baseline, surged to a peak of 1.38 during the stimulus-fueled recovery, followed by a structural re-pricing that established a new, markedly lower equilibrium averaging 0.75. The superior stability of BSQR is what makes this nuanced dynamic sequence\,---\,a transient symmetric spike followed by a powerful asymmetric divergence\,---\,unambiguously clear.

\begin{figure}[htbp]
    \centering
    
    \begin{subfigure}{\textwidth}
        \centering
        \includegraphics[width=0.8\textwidth]{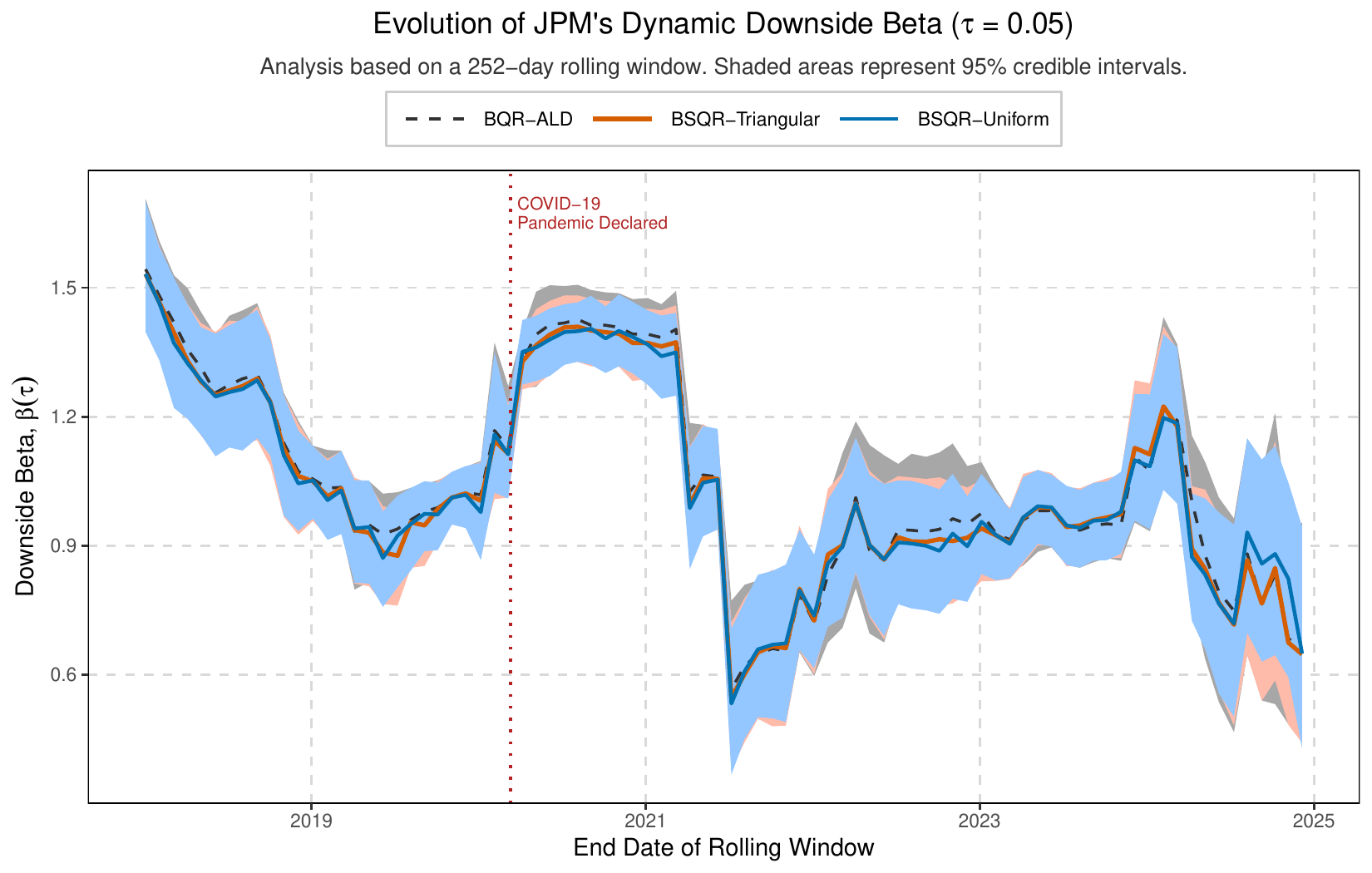} 
        \caption{Evolution of the dynamic downside beta ($\beta(0.05)$) for JPM.}
        \label{fig:rolling_beta_0.05}
    \end{subfigure}
    
    \vspace{1cm}
    
    \begin{subfigure}{\textwidth}
        \centering
        \includegraphics[width=0.8\textwidth]{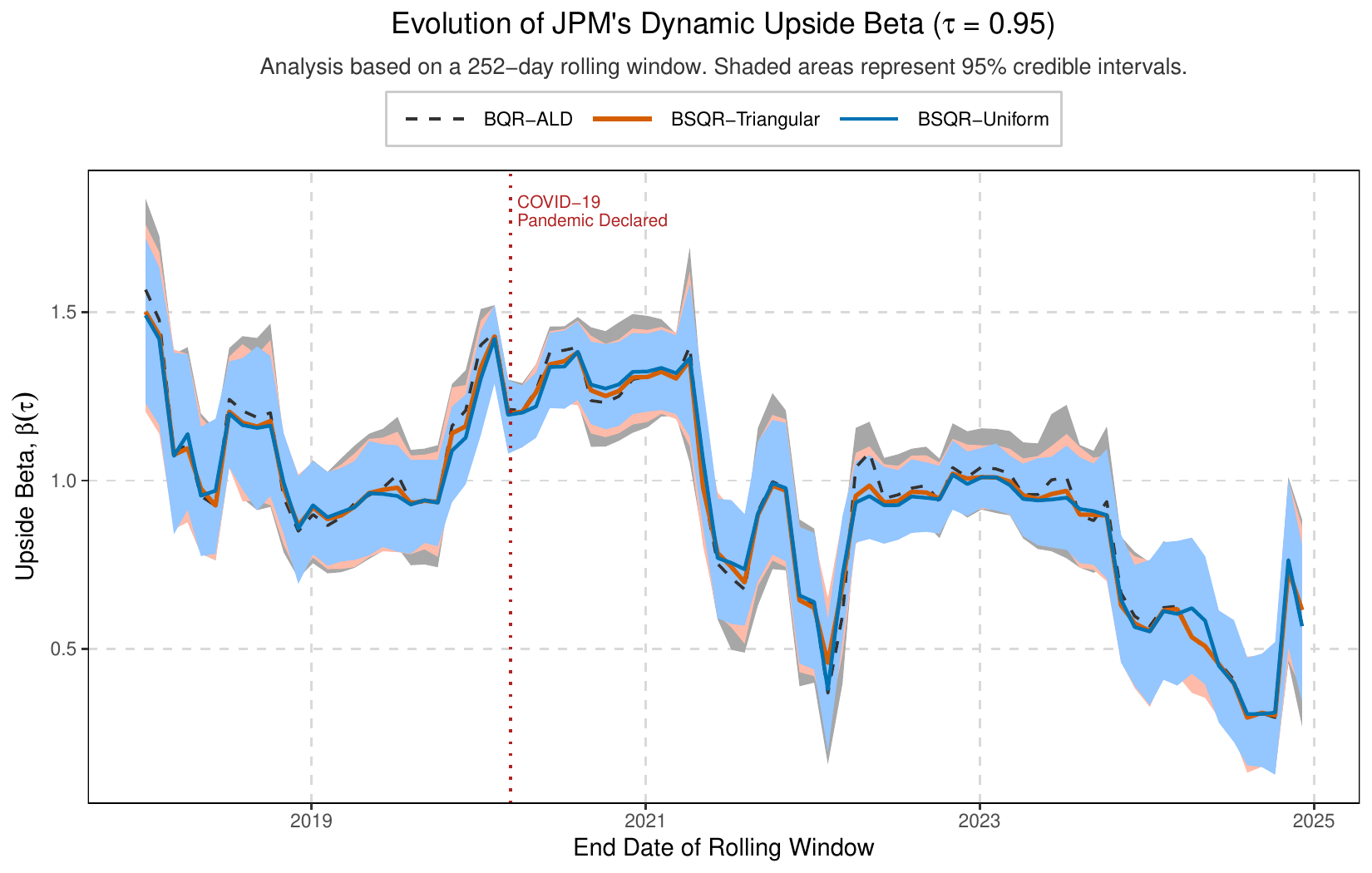} 
        \caption{Evolution of the dynamic upside beta ($\beta(0.95)$) for JPM.}
        \label{fig:rolling_beta_0.95}
    \end{subfigure}
    
    \captionsetup{justification=justified, font=small}
    \caption{Evolution of dynamic downside and upside systemic risk betas for JPM.}
    \label{fig:rolling_betas_combined} 
    
    \par\medskip
    
    \begin{minipage}{\textwidth}
        \captionsetup{width=0.9\textwidth, justification=justified, font=small}
        \textit{Note:} Comparison of the benchmark (BQR-ALD) with BSQR methods (uniform and triangular kernels). The figures plot the posterior mean and 95\% credible intervals for the (a) downside beta ($\tau=0.05$) and (b) upside beta ($\tau=0.95$), estimated using a 252-day rolling window. The BQR-ALD (benchmark) estimates exhibit wider and more volatile credible intervals compared to the smoother and tighter intervals from the BSQR methods.
    \end{minipage}
\end{figure}

\subsection{Predictive accuracy and sampler efficiency}

Beyond inferential quality, we assess BSQR's out-of-sample predictive accuracy and computational efficiency. In terms of predictive accuracy, measured by the average one-day-ahead check loss, Table \ref{tab:forecast_results_asymmetric} shows that BSQR is highly competitive. For the downside quantile level ($\tau=0.05$), the check loss of BSQR-Uniform (0.002998) and BSQR-Triangular (0.003031) is economically indistinguishable from that of BQR-ALD (0.003027). Notably, for the upside quantile level ($\tau=0.95$), both BSQR methods (0.001338 and 0.001349) outperform the BQR-ALD benchmark (0.001363). This is a powerful result: BSQR’s substantial gains in inferential stability are achieved without sacrificing\,---\,and in the upside case, even enhancing\,---\,predictive accuracy.

\begin{table}[htbp]
\centering
\caption{Out-of-sample forecasting performance (check loss).}
\label{tab:forecast_results_asymmetric}
\begin{threeparttable}
\begin{tabular}{lcc}
\toprule
\textbf{Methods} & \textbf{Downside ($\tau=0.05$)} & \textbf{Upside ($\tau=0.95$)}\\
\midrule
BQR-ALD & 0.003027 & 0.001363\\
BSQR-Uniform & \textbf{0.002998} & \textbf{0.001338}\\
BSQR-Triangular & 0.003031 & 0.001349\\
\bottomrule
\end{tabular}
\begin{tablenotes}
\small
\item \textit{Note:} Average one-day-ahead check loss. Lower values indicate better predictive accuracy.
\end{tablenotes}
\end{threeparttable}
\end{table}

Computationally, BSQR demonstrates superior sampler efficiency, with the Uniform kernel emerging as the dominant specification (Table \ref{tab:mcmc_efficiency_asymmetric_detailed}). The average minimum bulk effective sample size (Avg. Min. ESS) for BSQR-Uniform (2554.72) is approximately 83\% higher than that of BQR-ALD (1399.66). Crucially, this improvement is achieved with a lower computational time (0.69s vs. 0.77s), indicating that the combination of the Uniform kernel and HMC is strictly more efficient than the Gibbs sampling benchmark. While the BSQR-Triangular method incurs a higher computational cost (4.23s) due to the complexity of its smooth approximation, it similarly delivers a 79\% increase in ESS, offering a robust alternative. Overall, BSQR-Uniform provides the optimal balance, delivering higher predictive accuracy and nearly double the sampling efficiency in less time.

\begin{table}[htbp]
\centering
\caption{MCMC sampler efficiency in rolling-window analysis.}
\label{tab:mcmc_efficiency_asymmetric_detailed}
\begin{threeparttable}
\begin{tabular}{lcccc}
\toprule
\textbf{Methods} & \textbf{Avg. Min. ESS} & \textbf{Avg. Time (s)} & \textbf{Avg. Divergences} & \textbf{Avg. Selected $h$}\\
\midrule
BQR-ALD & 1399.66 & 0.77 & 0.00 & ---\\
BSQR-Uniform & \textbf{2554.72} & \textbf{0.69} & 0.00 & 0.60\\
BSQR-Triangular & 2506.05 & 4.23 & 0.01 & 0.66\\
\bottomrule
\end{tabular}
\begin{tablenotes}
\small
\item \textit{Note:} Efficiency metrics averaged over rolling windows. \textbf{Avg. Min. ESS} denotes the minimum bulk effective sample size across parameters. \textbf{Avg. Divergences} refers to the number of divergent transitions in NUTS.
\end{tablenotes}
\end{threeparttable}
\end{table}

\FloatBarrier

\subsection{Robustness to bandwidth selection}\label{sec:robustness_analysis}

To examine the influence of smoothing bandwidth (Section~\ref{sec:simulation}), we conducted a sensitivity analysis of the downside beta ($\beta(0.05)$) for a representative 252-day rolling window ending July 7, 2020, during significant market turmoil. Using the BSQR-Uniform method, we compared three bandwidths: the cross-validated choice ($h_{CV} \approx 1.08$), a smaller undersmoothed value ($h \approx 0.54$), and a larger oversmoothed one ($h \approx 2.16$). As shown in Figure~\ref{fig:sensitivity_analysis}, the posterior location of the systemic risk parameter is sensitive to $h$, with the mean shifting from about 1.36 to 1.42 as smoothing increases, underscoring the need for a principled, data-driven bandwidth choice. Yet, despite this point-estimate sensitivity, the economic conclusion holds: all posterior distributions concentrate on values well above 1, confirming high systemic risk in the crisis period regardless of bandwidth. This analysis both quantifies model sensitivity and highlights the value of cross-validation as an objective, replicable criterion for balancing fit and smoothing.

\begin{figure}[htbp]
    \centering
    \includegraphics[width=0.7\textwidth]{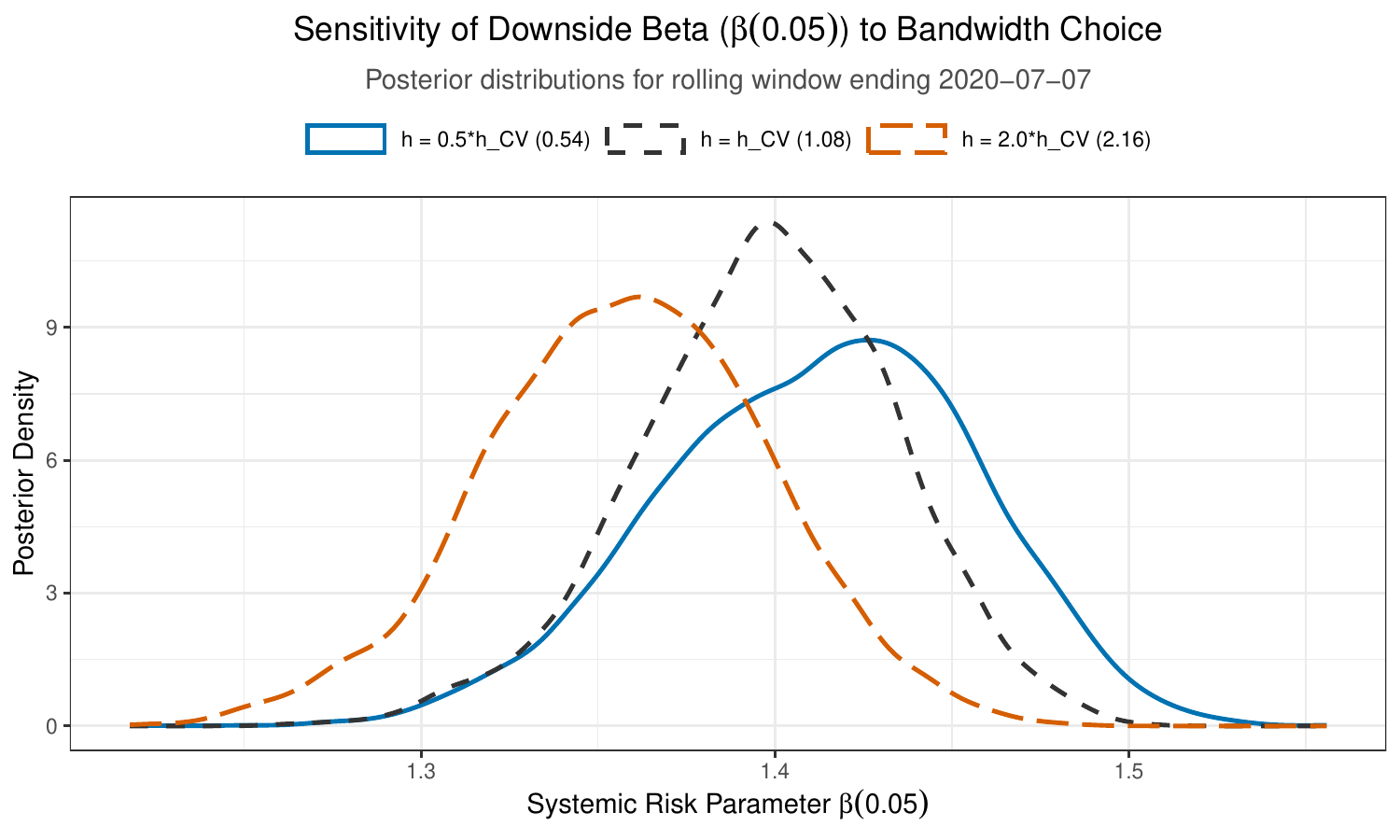}
\captionsetup{width=0.9\textwidth, justification=justified, font=small}
\caption{Sensitivity analysis of the downside beta ($\beta(0.05)$) to bandwidth choice. Posterior distributions for JPM are estimated on the 252-day window ending July 7, 2020. The figure compares the cross-validated bandwidth ($h_{\text{CV}}$, black dashed) with undersmoothed ($0.5h_{\text{CV}}$, orange long-dashed) and oversmoothed ($2.0h_{\text{CV}}$, blue solid) alternatives. While increasing bandwidth systematically shifts the posterior location to the right, the qualitative conclusion of high systemic risk ($\beta > 1$) remains robust.}
    \label{fig:sensitivity_analysis}
\end{figure}

Our empirical application shows BSQR to be a substantial methodological advance: it delivers more reliable and stable parameter estimates while maintaining competitive predictive performance. This enhanced inferential fidelity supports a more credible, nuanced understanding of economic phenomena, offering a powerful tool for researchers and practitioners prioritizing robustness.

\FloatBarrier
\section{Discussion}\label{sec:discussion}

This paper introduces Bayesian smoothed quantile regression (BSQR), a principled framework that reconciles the long-standing tension between computational efficiency and inferential validity in Bayesian quantile analysis. By constructing a fully differentiable pseudo-likelihood via kernel smoothing\,---\,a classical device for managing the bias--variance trade-off \shortcite{Gozalo2000}\,---\,BSQR delivers decision-theoretically aligned inference for conditional quantiles, rectifies the finite-sample bias documented for BQR-ALD \shortcite{Gneiting2011}, and enables efficient gradient-based MCMC such as HMC. Our theoretical analysis further establishes posterior consistency, Bernstein--von Mises (BvM) limits, and global posterior propriety, and it derives a generalized calibration condition for the scale parameter that guarantees valid frequentist coverage.

A distinct methodological advantage of BSQR over empirical-likelihood-based frameworks \shortcite{Tang2022} lies in its modularity and tractability. Whereas empirical likelihood methods often face severe geometric and computational hurdles in the presence of latent variables, many moment conditions, or high-dimensional constraints, BSQR provides an explicit, globally supported pseudo-likelihood. This ``plug-and-play'' structure facilitates seamless embedding into hierarchical and latent-variable architectures, while preserving a clear decision-theoretic interpretation. In particular, the differentiable formulation is well suited for extensions to endogeneity via control-function approaches \shortcite{Lee2007} and to semi- and nonparametric specifications (e.g., splines or Gaussian processes), directly addressing the limitations of purely linear quantile models.

Future research avenues are rich. Beyond a fully Bayesian treatment of bandwidth selection \shortcite{Silverman1986}, the BSQR framework naturally extends to multivariate quantiles \shortcite{Hallin2010} and to dynamic models with time-varying parameters \shortcite{Gerlach2011}, where smoothing-based pseudo-likelihoods can again regularize non-smooth objectives without sacrificing interpretability. Moreover, the coherent posterior quantification delivered by BSQR is directly applicable to probabilistic assessments of stochastic dominance, a central theme in modern econometrics \shortcite{Linton2005}, allowing researchers to test dominance orders without resorting to complex resampling schemes. In this sense, BSQR is not merely a computational refinement but a structured way of combining frequentist smoothing precision with Bayesian probabilistic richness, opening a flexible platform for future developments in distributional and risk-oriented econometrics.

\phantomsection
\section*{Declaration of competing interest}
\sectionbookmark{Declaration of competing interest}
The authors declare that they have no conflict of interest.

\phantomsection
\section*{Author contributions}
\sectionbookmark{Author contributions}

Bingqi Liu conceived the study, developed the BSQR framework, derived its theoretical properties, wrote the software, and conducted all numerical experiments. He also wrote the initial draft of the manuscript. Kangqiang Li provided assistance with the proofs of asymptotic posterior consistency and kernel effects on posterior concentration. Tianxiao Pang supervised the research and provided critical feedback on the manuscript. All authors reviewed and approved the final manuscript.

\phantomsection
\section*{Data and code availability}
\sectionbookmark{Data and code availability}
The source code and data required to replicate all numerical results in this paper are publicly available on GitHub at the following repository: \url{https://github.com/BeauquinLau/BSQR}.

\phantomsection
\begin{singlespace}
\bibliography{Bayesian_Smoothed_Quantile_Regression}
\end{singlespace}

\newpage
\appendix
\phantomsection
\addcontentsline{toc}{section}{Appendix A: Proofs}
\section*{Appendix A: Proofs}
\label{sec:AppendixA}
\setcounter{equation}{0}
\renewcommand{\theequation}{A.\arabic{equation}}
\setcounter{subsection}{0}
\renewcommand{\thesubsection}{A.\arabic{subsection}}
\renewcommand{\theHsubsection}{A.\arabic{subsection}}

\subsection{Proofs for posterior consistency}
\label{app:proofs_consistency}

Before proving the main consistency theorem, we establish the following technical lemma regarding the bias of the smoothed score function.

\begin{techlemma}[Bias of the smoothed score]\label{lem:score_bias}
    Under \textbf{Assumption} \ref{as:consistency_conditions} (specifically \textbf{A2} and \textbf{A4}), the expected value of the smoothed score function evaluated at the true parameter $\boldsymbol{\beta}_0(\tau)$ satisfies:
    \begin{align}
        \mathbb{E}[\Psi_h(\varepsilon_{0i}; \tau)] = -\frac{1}{2} f'_{\varepsilon}(0) \mu_2(K) h^2 + o(h^2) = O(h^2),
    \end{align}
    where $\mu_2(K) = \int_{-\infty}^{\infty} u^2 K(u) \,\mathrm{d}u$. This result implies that the smoothing introduces a deterministic bias of order $O(h^2)$.
\end{techlemma}

\phantomsection
\begin{proof}[Proof of Lemma~\ref{lem:score_bias}]
    Recall the definition $\Psi_h(e; \tau) = \int_{-\infty}^{\infty} (\tau - \mathbb{I}(e < v)) K_h(v) \,\mathrm{d}v$. The expectation with respect to the error $\varepsilon_{0i} \sim F_\varepsilon$ is:
    \begin{align*}
        \mathbb{E}[\Psi_h(\varepsilon_{0i}; \tau)] &= \int_{-\infty}^{\infty} \mathbb{E}[\tau - \mathbb{I}(\varepsilon_{0i} < v)] K_h(v) \,\mathrm{d}v = \int_{-\infty}^{\infty} (\tau - F_\varepsilon(v)) K_h(v) \,\mathrm{d}v.
    \end{align*}
    Using the change of variables $v = uh$ and noting that $F_\varepsilon(0) = \tau$, we perform a second-order Taylor expansion of $F_\varepsilon(uh)$ around 0 since $F_\varepsilon(\cdot)$ is twice differentiable around 0 (\textbf{Assumption A2a}):
    \begin{align*}
        F_\varepsilon(uh) = \tau + (uh) f_\varepsilon(0) + \frac{(uh)^2}{2} f'_\varepsilon(0) + o((uh)^2).
    \end{align*}
    Substituting this expansion back into the integral yields:
    \begin{align*}
        \mathbb{E}[\Psi_h(\varepsilon_{0i}; \tau)] &= \int_{-\infty}^{\infty} \left( \tau - \left[ \tau + uh f_\varepsilon(0) + \frac{u^2 h^2}{2} f'_\varepsilon(0) + o(u^2 h^2) \right] \right) K(u) \,\mathrm{d}u \\
        &= -h f_\varepsilon(0) \underbrace{\int_{-\infty}^{\infty} u K(u) \,\mathrm{d}u}_{=0} - \frac{h^2}{2} f'_\varepsilon(0) \underbrace{\int_{-\infty}^{\infty} u^2 K(u) \,\mathrm{d}u}_{=\mu_2(K)} + \int_{-\infty}^{\infty} o(u^2 h^2) K(u) \,\mathrm{d}u.
    \end{align*}
    The first-order term vanishes due to the symmetry of the kernel $K(\cdot)$ (\textbf{Assumption A4}). Since $\int_{-\infty}^{\infty} u^2 K(u) \mathrm{d}u < \infty$, the remainder term is of order $o(h^2)$. Thus, the dominant term is $O(h^2)$.
\end{proof}

\phantomsection
\begin{proof}[Proof of Theorem~\ref{thm:consistency}]
\label{pf:consistency}
The proof follows a standard strategy for demonstrating Bayesian posterior consistency (e.g., \shortciteNP{Chernozhukov2003, Ghosal2000}). The approach consists of two key stages: first, we establish the consistency of the frequentist M-estimator, $\hat{\boldsymbol{\beta}}_h(\tau)$, which minimizes the smoothed objective function $\widehat{R}_h(\boldsymbol{b}; \tau, h)$. Second, we leverage this result to show that the full Bayesian posterior concentrates around the true parameter value.

The M-estimator $\hat{\boldsymbol{\beta}}_h(\tau)$ is the minimizer of the sample objective function Eq.~\eqref{sqr_obj}. To prove its consistency, we analyze its population analogue, $R_h(\boldsymbol{b}; \tau, h) \coloneqq \mathbb{E}[\widehat{R}_h(\boldsymbol{b}; \tau, h)]$, and show that it is uniquely minimized at the true parameter $\boldsymbol{\beta}_0(\tau)$ as $n \to \infty$. First, we analyze the gradient of the population objective function with respect to $\boldsymbol{b}$, which is $\nabla_{\boldsymbol{b}} R_h(\boldsymbol{b}; \tau, h) = \mathbb{E}[-\boldsymbol{x}_i \Psi_h(e_i(\boldsymbol{b}); \tau)]$ by using the definition Eq.~\eqref{eq:Psi_h_definition}. Evaluating this at the true parameter $\boldsymbol{b} = \boldsymbol{\beta}_0(\tau)$ yields:
\begin{align*}
    \nabla_{\boldsymbol{\beta}_0(\tau)} R_h(\boldsymbol{\beta}_0(\tau); \tau, h) = \mathbb{E}[-\boldsymbol{x}_i \Psi_h(\varepsilon_{0i}; \tau)] = -\mathbb{E}_{\boldsymbol{x}}[\boldsymbol{x}_i] \cdot \mathbb{E}_{\varepsilon}[\Psi_h(\varepsilon_{0i}; \tau)],
\end{align*}
where the final equality uses the independence of $\boldsymbol{x}_i$ and $\varepsilon_{0i}$. The central task is to evaluate the expectation of the score function, $\mathbb{E}_{\varepsilon}[\Psi_h(\varepsilon_{0i}; \tau)]$. By Lemma~\ref{lem:score_bias}, we directly have:
\begin{align*}
    \mathbb{E}_{\varepsilon}[\Psi_h(\varepsilon_{0i}; \tau)] = O(h^2).
\end{align*}
Thus, $\nabla_{\boldsymbol{\beta}_0(\tau)} R_h(\boldsymbol{\beta}_0(\tau); \tau, h) = O(h^2)$, which converges to $\boldsymbol{0}$ as $n \to \infty$ due to $h \to 0$.

Next, we verify the second-order condition by examining the Hessian matrix of $R_h(\boldsymbol{b}; \tau, h)$ at $\boldsymbol{\beta}_0(\tau)$, denoted as $\boldsymbol{H}_h(\boldsymbol{\beta}_0(\tau))$. Note that
\begin{align*}
    \boldsymbol{H}_h(\boldsymbol{b}) \coloneqq \nabla^2_{\boldsymbol{b}} R_h(\boldsymbol{b};\tau,h) = \mathbb{E}[ \boldsymbol{x}_i \boldsymbol{x}_i^\top \cdot \Psi'_{h}(e_i(\boldsymbol{b}); \tau) ].
\end{align*}
At $\boldsymbol{b}=\boldsymbol{\beta}_0(\tau)$, we have $\Psi'_{h}(\varepsilon_{0i};\tau) = K_{h}(\varepsilon_{0i})$. As $h\to 0$, the expectation $\mathbb{E}[\Psi'_{h}(\varepsilon_{0i};\tau)] = \int_{-\infty}^{\infty} K_{h}(u) f_\varepsilon(u)\,\mathrm{d}u$ converges to $f_\varepsilon(0)$. Consequently, the Hessian matrix:
\begin{align*}
    \boldsymbol{H}_h(\boldsymbol{\beta}_0(\tau)) \to \mathbb{E}[\boldsymbol{x}_i \boldsymbol{x}_i^\top f_\varepsilon(0)] = f_\varepsilon(0) \mathbb{E}[\boldsymbol{x}_i \boldsymbol{x}_i^\top] = f_\varepsilon(0) \Sigma_X.
\end{align*}
This limiting matrix is positive definite, since $f_\varepsilon(0) > 0$ by \textbf{Assumption (A2a)} and $\Sigma_X$ is positive definite by \textbf{Assumption (A3)}. Given the convexity of the smoothed objective function, this ensures that $\boldsymbol{\beta}_0(\tau)$ is the unique global minimizer.

Since the population objective function $R_h(\boldsymbol{b};\tau,h)$ has a gradient that converges to zero and a positive definite Hessian at $\boldsymbol{\beta}_0(\tau)$, $\boldsymbol{\beta}_0(\tau)$ is the unique minimizer of $R_h(\boldsymbol{b};\tau,h)$ in the limit. To establish that $\hat{\boldsymbol{\beta}}_h(\tau) \xrightarrow{P} \boldsymbol{\beta}_0(\tau)$, M-estimation theory requires showing that $\sup_{\boldsymbol{b} \in \mathcal{B}} |\widehat{R}_h(\boldsymbol{b};\tau,h) - R_h(\boldsymbol{b};\tau,h)| \xrightarrow{P} 0$. This uniform convergence is guaranteed by the uniform law of large numbers (ULLN), for which we must verify the existence of a dominating function for $L_{h}(e_i(\boldsymbol{b}); \tau)$ with a finite expectation. The argument of the loss function is $e_i(\boldsymbol{b}) = y_i - \boldsymbol{x}_i^\top\boldsymbol{b} = \varepsilon_{0i} - \boldsymbol{x}_i^\top(\boldsymbol{b} - \boldsymbol{\beta}_0(\tau))$. Since the magnitude of $L_h(\cdot; \tau)$ is bounded by a multiple of its argument, the task reduces to uniformly bounding $|e_i(\boldsymbol{b})|$ over the compact set $\mathcal{B}$:
\begin{align*}
    \sup_{\boldsymbol{b} \in \mathcal{B}} |e_i(\boldsymbol{b})| \le |\varepsilon_{0i}| + \sup_{\boldsymbol{b} \in \mathcal{B}} |\boldsymbol{x}_i^\top(\boldsymbol{b} - \boldsymbol{\beta}_0(\tau))| \le |\varepsilon_{0i}| + \|\boldsymbol{x}_i\| \sup_{\boldsymbol{b} \in \mathcal{B}} \|\boldsymbol{b} - \boldsymbol{\beta}_0(\tau)\|.
\end{align*}

The inequality above shows that this bound depends on two key terms, both of which are controlled by our assumptions. First, the term $\sup_{\boldsymbol{b} \in \mathcal{B}} \|\boldsymbol{b} - \boldsymbol{\beta}_0(\tau)\|$ is bounded by the diameter of the parameter space, $D_{\mathcal{B}}$, because $\mathcal{B}$ is assumed to be compact (\textbf{Assumption A1}). Second, the covariate norm $\|\boldsymbol{x}_i\|$ is uniformly bounded by a constant $M_X$ (\textbf{Assumption A3}). These assumptions jointly allow us to construct a dominating function $D_i = C \cdot (|\varepsilon_{0i}| + M_X D_{\mathcal{B}})$ for some constant $C>0$. The final condition is to ensure its expectation, $\mathbb{E}[D_i]$, is finite. This is guaranteed by \textbf{Assumption (A2b)}, which explicitly requires that the error term $\varepsilon_{0i}$ has a finite first moment. With all conditions for the ULLN satisfied, the uniform convergence of the objective function is established. The consistency of the M-estimator, $\hat{\boldsymbol{\beta}}_h(\tau) \xrightarrow{P} \boldsymbol{\beta}_0(\tau)$, then follows directly from standard large-sample theory \shortcite{Newey1994}.

With these frequentist properties established, the proof of posterior consistency follows a standard argument. The Bayesian analysis employs a quasi-likelihood function constructed from the M-estimation objective: $L(\boldsymbol{\beta}, \theta \mid \boldsymbol{y}, \boldsymbol{\mathcal{X}}) \propto \exp(-n\theta \widehat{R}_h(\boldsymbol{\beta}; \tau, h))$. The properties we have just demonstrated are precisely the conditions that ensure this quasi-likelihood concentrates its mass in a shrinking neighborhood of the true parameter. Combined with a prior distribution $\pi(\boldsymbol{\beta})$ that assigns positive mass to any such neighborhood (\textbf{Assumption A5}), established theorems in Bayesian asymptotics (e.g., Theorem 1 in \shortciteNP{Chernozhukov2003}; see also Theorem 2.1 in \shortciteNP{Ghosal2000}) confirm that the posterior distribution inherits this concentration property. Consequently, the posterior distribution $\pi(\boldsymbol{\beta} \mid \boldsymbol{y}, \boldsymbol{\mathcal{X}}, \theta)$ contracts to a point mass at the true parameter $\boldsymbol{\beta}_0(\tau)$.
\end{proof}

\subsection{Proofs for asymptotic normality}
\label{app:proofs_normality}

\phantomsection
\begin{proof}[Proof of Proposition~\ref{prop:existence}]\label{pf:existence}
    The BSQR posterior density is proportional to the product of the prior and the quasi-likelihood:
    \[
    \pi(\boldsymbol{\beta} \mid \boldsymbol{y}, \boldsymbol{\mathcal{X}}) \propto \pi(\boldsymbol{\beta}) \exp\left( -\theta \sum_{i=1}^n L_h(y_i - \boldsymbol{x}_i^\top \boldsymbol{\beta}; \tau) \right).
    \]
    By definition, the smoothed loss function is a convolution: $L_h(u; \tau) = \int_{-\infty}^{\infty} \rho_\tau(u - v) K_h(v) \,\mathrm{d}v$. For any valid kernel PDF $K(\cdot)$ (which integrates to 1) and bandwidth $h>0$, this integral $L_h(u; \tau)$ is finite for all $u \in \mathbb{R}$. Since the exponential function $f(z) = \exp(-z)$ maps any finite real number to a strictly positive value, the term $\exp\left( -\theta \sum_{i=1}^n L_h(\cdot;\tau) \right)$ is strictly positive for all $\boldsymbol{\beta} \in \mathbb{R}^d$. Consequently, as long as the prior $\pi(\boldsymbol{\beta})$ is proper and positive over $\mathbb{R}^d$, the posterior density is strictly positive and well-defined everywhere. This contrasts with standard empirical likelihood approaches where the likelihood can be zero (or undefined) if the moment conditions cannot be satisfied within the convex hull of the data.
\end{proof}

\phantomsection
\begin{proof}[Proof of Proposition~\ref{prop:rate}]\label{pf:rate}
    
    Let $S_n(\boldsymbol{\beta}) \coloneqq \nabla_{\boldsymbol{\beta}} \widehat{R}_h(\boldsymbol{\beta}; \tau, h) = -\sum_{i=1}^n \Psi_h(y_i - \boldsymbol{x}_i^\top \boldsymbol{\beta}; \tau) \boldsymbol{x}_i$ be the score function of the smoothed objective. The estimator $\hat{\boldsymbol{\beta}}_h(\tau)$ satisfies the first-order condition $S_n(\hat{\boldsymbol{\beta}}_h(\tau)) = \boldsymbol{0}$.
    
    We verify the rate of convergence using a Taylor expansion of $S_n(\hat{\boldsymbol{\beta}}_h(\tau))$ around the true parameter $\boldsymbol{\beta}_0(\tau)$. Since $L_h(\cdot; \tau)$ is twice continuously differentiable (guaranteed by the continuity of $K(\cdot)$ in \textbf{Assumption A4}), we have:
    \begin{align*}
        \boldsymbol{0} = S_n(\hat{\boldsymbol{\beta}}_h(\tau)) = S_n(\boldsymbol{\beta}_0(\tau)) + \nabla_{\boldsymbol{\beta}} S_n(\tilde{\boldsymbol{\beta}}(\tau)) (\hat{\boldsymbol{\beta}}_h(\tau) - \boldsymbol{\beta}_0(\tau)),
    \end{align*}
    where $\tilde{\boldsymbol{\beta}}(\tau)$ lies between $\hat{\boldsymbol{\beta}}_h(\tau)$ and $\boldsymbol{\beta}_0(\tau)$. Rearranging the terms yields:
    \begin{align}\label{eq:rate_expansion}
        \hat{\boldsymbol{\beta}}_h(\tau) - \boldsymbol{\beta}_0(\tau) = - \left[ \nabla_{\boldsymbol{\beta}} S_n(\tilde{\boldsymbol{\beta}}(\tau)) \right]^{-1} S_n(\boldsymbol{\beta}_0(\tau)).
    \end{align}
Using the chain rule, the Hessian matrix is derived as follows:
    \begin{align*}
        \nabla_{\boldsymbol{\beta}} S_n(\boldsymbol{\beta}) &= \nabla_{\boldsymbol{\beta}} \left( -\sum_{i=1}^n \Psi_h(e_i(\boldsymbol{\beta}); \tau) \boldsymbol{x}_i \right) = -\sum_{i=1}^n \boldsymbol{x}_i \cdot \nabla_{\boldsymbol{\beta}} \left[ \Psi_h(e_i(\boldsymbol{\beta}); \tau) \right]^\top \\
        &= -\sum_{i=1}^n \boldsymbol{x}_i \left[ \Psi'_h(e_i(\boldsymbol{\beta}); \tau) \cdot \nabla_{\boldsymbol{\beta}} e_i(\boldsymbol{\beta})^\top \right] = -\sum_{i=1}^n \boldsymbol{x}_i \left[ \frac{1}{h} K\left(\frac{e_i(\boldsymbol{\beta})}{h}\right) \cdot (-\boldsymbol{x}_i^\top) \right] \\
        &= \sum_{i=1}^n \frac{1}{h} K\left(\frac{e_i(\boldsymbol{\beta})}{h}\right) \boldsymbol{x}_i \boldsymbol{x}_i^\top.
    \end{align*}
    Under the consistency result (Theorem \ref{thm:consistency}), we have $\tilde{\boldsymbol{\beta}}(\tau) \xrightarrow{P} \boldsymbol{\beta}_0(\tau)$. By the law of large numbers and the continuous mapping theorem, the normalized Hessian evaluated at the intermediate point converges to the population Hessian at the truth:
    \begin{align*}
        \frac{1}{n} \nabla_{\boldsymbol{\beta}} S_n(\tilde{\boldsymbol{\beta}}(\tau)) &= \frac{1}{n} \nabla_{\boldsymbol{\beta}} S_n(\boldsymbol{\beta}_0(\tau)) + o_p(1) \\
        &= \frac{1}{n} \sum_{i=1}^n \frac{1}{h} K\left(\frac{e_i(\boldsymbol{\beta}_0(\tau))}{h}\right) \boldsymbol{x}_i \boldsymbol{x}_i^\top + o_p(1) \\
        &\xrightarrow{P} \mathbb{E}_{\boldsymbol{x}} \left[ \boldsymbol{x}_i \boldsymbol{x}_i^\top f_{\varepsilon}(0) \right] \\
        &= f_{\varepsilon}(0) \mathbb{E}[\boldsymbol{x}_i \boldsymbol{x}_i^\top] = f_{\varepsilon}(0)\Sigma_X = \mathcal{H}(\boldsymbol{\beta}_0(\tau)).
    \end{align*}
    Since $\mathcal{H}(\boldsymbol{\beta}_0(\tau))$ is positive definite and invertible, its inverse is $O_p(1)$.
    Substituting this into Eq.~\eqref{eq:rate_expansion}, we obtain:
    \begin{align}\label{eq:rate_simplified}
        \hat{\boldsymbol{\beta}}_h(\tau) - \boldsymbol{\beta}_0(\tau) = O_p(1) \cdot \frac{1}{n} S_n(\boldsymbol{\beta}_0(\tau)).
    \end{align}
    
    Next, we decompose the normalized score vector at the truth into a deterministic mean term and a zero-mean stochastic term:
    \begin{align*}
        \frac{1}{n} S_n(\boldsymbol{\beta}_0(\tau)) = \underbrace{\mathbb{E}\left[ \frac{1}{n} S_n(\boldsymbol{\beta}_0(\tau)) \right]}_{\text{Deterministic Component}} + \underbrace{\left( \frac{1}{n} S_n(\boldsymbol{\beta}_0(\tau)) - \mathbb{E}\left[ \frac{1}{n} S_n(\boldsymbol{\beta}_0(\tau)) \right] \right)}_{\text{Stochastic Component}}.
    \end{align*}
    Inheriting the i.i.d. structure of the observations $(y_i, \boldsymbol{x}_i)$ (Section \ref{subsec:motivation_and_background}), the stochastic term is an average of i.i.d. zero-mean random vectors with finite second moments (guaranteed by the boundedness of $\Psi_h$ and covariates in \textbf{Assumption A3}). By the central limit theorem, $\sqrt{n} \left( \frac{1}{n} S_n(\boldsymbol{\beta}_0(\tau)) - \mathbb{E}\left[ \frac{1}{n} S_n(\boldsymbol{\beta}_0(\tau)) \right] \right)$ converges in distribution to a normal distribution. Therefore:
    \begin{align*}
        \frac{1}{n} S_n(\boldsymbol{\beta}_0(\tau)) - \mathbb{E}\left[ \frac{1}{n} S_n(\boldsymbol{\beta}_0(\tau)) \right] = O_p(n^{-1/2}).
    \end{align*}
Furthermore, we analyze the expectation $\mathbb{E}[n^{-1} S_n(\boldsymbol{\beta}_0(\tau))] = -\mathbb{E}[\boldsymbol{x}_i] \mathbb{E}[\Psi_h(\varepsilon_{0i}; \tau)]$.
Invoking Lemma~\ref{lem:score_bias}, the expected smoothed score satisfies $\mathbb{E}[\Psi_h(\varepsilon_{0i}; \tau)] = O(h^2)$. 
Consequently, the bias term is simply:
\begin{align*}
    \mathbb{E}\left[ \frac{1}{n} S_n(\boldsymbol{\beta}_0(\tau)) \right] = -\mathbb{E}[\boldsymbol{x}_i] \cdot O(h^2) = O(h^2).
\end{align*}
Substituting these results back into Eq.~\eqref{eq:rate_simplified}:
    \begin{align*}
        \hat{\boldsymbol{\beta}}_h(\tau) - \boldsymbol{\beta}_0(\tau) &= O_p(1) \left[ O(h^2) + O_p(n^{-1/2}) \right] \\
        &= O_p(n^{-1/2} + h^2).
    \end{align*}
    This completes the proof.
\end{proof}

\phantomsection
\begin{proof}[Proof of Theorem~\ref{thm:bvm}]\label{pf:bvm}
    The proof relies on the local asymptotic normality (LAN) expansion of the log-posterior density. Due to the smoothing of the check loss via a continuous kernel $K(\cdot)$ (\textbf{Assumption A4}), the objective function is smooth, allowing us to employ a classical Taylor expansion approach.

    Recalling Eq.~\eqref{eq:sqr_log_likelihood}, since the normalizing term $-n \log Z(\theta, \tau, h)$ is independent of $\boldsymbol{\beta}$, we define the effective log-likelihood for $\boldsymbol{\beta}$ as $\ell_n(\boldsymbol{\beta}) \coloneqq -\theta \sum_{i=1}^n L_h(e_i(\boldsymbol{\beta}); \tau)$. We define the centered local parameter $\boldsymbol{t} = \sqrt{n}(\boldsymbol{\beta} - \hat{\boldsymbol{\beta}}_h(\tau))$, where $\hat{\boldsymbol{\beta}}_h(\tau)$ is the minimizer of $-\ell_n(\boldsymbol{\beta})$ (the sample smoothed risk minimizer). We examine the posterior density of $\boldsymbol{t}$, denoted as $\pi^*(\boldsymbol{t} \mid \boldsymbol{y})$. To characterize its asymptotic behavior, we begin by expanding the effective log-likelihood $\ell_n(\boldsymbol{\beta})$ around its minimizer $\hat{\boldsymbol{\beta}}_h(\tau)$. Utilizing the $C^2$ smoothness of $L_h(\cdot; \tau)$ guaranteed by \textbf{Assumption (A4)}, we have:
    \begin{align*}
    \ell_n(\boldsymbol{\beta}) = \ell_n(\hat{\boldsymbol{\beta}}_h(\tau)) + (\boldsymbol{\beta} - \hat{\boldsymbol{\beta}}_h(\tau))^\top \nabla_{\boldsymbol{\beta}} \ell_n(\hat{\boldsymbol{\beta}}_h(\tau)) + \frac{1}{2} (\boldsymbol{\beta} - \hat{\boldsymbol{\beta}}_h(\tau))^\top \nabla^2_{\boldsymbol{\beta}} \ell_n(\tilde{\boldsymbol{\beta}}(\tau)) (\boldsymbol{\beta} - \hat{\boldsymbol{\beta}}_h(\tau)),
	\end{align*}
    where $\tilde{\boldsymbol{\beta}}(\tau)$ lies on the line segment connecting $\boldsymbol{\beta}$ and $\hat{\boldsymbol{\beta}}_h(\tau)$.     
    Using the fact that the gradient vanishes at the minimizer ($\nabla_{\boldsymbol{\beta}} \ell_n(\hat{\boldsymbol{\beta}}_h(\tau)) = \boldsymbol{0}$) and substituting the local parameter parametrization $\boldsymbol{\beta} - \hat{\boldsymbol{\beta}}_h(\tau) = \boldsymbol{t}/\sqrt{n}$, the linear term disappears. The quadratic term simplifies as follows:
    \begin{align}\label{eq:taylor_expansion}
    \ell_n(\hat{\boldsymbol{\beta}}_h(\tau) + \boldsymbol{t}/\sqrt{n}) - \ell_n(\hat{\boldsymbol{\beta}}_h(\tau)) &= \frac{1}{2} \left(\frac{\boldsymbol{t}}{\sqrt{n}}\right)^\top \nabla^2_{\boldsymbol{\beta}} \ell_n(\tilde{\boldsymbol{\beta}}(\tau)) \left(\frac{\boldsymbol{t}}{\sqrt{n}}\right) \nonumber \\
    &= \frac{1}{2n} \boldsymbol{t}^\top \left( -\theta \sum_{i=1}^n \nabla^2_{\boldsymbol{\beta}} L_h(e_i(\tilde{\boldsymbol{\beta}}(\tau)); \tau) \right) \boldsymbol{t} \nonumber \\
    &= -\frac{1}{2} \boldsymbol{t}^\top \left( \frac{\theta}{n} \sum_{i=1}^n \nabla^2_{\boldsymbol{\beta}} L_h(e_i(\tilde{\boldsymbol{\beta}}(\tau)); \tau) \right) \boldsymbol{t}.
	\end{align}    
    The asymptotic behavior of the log-likelihood ratio in Eq.~\eqref{eq:taylor_expansion} is governed by the convergence of the scaled sample Hessian appearing in the parentheses. Recall from Section~\ref{subsec:technical_foundation_smoothing} (Eq.~\eqref{eq:hessian_matrix_structure}) that the Hessian of the smoothed objective is $\boldsymbol{H}_{\text{SQR}}(\boldsymbol{\beta}) = \sum_{i=1}^n 1/h K(e_i(\boldsymbol{\beta})/h) \boldsymbol{x}_i \boldsymbol{x}_i^\top$. Under \textbf{Assumption}~\ref{as:consistency_conditions} and general M-estimation theory (e.g., \shortciteNP{White1982}), the estimator $\hat{\boldsymbol{\beta}}_h(\tau)$ converges in probability to the pseudo-true parameter $\boldsymbol{\beta}^*_h(\tau) = \arg\min_{\boldsymbol{b}} \mathbb{E}[L_h(e_i(\boldsymbol{b}); \tau)]$. 
    Recalling that $\tilde{\boldsymbol{\beta}}(\tau)$ lies on the line segment connecting $\hat{\boldsymbol{\beta}}_h(\tau)$ and $\boldsymbol{\beta} = \hat{\boldsymbol{\beta}}_h(\tau) + \boldsymbol{t}/\sqrt{n}$, the distance $\|\tilde{\boldsymbol{\beta}}(\tau) - \hat{\boldsymbol{\beta}}_h(\tau)\|$ is bounded by $\|\boldsymbol{t}\|/\sqrt{n}$, which vanishes as $n \to \infty$. Consequently, the consistency of $\hat{\boldsymbol{\beta}}_h(\tau)$ implies $\tilde{\boldsymbol{\beta}}(\tau) \xrightarrow{P} \boldsymbol{\beta}^*_h(\tau)$.
    
    With the consistency of $\tilde{\boldsymbol{\beta}}(\tau)$ established, we explicitly establish the convergence of the sample Hessian matrix. We decompose the difference between the sample Hessian at the random intermediate point $\tilde{\boldsymbol{\beta}}(\tau)$ and the population Hessian at the fixed pseudo-truth $\boldsymbol{\beta}^*_h(\tau)$ as follows:
    \begin{align*}
         & \left\| \frac{1}{n} \sum_{i=1}^n \nabla^2_{\boldsymbol{\beta}} L_h(e_i(\tilde{\boldsymbol{\beta}}(\tau)); \tau) - \mathcal{H}(\boldsymbol{\beta}^*_h(\tau)) \right\| \\
        \le{}& \underbrace{\left\| \frac{1}{n} \sum_{i=1}^n \nabla^2_{\boldsymbol{\beta}} L_h(e_i(\tilde{\boldsymbol{\beta}}(\tau)); \tau) - \mathcal{H}(\tilde{\boldsymbol{\beta}}(\tau)) \right\|}_{\text{(I) Stochastic Equicontinuity}} + \underbrace{\left\| \mathcal{H}(\tilde{\boldsymbol{\beta}}(\tau)) - \mathcal{H}(\boldsymbol{\beta}^*_h(\tau)) \right\|}_{\text{(II) Continuity}}.
    \end{align*}
    The first term (I) converges to zero in probability by \textbf{Assumption (A7)}, which guarantees the uniform convergence of the sample Hessian over a neighborhood of $\boldsymbol{\beta}^*_h(\tau)$. The second term (II) converges to zero by the Continuous Mapping Theorem, given that the population Hessian $\mathcal{H}(\cdot)$ is continuous and $\tilde{\boldsymbol{\beta}}(\tau) \xrightarrow{P} \boldsymbol{\beta}^*_h(\tau)$.
    
    Thus, we have established that $\frac{1}{n} \sum_{i=1}^n \nabla^2_{\boldsymbol{\beta}} L_h(e_i(\tilde{\boldsymbol{\beta}}(\tau)); \tau) = \mathcal{H}(\boldsymbol{\beta}^*_h(\tau)) + o_p(1)$. Substituting this limit back into Eq.~\eqref{eq:taylor_expansion}, the log-likelihood ratio behaves asymptotically as:
	\begin{align} \label{eq:quadratic_limit}
    	\ell_n(\hat{\boldsymbol{\beta}}_h(\tau) + \boldsymbol{t}/\sqrt{n}) - \ell_n(\hat{\boldsymbol{\beta}}_h(\tau)) &= -\frac{1}{2} \boldsymbol{t}^\top \left( \theta \mathcal{H}(\boldsymbol{\beta}^*_h(\tau)) + o_p(1) \right) \boldsymbol{t} \nonumber \\
    	&= -\frac{1}{2} \boldsymbol{t}^\top (\theta \mathcal{H}(\boldsymbol{\beta}^*_h(\tau))) \boldsymbol{t} + o_p(1).
	\end{align}
        
    With the quadratic approximation of the likelihood established in Eq.~\eqref{eq:quadratic_limit}, we now derive the limiting form of the posterior distribution for the local parameter $\boldsymbol{t} = \sqrt{n}(\boldsymbol{\beta} - \hat{\boldsymbol{\beta}}_h(\tau))$. By the standard change of variable formula, the posterior density of $\boldsymbol{t}$, denoted as $\pi^*(\boldsymbol{t} \mid \boldsymbol{y})$, relates to the posterior of $\boldsymbol{\beta}$ via the Jacobian determinant of the transformation $\boldsymbol{\beta} = \hat{\boldsymbol{\beta}}_h(\tau) + n^{-1/2}\boldsymbol{t}$:
    \begin{align*}
        \pi^*(\boldsymbol{t} \mid \boldsymbol{y}) &= \pi\left( \hat{\boldsymbol{\beta}}_h(\tau) + \frac{\boldsymbol{t}}{\sqrt{n}} \bigg| \boldsymbol{y} \right) \left| \det\left( \frac{\partial \boldsymbol{\beta}}{\partial \boldsymbol{t}} \right) \right| \\
        &\propto \exp\left( \ell_n(\hat{\boldsymbol{\beta}}_h(\tau) + \boldsymbol{t}/\sqrt{n}) \right) \pi\left( \hat{\boldsymbol{\beta}}_h(\tau) + \boldsymbol{t}/\sqrt{n} \right) \cdot n^{-d/2} \\
        &\propto \exp\left( \ell_n(\hat{\boldsymbol{\beta}}_h(\tau) + \boldsymbol{t}/\sqrt{n}) \right) \pi\left( \hat{\boldsymbol{\beta}}_h(\tau) + \boldsymbol{t}/\sqrt{n} \right).
    \end{align*}
    Here, the Jacobian determinant $n^{-d/2}$ is a multiplicative constant with respect to $\boldsymbol{t}$ and is absorbed into the normalization constant. To analyze the asymptotic behavior, we normalize the likelihood by subtracting the constant $\ell_n(\hat{\boldsymbol{\beta}}_h(\tau))$ (the maximum log-likelihood value) inside the exponent. This does not affect the distribution since it corresponds to multiplying by a constant independent of $\boldsymbol{t}$. Substituting the quadratic approximation Eq.~\eqref{eq:quadratic_limit} into the exponent, we obtain:
    \begin{align*}
        \pi^*(\boldsymbol{t} \mid \boldsymbol{y}) &\propto \exp\left( \ell_n(\hat{\boldsymbol{\beta}}_h(\tau) + \boldsymbol{t}/\sqrt{n}) - \ell_n(\hat{\boldsymbol{\beta}}_h(\tau)) \right) \pi\left( \hat{\boldsymbol{\beta}}_h(\tau) + \boldsymbol{t}/\sqrt{n} \right) \\
        &= \exp\left( -\frac{1}{2} \boldsymbol{t}^\top (\theta \mathcal{H}(\boldsymbol{\beta}^*_h(\tau))) \boldsymbol{t} + o_p(1) \right) \pi\left( \hat{\boldsymbol{\beta}}_h(\tau) + \boldsymbol{t}/\sqrt{n} \right).
    \end{align*}
    
    Now we consider the limit as $n \to \infty$ for any fixed $\boldsymbol{t}$.
    First, regarding the prior term, since $\hat{\boldsymbol{\beta}}_h(\tau) \xrightarrow{P} \boldsymbol{\beta}^*_h(\tau)$ and the prior density $\pi(\cdot)$ is continuous and positive at $\boldsymbol{\beta}^*_h(\tau)$, we have:
    \begin{align*}
        \pi\left( \hat{\boldsymbol{\beta}}_h(\tau) + \boldsymbol{t}/\sqrt{n} \right) \xrightarrow{P} \pi(\boldsymbol{\beta}^*_h(\tau)) > 0.
    \end{align*}
    This implies that the prior becomes asymptotically constant over the local $\sqrt{n}$-neighborhood.
    Second, the exponential term converges to the kernel of a multivariate normal distribution:
    \begin{align*}
        \exp\left( -\frac{1}{2} \boldsymbol{t}^\top (\theta \mathcal{H}(\boldsymbol{\beta}^*_h(\tau))) \boldsymbol{t} + o_p(1) \right) \xrightarrow{d} \exp\left( -\frac{1}{2} \boldsymbol{t}^\top \left[ \theta \mathcal{H}(\boldsymbol{\beta}^*_h(\tau)) \right] \boldsymbol{t} \right).
    \end{align*}
    We recognize the term in the square brackets, $\theta \mathcal{H}(\boldsymbol{\beta}^*_h(\tau))$, as the precision matrix (inverse covariance matrix). 
     
    Combining this likelihood convergence with the fact that the prior becomes asymptotically constant, the posterior density converges pointwise to the multivariate normal density $\phi\left(\boldsymbol{t}; \boldsymbol{0}, (\theta \mathcal{H}(\boldsymbol{\beta}^*_h(\tau)))^{-1}\right)$. To extend this pointwise convergence to the total variation distance convergence, we invoke standard Bayesian asymptotic arguments (e.g., Theorem 2.1 in \shortciteNP{Ghosal2000}). This theorem ensures that the posterior mass outside any compact neighborhood of $\hat{\boldsymbol{\beta}}_h(\tau)$ vanishes exponentially fast, thereby preventing the integral of the absolute difference from diverging. Consequently, by Scheff\'{e}'s theorem, the pointwise convergence implies convergence in $L_1$ norm. The total variation distance converges to zero in probability:
    \begin{align*}
        \int_{\mathbb{R}^d} \left| \pi^*(\boldsymbol{t} \mid \boldsymbol{y}) - \phi\left(\boldsymbol{t}; \boldsymbol{0}, (\theta \mathcal{H}(\boldsymbol{\beta}^*_h(\tau)))^{-1}\right) \right| \mathrm{d} \boldsymbol{t} 
        &= \left\| \pi^*(\boldsymbol{t} \mid \boldsymbol{y}) - \phi\left(\boldsymbol{t}; \boldsymbol{0}, (\theta \mathcal{H}(\boldsymbol{\beta}^*_h(\tau)))^{-1}\right) \right\|_{TV} \\
        &= \left\| \pi(\boldsymbol{\vartheta} \mid \boldsymbol{y}, \boldsymbol{\mathcal{X}}, \theta) - \phi\left(\boldsymbol{\vartheta}; \boldsymbol{0}, \big(\theta \mathcal{H}(\boldsymbol{\beta}^*_h(\tau))\big)^{-1}\right) \right\|_{TV} \\
        &\xrightarrow{P} 0.
    \end{align*}
    This completes the proof.
\end{proof}

\phantomsection
\begin{proof}[Proof of Corollary~\ref{cor:validity_true}]\label{pf:cor_validity_true}
    We examine the asymptotic behavior of the quadratic form defining the credible set, evaluated at the true parameter $\boldsymbol{\beta}_0(\tau)$:
    \[
    Q_n(\boldsymbol{\beta}_0(\tau)) \coloneqq n(\boldsymbol{\beta}_0(\tau) - \hat{\boldsymbol{\beta}}_h(\tau))^\top (\hat{\theta} \mathcal{H}(\boldsymbol{\beta}^*_h(\tau))) (\boldsymbol{\beta}_0(\tau) - \hat{\boldsymbol{\beta}}_h(\tau)).
    \]
    To derive its limit, we first establish the asymptotic representation of the smoothed estimator $\hat{\boldsymbol{\beta}}_h(\tau)$. Recall from Eq.~\eqref{sqr_obj} that $\hat{\boldsymbol{\beta}}_h(\tau)$ minimizes the empirical smoothed risk $\widehat{R}_h(\boldsymbol{\beta})$. Thus, it satisfies the first-order condition $\nabla_{\boldsymbol{\beta}} \widehat{R}_h(\hat{\boldsymbol{\beta}}_h(\tau)) = \boldsymbol{0}$. A Taylor expansion of the gradient around the pseudo-truth $\boldsymbol{\beta}^*_h(\tau)$ yields:
    \[
    \boldsymbol{0} = \nabla_{\boldsymbol{\beta}} \widehat{R}_h(\hat{\boldsymbol{\beta}}_h(\tau)) = \nabla_{\boldsymbol{\beta}} \widehat{R}_h(\boldsymbol{\beta}^*_h(\tau)) + \nabla^2_{\boldsymbol{\beta}} \widehat{R}_h(\bar{\boldsymbol{\beta}}(\tau)) (\hat{\boldsymbol{\beta}}_h(\tau) - \boldsymbol{\beta}^*_h(\tau)),
    \]
    where $\bar{\boldsymbol{\beta}}(\tau)$ lies on the line segment connecting $\hat{\boldsymbol{\beta}}_h(\tau)$ and $\boldsymbol{\beta}^*_h(\tau)$. Note that the Hessian of the objective function, $\nabla^2_{\boldsymbol{\beta}} \widehat{R}_h(\bar{\boldsymbol{\beta}}(\tau))$, corresponds to the sample Hessian $\widehat{\boldsymbol{H}}_n(\bar{\boldsymbol{\beta}}(\tau))$ defined in Footnote~\ref{eq:sample_Hessian}. Inverting this Hessian, we obtain the stochastic component:
    \[
    \sqrt{n}(\hat{\boldsymbol{\beta}}_h(\tau) - \boldsymbol{\beta}^*_h(\tau)) = - \left[ \widehat{\boldsymbol{H}}_n(\bar{\boldsymbol{\beta}}(\tau)) \right]^{-1} \left( \sqrt{n} \nabla_{\boldsymbol{\beta}} \widehat{R}_h(\boldsymbol{\beta}^*_h(\tau)) \right) \eqqcolon \boldsymbol{\xi}_n.
    \]
    Since $\hat{\boldsymbol{\beta}}_h(\tau) \xrightarrow{P} \boldsymbol{\beta}^*_h(\tau)$ and $\bar{\boldsymbol{\beta}}(\tau)$ lies on the line segment connecting them, it follows that $\bar{\boldsymbol{\beta}}(\tau) \xrightarrow{P} \boldsymbol{\beta}^*_h(\tau)$. By \textbf{Assumption (A7)} (uniform convergence), the sample Hessian converges in probability: $\widehat{\boldsymbol{H}}_n(\bar{\boldsymbol{\beta}}(\tau)) \xrightarrow{P} \mathcal{H}(\boldsymbol{\beta}^*_h(\tau))$.
    Regarding the gradient term, the central limit theorem applied to the sum of i.i.d. zero-mean score vectors, which possess finite second moments by the boundedness of covariates in \textbf{Assumption (A3)}, implies $\sqrt{n} \nabla_{\boldsymbol{\beta}} \widehat{R}_h(\boldsymbol{\beta}^*_h(\tau)) \xrightarrow{d} \mathcal{N}(\boldsymbol{0}, \mathcal{J}(\boldsymbol{\beta}^*_h(\tau)))$.
    Therefore, by Slutsky's theorem, the product $\boldsymbol{\xi}_n$ converges in distribution to $\mathcal{N}(\boldsymbol{0}, \boldsymbol{\Sigma}_{\mathrm{sand}})$, where the sandwich covariance is defined as $\boldsymbol{\Sigma}_{\mathrm{sand}} \coloneqq \mathcal{H}(\boldsymbol{\beta}^*_h(\tau))^{-1}\mathcal{J}(\boldsymbol{\beta}^*_h(\tau))\mathcal{H}(\boldsymbol{\beta}^*_h(\tau))^{-1}$.
    
    Next, we address the deterministic bias term $\boldsymbol{\beta}^*_h(\tau) - \boldsymbol{\beta}_0(\tau)$. As established in \shortciteA{Horowitz1998a} using the second-order properties of the kernel $K(\cdot)$, the smoothing bias satisfies $\|\boldsymbol{\beta}^*_h(\tau) - \boldsymbol{\beta}_0(\tau)\| = O(h^2)$. Combining these two parts, we arrive at the bias-variance decomposition:
    \begin{align*}
        \sqrt{n}(\hat{\boldsymbol{\beta}}_h(\tau) - \boldsymbol{\beta}_0(\tau)) &= \sqrt{n}(\hat{\boldsymbol{\beta}}_h(\tau) - \boldsymbol{\beta}^*_h(\tau)) + \sqrt{n}(\boldsymbol{\beta}^*_h(\tau) - \boldsymbol{\beta}_0(\tau)) \\
        &= \underbrace{\boldsymbol{\xi}_n}_{\text{Stochastic Term}} + \underbrace{\sqrt{n} O(h^2)}_{\text{Smoothing Bias}}.
    \end{align*}
    The undersmoothing condition $n h^4 \to 0$ ensures that the bias term vanishes asymptotically ($\sqrt{n} O(h^2) \to 0$). Substituting this decomposition into $Q_n(\boldsymbol{\beta}_0(\tau))$ and noting that the cross-terms vanish asymptotically, we obtain the dominant stochastic component:
    \[
    Q_n(\boldsymbol{\beta}_0(\tau)) = \boldsymbol{\xi}_n^\top (\hat{\theta} \mathcal{H}(\boldsymbol{\beta}^*_h(\tau))) \boldsymbol{\xi}_n + o_p(1).
    \]
        
    We now derive the asymptotic distribution of the quadratic form $Q_n(\boldsymbol{\beta}_0(\tau))$. Under the calibration condition specified in Eq.~\eqref{eq:calibration_condition}, the scale estimator satisfies $\hat{\theta} \xrightarrow{P} \theta^*$. Consequently, by the continuous mapping theorem, the weighting matrix in $Q_n(\boldsymbol{\beta}_0(\tau))$ converges in probability to:
    \[
    \hat{\theta} \mathcal{H}(\boldsymbol{\beta}^*_h(\tau)) \xrightarrow{P} \theta^* \mathcal{H}(\boldsymbol{\beta}^*_h(\tau)).
    \]
    Invoking Slutsky's theorem again, $Q_n(\boldsymbol{\beta}_0(\tau))$ converges in distribution to a random variable $W$:
    \[
    Q_n(\boldsymbol{\beta}_0(\tau)) \xrightarrow{d} W \coloneqq \boldsymbol{\xi}^\top (\theta^* \mathcal{H}(\boldsymbol{\beta}^*_h(\tau))) \boldsymbol{\xi},
    \]
    where $\boldsymbol{\xi} \sim \mathcal{N}(\boldsymbol{0}, \boldsymbol{\Sigma}_{\mathrm{sand}})$. The limiting variable $W$ can be represented as a weighted sum of independent chi-squared variables:
    \[
    W = \sum_{j=1}^d \lambda_j Z_j^2, \quad Z_j \stackrel{i.i.d.}{\sim} \mathcal{N}(0,1),
    \]
    where $\lambda_1, \dots, \lambda_d$ are the deterministic eigenvalues of the matrix product $\boldsymbol{\Lambda}^* \coloneqq (\theta^* \mathcal{H}(\boldsymbol{\beta}^*_h(\tau))) \boldsymbol{\Sigma}_{\mathrm{sand}}$.
    
    While strict convergence to a standard $\chi^2_d$ distribution would require $\lambda_j = 1$ for all $j$ (implying $(\theta^* \mathcal{H}(\boldsymbol{\beta}^*_h(\tau)))^{-1} = \boldsymbol{\Sigma}_{\mathrm{sand}}$), such matrix equality generally fails under model misspecification. However, our proposed calibration imposes a strict constraint on the first moment of the limiting distribution $W$. The expectation of $W$ is strictly determined by the trace of the product matrix $\boldsymbol{\Lambda}^*$. Exploiting the cyclic invariance of the trace operator, we explicitly derive:
    \begin{align*}
    \mathbb{E}[W] &= \sum_{j=1}^d \lambda_j = \mathrm{tr}(\boldsymbol{\Lambda}^*) \\
    &= \mathrm{tr}\left( (\theta^* \mathcal{H}(\boldsymbol{\beta}^*_h(\tau))) \left[ \mathcal{H}(\boldsymbol{\beta}^*_h(\tau))^{-1} \mathcal{J}(\boldsymbol{\beta}^*_h(\tau)) \mathcal{H}(\boldsymbol{\beta}^*_h(\tau))^{-1} \right] \right) \\
    &= \theta^* \, \mathrm{tr}\left( \mathcal{J}(\boldsymbol{\beta}^*_h(\tau)) \mathcal{H}(\boldsymbol{\beta}^*_h(\tau))^{-1} \right).
    \end{align*}
    Substituting the definition of the target parameter $\theta^* = d / \mathrm{tr}(\mathcal{J}(\boldsymbol{\beta}^*_h(\tau))\mathcal{H}(\boldsymbol{\beta}^*_h(\tau))^{-1})$ from Eq.~\eqref{eq:calibration_condition}, we obtain:
    \[
    \mathbb{E}[W] = \left( \frac{d}{\mathrm{tr}(\mathcal{J}(\boldsymbol{\beta}^*_h(\tau))\mathcal{H}(\boldsymbol{\beta}^*_h(\tau))^{-1})} \right) \mathrm{tr}\left( \mathcal{J}(\boldsymbol{\beta}^*_h(\tau))\mathcal{H}(\boldsymbol{\beta}^*_h(\tau))^{-1} \right) = d.
    \]
    Consequently, although the limiting random variable $W$ is a weighted sum of chi-squares, our calibration ensures its expectation matches that of a standard $\chi^2_d$ distribution exactly ($\mathbb{E}[W] = d$). While the Satterthwaite approximation \shortcite{Satterthwaite1946} typically matches both first and second moments to a scaled chi-square distribution with fractional degrees of freedom, our strategy follows \shortciteA{Mueller2013} by enforcing the generalized information equality (first-moment matching) to approximate $W$ with a standard $\chi^2_d$ reference distribution.

    To characterize the approximation error, we consider the Edgeworth-type expansion of the CDF of $W$ around the reference $\chi^2_d$ distribution (see, e.g., \shortciteNP{Rothenberg1984}). Let $G_d(x)$ denote the CDF of the standard $\chi^2_d$ distribution, and let $q_{1-\alpha} = \chi^2_{d, 1-\alpha}$ denote the critical value. Since our calibration ensures that the first-moment mismatch is zero (i.e., $\mathbb{E}[W] = d$), the leading-order error term involving the first derivative $G_d'(q_{1-\alpha})$ vanishes exactly. The asymptotic coverage probability is thus dominated by the second-moment mismatch:
    \begin{align*}
    \lim_{n \to \infty} \mathbb{P}\left( \boldsymbol{\beta}_0(\tau) \in C_{1-\alpha} \right) &= F_W(q_{1-\alpha}) \\
    &= G_d(q_{1-\alpha}) - \underbrace{(\mathbb{E}[W] - d)}_{=0} G_d'(q_{1-\alpha}) + O\left( |\text{Var}(W) - 2d| \right) \\
    &= (1-\alpha) + O\left( |\text{Var}(W) - 2d| \right).
    \end{align*}
    The error term $O(|\text{Var}(W) - 2d|)$ reflects the residual variance mismatch inherent in scalar calibration for multidimensional parameters ($d>1$). However, as established by \shortciteA{Mueller2013}, this trace-based calibration minimizes the Kullback-Leibler divergence between the true and approximated distributions within the scalar family. Therefore, under the generalized Wilks' phenomenon, this moment-matched limit ensures that the Bayesian credible intervals possess valid frequentist coverage to the first order. This completes the proof.
\end{proof}

\subsection{Proofs for posterior propriety}
\label{app:proofs_propriety}

\phantomsection
\begin{proof}[Proof of Theorem~\ref{thm:propriety1}]
\label{pf:thm_propriety1}
\textit{Proof of Part (i).} Let $I_1 = \int_{\mathbb{R}^d} L(\boldsymbol{y} \mid \boldsymbol{\mathcal{X}}, \boldsymbol{\beta}, \theta; \tau, h) \pi(\boldsymbol{\beta}) \,\mathrm{d}\boldsymbol{\beta}$. From Eq.~\eqref{eq:sqr_likelihood}, we can express $I_1$ as:
\begin{align*}
    I_1 \propto (Z(\theta, \tau, h))^{-n} \int_{\mathbb{R}^d} \exp\left( -\theta S(\boldsymbol{\beta}; \tau, h) \right) \,\mathrm{d}\boldsymbol{\beta}.
\end{align*}
Since $L_h(e; \tau) = (\rho_{\tau} * K_h)(e) \ge 0$ (due to the non-negativity of $\rho_\tau$ and $K_h$), it follows that $S(\boldsymbol{\beta}; \tau, h) \ge 0$. The integrand $\exp(-\theta S(\boldsymbol{\beta};\tau,h))$ is continuous and strictly positive, and $(Z(\theta, \tau, h))^{-n} > 0$, which implies $I_1 > 0$.

To establish the finiteness of $I_1$, we begin by deriving a lower bound for $L_h(e; \tau)$. The check function $\rho_{\tau}(u)$ satisfies $\rho_{\tau}(u) \ge \min(\tau, 1-\tau)|u|$. Letting $c_{\tau} = \min(\tau, 1-\tau) > 0$, we obtain the following lower bound for the smoothed loss:
\begin{align*}
    L_h(e; \tau) &= \int_{-\infty}^{\infty} \rho_{\tau}(e-v) K_h(v) \,\mathrm{d} v \ge c_{\tau} \int_{-\infty}^{\infty} |e-v| K_h(v) \,\mathrm{d} v \ge c_{\tau}\int_{-\infty}^{\infty} (|e| - |v|) K_h(v) \,\mathrm{d} v \\
    &= c_{\tau}|e| \int_{-\infty}^{\infty} K_h(v) \,\mathrm{d} v - c_{\tau}\int_{-\infty}^{\infty} |v| K_h(v) \,\mathrm{d} v = c_{\tau}|e| - c_{\tau}\int_{-\infty}^{\infty} |v| \frac{1}{h} K\left(\frac{v}{h}\right) \,\mathrm{d} v \\
    &= c_{\tau}(|e| - h M_K)
\end{align*}
by applying the reverse triangle inequality $|e-v| \ge |e| - |v|$ and the substitution $u = v/h$, where $M_K = \int_{-\infty}^{\infty} |u| K(u) \,\mathrm{d} u$ is finite by assumption. More precisely, there exist constants $c_1 = c_{\tau} > 0$ and $C_1 = c_{\tau} h M_K \ge 0$ such that $L_h(e; \tau) \ge c_1 |e| - C_1$ for all $e$.

Next, we derive a lower bound for $S(\boldsymbol{\beta}; \tau, h)$. Summing over all observations gives
\begin{align*}
S(\boldsymbol{\beta}; \tau, h) &\ge \sum_{i=1}^n (c_1 |y_i - \boldsymbol{x}_i^\top\boldsymbol{\beta}| - C_1) = c_1 \|\boldsymbol{y} - \boldsymbol{\mathcal{X}}\boldsymbol{\beta}\|_1 - n C_1.
\end{align*}
By norm equivalence, there exists $c_2 > 0$ such that $\|\boldsymbol{z}\|_1 \ge c_2 \|\boldsymbol{z}\|_2$, which implies $S(\boldsymbol{\beta}; \tau, h) \ge c_1 c_2 \|\boldsymbol{y} - \boldsymbol{\mathcal{X}}\boldsymbol{\beta}\|_2 - n C_1$. Using the reverse triangle inequality, we have $\|\boldsymbol{y} - \boldsymbol{\mathcal{X}}\boldsymbol{\beta}\|_2 \ge \|\boldsymbol{\mathcal{X}}\boldsymbol{\beta}\|_2 - \|\boldsymbol{y}\|_2$. Since the design matrix $\boldsymbol{\mathcal{X}}$ has full column rank $d$, the mapping $\boldsymbol{\beta} \mapsto \boldsymbol{\mathcal{X}}\boldsymbol{\beta}$ is injective. This implies that $\|\boldsymbol{\mathcal{X}}\boldsymbol{\beta}\|_2$ defines a norm on $\mathbb{R}^d$. By norm equivalence on $\mathbb{R}^d$, there exists $c_3 > 0$ such that $\|\boldsymbol{\mathcal{X}}\boldsymbol{\beta}\|_2 \ge c_3 \|\boldsymbol{\beta}\|_2$.
Combining these, we obtain for sufficiently large $\|\boldsymbol{\beta}\|_2$:
\begin{align*}
S(\boldsymbol{\beta}; \tau, h) &\ge c_1 c_2 (c_3 \|\boldsymbol{\beta}\|_2 - \|\boldsymbol{y}\|_2) - n C_1 = (c_1 c_2 c_3) \|\boldsymbol{\beta}\|_2 - (c_1 c_2 \|\boldsymbol{y}\|_2 + n C_1).
\end{align*}

Letting $c_L = c_1 c_2 c_3 > 0$ and $C_S = c_1 c_2 \|\boldsymbol{y}\|_2 + n C_1 \ge 0$, we have the linear growth condition:
\begin{align*}
S(\boldsymbol{\beta}; \tau, h) \ge c_L \|\boldsymbol{\beta}\|_2 - C_S.
\end{align*}
This growth rate ensures the exponential decay of the integrand:
\begin{align*}
\exp(-\theta S(\boldsymbol{\beta}; \tau, h)) &\le \exp(-\theta(c_4 \|\boldsymbol{\beta}\|_2 - C_2)) = C e^{-A \|\boldsymbol{\beta}\|_2},
\end{align*}
where $A = \theta c_4 > 0$ and $C = e^{\theta C_2} > 0$.

To evaluate the integral $\int_{\mathbb{R}^d} e^{-A \|\boldsymbol{\beta}\|_2} \mathrm{d}\boldsymbol{\beta}$, we employ $d$-dimensional hyperspherical coordinates. Letting $r = \|\boldsymbol{\beta}\|_2$, the volume element becomes $\mathrm{d}\boldsymbol{\beta} = r^{d-1} \mathrm{d}r \mathrm{d}\Omega_{d-1}$, where $\mathrm{d}\Omega_{d-1}$ is the surface element on the unit $(d-1)$-sphere. Thus,
\begin{align*}
\int_{\mathbb{R}^d} e^{-A \|\boldsymbol{\beta}\|_2} \mathrm{d}\boldsymbol{\beta} &= \int_{\Omega_{d-1}} \int_0^\infty e^{-Ar} r^{d-1}\, \mathrm{d}r \, \mathrm{d}\Omega_{d-1} = \left( \int_{\Omega_{d-1}} \mathrm{d}\Omega_{d-1} \right) \left( \int_0^\infty e^{-Ar} r^{d-1} \mathrm{d}r \right) \\
&= S_{d-1} \int_0^\infty e^{-Ar} r^{d-1} \mathrm{d}r,
\end{align*}
where $S_{d-1} = 2\pi^{d/2}/\Gamma(d/2)$ is the surface area of the unit $(d-1)$-sphere. The substitution $t = Ar$ yields
\begin{align*}
\int_0^\infty e^{-Ar} r^{d-1}\, \mathrm{d}r &= \int_0^\infty e^{-t} \left(\frac{t}{A}\right)^{d-1} \frac{1}{A} \,\mathrm{d}t = \frac{1}{A^d} \int_0^\infty t^{d-1} e^{-t} \,\mathrm{d}t = \frac{\Gamma(d)}{A^d} = \frac{(d-1)!}{A^d},
\end{align*}
since $d$ is a positive integer. Therefore, the integral of the non-negative function $\exp(-\theta S(\boldsymbol{\beta}; \tau, h))$ is
\begin{align*}
\int_{\mathbb{R}^d} \exp(-\theta S(\boldsymbol{\beta}; \tau, h)) \,\mathrm{d}\boldsymbol{\beta} \le \int_{\mathbb{R}^d} C e^{-A \|\boldsymbol{\beta}\|_2} \,\mathrm{d}\boldsymbol{\beta} = C \int_{\mathbb{R}^d} e^{-A \|\boldsymbol{\beta}\|_2} \,\mathrm{d}\boldsymbol{\beta} = C S_{d-1} \frac{(d-1)!}{A^d} < \infty.
\end{align*}
Since $\exp(-\theta S(\boldsymbol{\beta}; \tau, h))$ is non-negative and dominated by the integrable function $C e^{-A \|\boldsymbol{\beta}\|_2}$, the comparison test guarantees the convergence of $I_1$. Combined with the earlier result $I_1 > 0$, we conclude that $0 < I_1 < \infty$, establishing the propriety of the posterior distribution $\pi(\boldsymbol{\beta}\mid \boldsymbol{y}, \boldsymbol{\mathcal{X}},\theta)$.

\textit{Proof of Part (ii).} The joint posterior distribution satisfies
\begin{align*}
    \pi(\boldsymbol{\beta}, \theta \mid \boldsymbol{y}, \boldsymbol{\mathcal{X}}) \propto \big(Z(\theta, \tau, h)\big)^{-n} \exp\left(-\theta S(\boldsymbol{\beta}; \tau, h)\right)\pi(\theta),
\end{align*}
with propriety requiring the normalization integral 
\begin{align}\label{eq:J_1}
    J_1 \coloneqq \int_0^\infty \left(Z(\theta, \tau, h)\right)^{-n} I_{\boldsymbol{\beta}}(\theta)\pi(\theta) \,\mathrm{d} \theta
\end{align}
to be finite and positive, where $I_{\boldsymbol{\beta}}(\theta) \coloneqq \int_{\mathbb{R}^d} \exp\left(-\theta S(\boldsymbol{\beta}; \tau, h)\right) \,\mathrm{d} \boldsymbol{\beta}$. From the proof of Theorem~\ref{thm:propriety1}(i), we know that there exist constants $c_L > 0$ and $C_S \geq 0$ such that $S(\boldsymbol{\beta}; \tau, h) \geq c_L \|\boldsymbol{\beta}\|_2 - C_S$. This yields the bound
\begin{align*}
    \exp\big(-\theta S(\boldsymbol{\beta}; \tau, h)\big) 
    &\leq \exp\big(-\theta(c_L \|\boldsymbol{\beta}\|_2 - C_S)\big) = e^{\theta C_S} \exp\big(-\theta c_L \|\boldsymbol{\beta}\|_2\big).
\end{align*}
Hence $I_{\boldsymbol{\beta}}(\theta)$ is dominated by
\begin{align}\label{I_beta}
    I_{\boldsymbol{\beta}}(\theta) \leq e^{\theta C_S} \int_{\mathbb{R}^d} \exp\big(-\theta c_L \|\boldsymbol{\beta}\|_2\big) \mathrm{d}\boldsymbol{\beta}.
\end{align}
The radial integral evaluates via hyperspherical coordinates to $\int_{\mathbb{R}^d} \exp(-\alpha \|\boldsymbol{x}\|_2) \mathrm{d}\boldsymbol{x} = \frac{2\pi^{d/2}}{\Gamma(d/2)} \Gamma(d) \alpha^{-d}$ for $\alpha > 0$. Substituting $\alpha = \theta c_L$ we obtain the upper bound for Eq.~\eqref{I_beta}:
\begin{align*}
    I_{\boldsymbol{\beta}}(\theta) \leq K_1 e^{\theta C_S} \theta^{-d}, \quad K_1 = \frac{2\pi^{d/2}\Gamma(d)}{\Gamma(d/2) c_L^d} > 0.
\end{align*}
Thus, substitution into \eqref{eq:J_1} yields the sufficient condition \eqref{eq:integral_condition_theta_prior}:
\begin{align*}
  J_1 \leq K_1 \int_0^\infty \frac{\pi(\theta) e^{\theta C_S}}{\big(Z(\theta, \tau, h)\big)^n \theta^d} \mathrm{d}\theta < \infty.
\end{align*}

Now, consider the specific case where the prior for $\theta$ is a Gamma distribution, $\pi(\theta) = \frac{b^a}{\Gamma(a)}\theta^{a-1}e^{-b\theta}$, for $a>0, b>0$. To establish the propriety of the joint posterior, it is sufficient to show that the integral of its upper bound, as derived from the bound on $I_{\boldsymbol{\beta}}(\theta)$, is finite. Let's analyze the integrand of this bounding integral:
\begin{align*}
    \frac{\pi(\theta) e^{\theta C_S}}{\big(Z(\theta, \tau, h)\big)^n \theta^d} &= \frac{b^a}{\Gamma(a)}\theta^{a-1}e^{-b\theta} \cdot \frac{e^{\theta C_S}}{\big(Z(\theta, \tau, h)\big)^n \theta^d} \\
    &= \frac{b^a}{\Gamma(a)} \left(Z(\theta, \tau, h)\right)^{-n} \theta^{a-d-1} e^{-(b - C_S)\theta}.
\end{align*}
The convergence of the integral of this expression over $(0, \infty)$ depends on its asymptotic behavior as $\theta \to \infty$ and $\theta \to 0$:
\begin{enumerate}
\item \textbf{Behavior as $\theta \to \infty$:}
To establish the convergence condition for large $\theta$, we first examine the asymptotic behavior of $(Z(\theta, \tau, h))^{-n}$. This entails applying Laplace's method to $Z(\theta, \tau, h)$, which exploits the fact that, for large $\theta$, the integrand $e^{-\theta L_h(u;\tau)}$ concentrates sharply around the global minimum of $L_h(u;\tau)$. By assumption in Theorem~\ref{thm:propriety1}(ii), this minimum is $L_{\min} > 0$ and achieved at a unique point $u_{\min}$. The Taylor expansion of $L_h(u;\tau)$ around $u_{\min}$ yields
\[
L_h(u;\tau) = L_{\min} + \frac{L_h''(u_{\min};\tau)}{2}(u - u_{\min})^2 + O\big((u - u_{\min})^3\big).
\]
The leading asymptotic contribution to $Z(\theta, \tau, h)$ arises from integrating the exponential of the constant and quadratic terms:
{\small
\begin{align*}
\int_{-\infty}^\infty \exp\left(-\theta \left[ L_{\min} + \frac{L_h''(u_{\min};\tau)}{2}(u - u_{\min})^2 \right]\right) \,\mathrm{d}u
&= e^{-\theta L_{\min}} \int_{-\infty}^\infty \exp\left(-\frac{\theta L_h''(u_{\min};\tau)}{2}(u - u_{\min})^2\right) \,\mathrm{d}u.
\end{align*}
}
The integral evaluates to the Gaussian form $\sqrt{2\pi / (\theta L_h''(u_{\min};\tau))}$. Laplace's method formalizes that the full integral $Z(\theta, \tau, h)$ equals this leading term times a factor approaching 1, incorporating higher-order corrections:
\begin{align*}
	Z(\theta, \tau, h) = \sqrt{\frac{2\pi}{L_h''(u_{\min};\tau)}} \, \theta^{-1/2} e^{-\theta L_{\min}} \left(1 + o(1)\right).
\end{align*}
Consequently,
\begin{align*}
(Z(\theta, \tau, h))^{-n} &= \left(\sqrt{\frac{2\pi}{L_h''(u_{\min};\tau)}} \, \theta^{-1/2} e^{-\theta L_{\min}} (1 + o(1)) \right)^{-n} \\
&= C\cdot \theta^{n/2} e^{n L_{\min} \theta} \left(1 + o(1)\right),
\end{align*}
where the constant $C = \left(L_h''(u_{\min};\tau) / (2\pi)\right)^{n/2} > 0$ absorbs the prefactor, and $(1 + o(1))^{-n} = 1 + o(1)$ for fixed $n$. Disregarding the constant $\frac{b^a}{\Gamma(a)}$, the relevant expression becomes
\begin{align*}
(Z(\theta, \tau, h))^{-n} \theta^{a - d - 1} e^{-(b - C_S) \theta} &= C \cdot \theta^{n/2} e^{n L_{\min} \theta} (1 + o(1)) \cdot \theta^{a - d - 1} \cdot e^{-(b - C_S) \theta} \\
&= C \cdot \theta^{a + n/2 - d - 1} e^{-\left(b - C_S - n L_{\min}\right) \theta} (1 + o(1)).
\end{align*}
For the integral over $\theta$ to converge as the upper limit tends to infinity, the integrand must decay. The exponential $e^{-\left(b - C_S - n L_{\min}\right) \theta}$ dominates the polynomial terms for large $\theta$, necessitating a positive coefficient for exponential decay:
\begin{align*}
	b - C_S - n L_{\min} > 0 \quad \implies \quad b > C_S + n L_{\min}.
\end{align*}

    \item \textbf{Behavior as $\theta \to 0$:} Near zero, we use the assumption that $(Z(\theta, \tau, h))^{-n} = O(\theta^{k_Z})$ for some constant $k_Z$. The integrand's behavior is therefore dominated by the polynomial term:
    $$ O(\theta^{k_Z}) \cdot \theta^{a-d-1} e^{-(b-C_S)\theta} \sim O(\theta^{a+k_Z-d-1}). $$
    For an integral of the form $\int_0^\epsilon \theta^p \,\mathrm{d}\theta$ to converge, the exponent must satisfy $p > -1$. Applying this to our case, we require:
    $$ a + k_Z - d - 1 > -1 \implies a + k_Z > d. $$
\end{enumerate}
Both conditions must hold simultaneously to guarantee the convergence of the bounding integral. Therefore, for a $\text{Gamma}(a, b)$ prior on $\theta$, the joint posterior distribution $\pi(\boldsymbol{\beta}, \theta \mid \boldsymbol{y}, \boldsymbol{\mathcal{X}})$ is proper if $b > C_S + n L_{\min}$ and $a + k_Z > d$.
\end{proof}

\phantomsection
\begin{proof}[Proof of Theorem~\ref{thm:propriety2}]\label{pf:thm_propriety2}
\textit{Proof of Part (i).} The integral of interest, ignoring proportionality constants from the likelihood not involving $\boldsymbol{\beta}$, is
\begin{align*}
  I_{2} \coloneqq \int_{\mathbb{R}^d} \exp\left( -\theta S(\boldsymbol{\beta}; \tau, h) \right) \exp\left( -\frac{\|\boldsymbol{\beta}\|_2^2}{2\sigma_{\boldsymbol{\beta}}^2} \right) \,\mathrm{d}\boldsymbol{\beta}.
\end{align*}
Let $g(\boldsymbol{\beta}) = \exp\left( -\theta S(\boldsymbol{\beta}; \tau, h) - \|\boldsymbol{\beta}\|_2^2 / (2\sigma_{\boldsymbol{\beta}}^2) \right)$.
Since $S(\boldsymbol{\beta}; \tau, h) \ge 0$ and $\theta > 0$, the term $\exp(-\theta S(\boldsymbol{\beta}; \tau, h))$ is strictly positive and bounded above by 1. The Gaussian prior term $\exp(-\|\boldsymbol{\beta}\|_2^2 / (2\sigma_{\boldsymbol{\beta}}^2))$ is strictly positive. Thus, $g(\boldsymbol{\beta})$ is strictly positive and continuous. Its integral over $\mathbb{R}^d$ must therefore be strictly positive, ensuring $I_{2} > 0$.

To establish finiteness, we use the comparison test. Since $-\theta S(\boldsymbol{\beta}; \tau, h) \le 0$, we have
$$ -\theta S(\boldsymbol{\beta}; \tau, h) - \frac{\|\boldsymbol{\beta}\|_2^2}{2\sigma_{\boldsymbol{\beta}}^2} \le - \frac{\|\boldsymbol{\beta}\|_2^2}{2\sigma_{\boldsymbol{\beta}}^2}. $$
Exponentiating gives $g(\boldsymbol{\beta}) \le \exp\left( - \frac{\|\boldsymbol{\beta}\|_2^2}{2\sigma_{\boldsymbol{\beta}}^2} \right)$.
The integral of this upper bound is $\int_{\mathbb{R}^d} \exp\left( - \frac{\|\boldsymbol{\beta}\|_2^2}{2\sigma_{\boldsymbol{\beta}}^2} \right) \,\mathrm{d}\boldsymbol{\beta} = (2\pi\sigma_{\boldsymbol{\beta}}^2)^{d/2}$, which is finite. By comparison, $I_{2}$ is also finite. Thus, $0 < I_{2} < \infty$, and the posterior is proper.

\textit{Proof of Part (ii).}
The joint posterior is proportional to 
\begin{align*}
  \pi(\boldsymbol{\beta}, \theta \mid \boldsymbol{y}, \boldsymbol{\mathcal{X}}, \sigma_{\boldsymbol{\beta}}^2, a_{\theta}, b_{\theta}) \propto (Z(\theta, \tau, h))^{-n} \exp(-\theta S(\boldsymbol{\beta}; \tau, h)) \pi(\boldsymbol{\beta} \mid \sigma_{\boldsymbol{\beta}}^2) \pi(\theta \mid a_\theta, b_\theta).
\end{align*}
The joint posterior density is proportional to
\begin{align*}
  \pi(\boldsymbol{\beta}, \theta \mid \boldsymbol{y}, \boldsymbol{\mathcal{X}}, \sigma_{\boldsymbol{\beta}}^2) \propto L(\boldsymbol{y} \mid \boldsymbol{\mathcal{X}}, \boldsymbol{\beta}, \theta) \pi(\boldsymbol{\beta} \mid \sigma_{\boldsymbol{\beta}}^2) \pi(\theta).
\end{align*}
Substituting the likelihood form $L(\boldsymbol{y} \mid \boldsymbol{\mathcal{X}}, \boldsymbol{\beta}, \theta; \tau, h)$ as showed in Eq.~\eqref{eq:sqr_likelihood}, we have
\begin{align*}
  \pi(\boldsymbol{\beta}, \theta \mid \boldsymbol{y}, \boldsymbol{\mathcal{X}}, \sigma_{\boldsymbol{\beta}}^2) \propto (Z(\theta, \tau, h))^{-n} \exp(-\theta S(\boldsymbol{\beta}; \tau, h)) \pi(\boldsymbol{\beta} \mid \sigma_{\boldsymbol{\beta}}^2) \pi(\theta).
\end{align*}
To show propriety, we need to demonstrate that the integral of this joint posterior over $\boldsymbol{\beta}$ and $\theta$ is finite and positive:
\begin{align*}
  J_{2} \coloneqq \int_0^\infty \int_{\mathbb{R}^d} (Z(\theta, \tau, h))^{-n} \exp(-\theta S(\boldsymbol{\beta}; \tau, h)) \pi(\boldsymbol{\beta} \mid \sigma_{\boldsymbol{\beta}}^2) \pi(\theta) \,\mathrm{d}\boldsymbol{\beta} \,\mathrm{d}\theta.
\end{align*}
Since the integrand is non-negative, we can apply Tonelli's theorem to change the order of integration:
\begin{align*}
  J_{2} = \int_0^\infty (Z(\theta, \tau, h))^{-n} \pi(\theta) \left( \int_{\mathbb{R}^d} \exp(-\theta S(\boldsymbol{\beta}; \tau, h)) \pi(\boldsymbol{\beta} \mid \sigma_{\boldsymbol{\beta}}^2) \,\mathrm{d}\boldsymbol{\beta} \right) \,\mathrm{d} \theta.
\end{align*}
Since the smoothed loss function is non-negative ($S(\boldsymbol{\beta}; \tau, h) \ge 0$), the term $\exp(-\theta S(\boldsymbol{\beta}; \tau, h))$ is bounded by 1. Therefore, the inner integral over $\boldsymbol{\beta}$ is bounded by the integral of the prior density itself:
\begin{align*}
    \int_{\mathbb{R}^d} \exp(-\theta S(\boldsymbol{\beta}; \tau, h)) \pi(\boldsymbol{\beta} \mid \sigma_{\boldsymbol{\beta}}^2) \,\mathrm{d}\boldsymbol{\beta} \le \int_{\mathbb{R}^d} 1 \cdot \pi(\boldsymbol{\beta} \mid \sigma_{\boldsymbol{\beta}}^2) \,\mathrm{d}\boldsymbol{\beta} = 1.
\end{align*}
Substituting this upper bound into the expression for $J_{2}$, we obtain:
$$ J_{2} \le \int_0^\infty (Z(\theta, \tau, h))^{-n} \pi(\theta) \cdot 1 \,\mathrm{d}\theta. $$
Therefore, if the integral $\int_0^\infty (Z(\theta, \tau, h))^{-n} \pi(\theta) \,\mathrm{d}\theta$ converges to a finite positive value, then $J_{2}$ is finite. The strict positivity of $J_{2}$ is guaranteed because $(Z(\theta, \tau, h))^{-n} > 0$ and the priors are proper. For the specific instance where $\pi(\theta) \sim \text{Gamma}(a_\theta, b_\theta)$, we analyze the convergence of
$$ \int_0^\infty (Z(\theta, \tau, h))^{-n} \theta^{a_\theta-1} e^{-b_\theta\theta} \,\mathrm{d}\theta. $$
Using the asymptotic behavior derived in the proof of Theorem~\ref{thm:propriety1}, we know $(Z(\theta, \tau, h))^{-n}$ behaves like $\theta^{n/2} e^{n L_{\min} \theta}$ as $\theta \to \infty$. Thus, the integrand is proportional to $\theta^{a_\theta + n/2 - 1} e^{-(b_\theta - n L_{\min})\theta}$. For convergence at infinity, we require the exponential decay rate to be positive: $b_\theta > n L_{\min}$. Near zero, assuming $(Z(\theta, \tau, h))^{-n} = O(\theta^{k_Z})$, convergence requires $a_\theta + k_Z > 0$. Thus, the joint posterior is proper if $b_\theta > n L_{\min}$ and $a_\theta + k_Z > 0$.
\end{proof}

\phantomsection
\begin{proof}[Proof of Corollary~\ref{cor:1}]\label{pf:cor_1}
\textit{Proof of Part (i).} 
We aim to show the finiteness of the marginal posterior integral for a fixed $\theta$. Let this integral be denoted as $J_{3}$:
\begin{align*}
    J_{3} \coloneqq \int_0^\infty \int_{\mathbb{R}^d} L(\boldsymbol{y} \mid \boldsymbol{\mathcal{X}}, \boldsymbol{\beta}, \theta; \tau, h) \pi(\boldsymbol{\beta}\mid \sigma_{\boldsymbol{\beta}}^2) \pi(\sigma_{\boldsymbol{\beta}}^2\mid a_0, b_0) \,\mathrm{d}\boldsymbol{\beta} \,\mathrm{d}\sigma_{\boldsymbol{\beta}}^2.
\end{align*}
Recall the explicit form of the likelihood function:
\begin{align*}
    L(\boldsymbol{y} \mid \boldsymbol{\mathcal{X}}, \boldsymbol{\beta}, \theta; \tau, h) = (Z(\theta, \tau, h))^{-n} \exp\left(-\theta S(\boldsymbol{\beta}; \tau, h)\right).
\end{align*}
Since the smoothed loss function is non-negative, $S(\boldsymbol{\beta}; \tau, h) \ge 0$, which implies that the exponential term is bounded: $\exp(-\theta S(\boldsymbol{\beta}; \tau, h)) \le 1$.

Since the integrand is non-negative, by Tonelli's theorem, we can evaluate the integral with respect to $\boldsymbol{\beta}$ first. For any fixed variance parameter $\sigma_{\boldsymbol{\beta}}^2 > 0$:
\begin{align*}
    &\int_{\mathbb{R}^d} L(\boldsymbol{y} \mid \boldsymbol{\mathcal{X}}, \boldsymbol{\beta}, \theta; \tau, h) \pi(\boldsymbol{\beta}\mid \sigma_{\boldsymbol{\beta}}^2) \,\mathrm{d}\boldsymbol{\beta} \\
    ={}& \int_{\mathbb{R}^d} (Z(\theta, \tau, h))^{-n} \exp\left(-\theta S(\boldsymbol{\beta}; \tau, h)\right) \pi(\boldsymbol{\beta}\mid \sigma_{\boldsymbol{\beta}}^2) \,\mathrm{d}\boldsymbol{\beta} \\
    ={}& (Z(\theta, \tau, h))^{-n} \int_{\mathbb{R}^d} \exp\left(-\theta S(\boldsymbol{\beta}; \tau, h)\right) \pi(\boldsymbol{\beta}\mid \sigma_{\boldsymbol{\beta}}^2) \,\mathrm{d}\boldsymbol{\beta} \\
    \le{}& (Z(\theta, \tau, h))^{-n} \int_{\mathbb{R}^d} 1 \cdot \pi(\boldsymbol{\beta}\mid \sigma_{\boldsymbol{\beta}}^2) \,\mathrm{d}\boldsymbol{\beta} \\
    ={}& (Z(\theta, \tau, h))^{-n}.
\end{align*}
The last equality holds because the conditional Gaussian prior density $\pi(\boldsymbol{\beta}\mid \sigma_{\boldsymbol{\beta}}^2)$ is proper and integrates to 1 over $\mathbb{R}^d$. Now, we substitute this upper bound back into the expression for $J_3$ to integrate over $\sigma_{\boldsymbol{\beta}}^2$:
\begin{align*}
    J_{3} &\le \int_0^\infty (Z(\theta, \tau, h))^{-n} \pi(\sigma_{\boldsymbol{\beta}}^2\mid a_0, b_0) \,\mathrm{d}\sigma_{\boldsymbol{\beta}}^2 \\
    &= (Z(\theta, \tau, h))^{-n} \int_0^\infty \pi(\sigma_{\boldsymbol{\beta}}^2\mid a_0, b_0) \,\mathrm{d}\sigma_{\boldsymbol{\beta}}^2 \\
    &= (Z(\theta, \tau, h))^{-n} \cdot 1 < \infty.
\end{align*}
The integral over $\sigma_{\boldsymbol{\beta}}^2$ is exactly 1 because the Inverse-Gamma hyperprior density $\pi(\sigma_{\boldsymbol{\beta}}^2\mid a_0, b_0)$ is proper. Thus, the marginal posterior is proper for any fixed $\theta$.

\textit{Proof of Part (ii).} 
The propriety of the fully joint posterior requires the finiteness of the integral over all parameters $(\boldsymbol{\beta}, \sigma_{\boldsymbol{\beta}}^2, \theta)$. Let this total integral be $J_{total}$:
\begin{align*}
    J_{total} \coloneqq \int_0^\infty \left( \int_0^\infty \int_{\mathbb{R}^d} L(\boldsymbol{y} \mid \boldsymbol{\mathcal{X}}, \boldsymbol{\beta}, \theta; \tau, h) \pi(\boldsymbol{\beta}\mid \sigma_{\boldsymbol{\beta}}^2) \pi(\sigma_{\boldsymbol{\beta}}^2\mid a_0, b_0) \,\mathrm{d}\boldsymbol{\beta} \,\mathrm{d}\sigma_{\boldsymbol{\beta}}^2 \right) \pi(\theta) \,\mathrm{d}\theta.
\end{align*}
The term inside the parentheses corresponds exactly to $J_3$ from Part (i). Using the upper bound derived in Part (i), $J_3 \le (Z(\theta, \tau, h))^{-n}$, we have:
\begin{align*}
    J_{total} \le \int_0^\infty (Z(\theta, \tau, h))^{-n} \pi(\theta) \,\mathrm{d}\theta.
\end{align*}
Therefore, the joint posterior is proper if the integral $\int_0^\infty (Z(\theta, \tau, h))^{-n} \pi(\theta) \,\mathrm{d}\theta$ converges to a finite value. This condition relies on the asymptotic behavior of $(Z(\theta, \tau, h))^{-n}$. Recall from the proof of Theorem~\ref{thm:propriety1} that $(Z(\theta, \tau, h))^{-n} \propto \theta^{n/2} e^{n L_{\min} \theta}$ as $\theta \to \infty$ and behaves like $O(\theta^{k_Z})$ as $\theta \to 0$. In the current case with a proper prior for $\boldsymbol{\beta}$, the exponential growth from the integration over $\boldsymbol{\beta}$ (previously $e^{\theta C_S}$) is absent (effectively $C_S=0$). Furthermore, the singularity factor $\theta^{-d}$ associated with the improper prior is also absent. Thus, for the Gamma prior on $\theta$ to ensure convergence, we require $b_\theta > n L_{\min}$ and $a_\theta + k_Z > 0$.
\end{proof}

\subsection{Proofs for kernel selection theory}
\label{app:proofs_kernel}

\phantomsection
\begin{proof}[Proof of Proposition~\ref{prop:Z_properties}]\label{pf:prop1}
\textit{Proof of Part (i).} The derivative of $Z(\theta, \tau, h)$ with respect to $\theta$ is found by differentiating under the integral sign. This interchange of operations is justified by the Leibniz integral rule provided the partial derivative of the integrand is bounded in magnitude by an integrable function. For any $\theta$ in a compact interval $[a, b]$ with $0 < a < b < \infty$, we can establish such a bound:
$$ \left| \frac{\partial}{\partial \theta}\exp(-\theta L_h(\cdot; \tau)) \right| = L_h(\cdot; \tau)\exp(-\theta L_h(\cdot; \tau)) \le L_h(\cdot; \tau)\exp(-a L_h(\cdot; \tau)). $$
The dominating function on the right-hand side is integrable on $\mathbb{R}$ because, due to the exponential factor, it approaches zero as $|u| \to \infty$ at a rate faster than any inverse polynomial (e.g., faster than $1/u^2$). The conditions for the Leibniz rule are therefore met, and we can write:
$$ \frac{\partial Z(\theta, \tau, h)}{\partial \theta} = \int_{-\infty}^{\infty} -L_h(u; \tau) \exp\left( -\theta L_h(u; \tau) \right) \,\mathrm{d}u. $$
Since $\exp\left( -\theta L_h(u; \tau) \right) > 0$ and $L_h(u; \tau)$ is non-negative and not identically zero, the integrand is non-positive and strictly negative on a set of non-zero measure. Therefore, $\partial Z(\theta, \tau, h)/\partial \theta < 0$, establishing strict decrease.

\textit{Proof of Part (ii).} To show convexity of $\log Z(\theta, \tau, h)$, we examine its second derivative with respect to $\theta$. The first-order partial derivative of $\log Z(\theta, \tau, h)$ with respect to $\theta$ is:
$$ \frac{\partial \log Z(\theta, \tau, h)}{\partial \theta} = \frac{1}{Z(\theta, \tau, h)} \frac{\partial Z(\theta, \tau, h)}{\partial \theta} = \frac{\int_{-\infty}^{\infty} -L_h(u; \tau) e^{-\theta L_h(u; \tau)} \,\mathrm{d} u}{\int_{-\infty}^{\infty} e^{-\theta L_h(u; \tau)} \,\mathrm{d} u} = -\mathbb{E}_{p_u}[L_h(u; \tau)], $$
where $p_u(u \mid \theta) = \frac{\exp(-\theta L_h(u; \tau))}{Z(\theta, \tau, h)}$ can be interpreted as a probability density for $u$ parameterized by $\theta$. The same justification allows for a second differentiation under the integral sign, and we have
\begin{align*}
& \frac{\partial^2 \log Z(\theta, \tau, h)}{\partial \theta^2} \\
={}& \frac{\partial}{\partial \theta} \left( \frac{\int_{-\infty}^{\infty} -L_h(u; \tau) e^{-\theta L_h(u; \tau)} \,\mathrm{d} u}{Z(\theta, \tau, h)} \right) \\
={}& \frac{(\int_{-\infty}^{\infty} L_h^2(u; \tau) e^{-\theta L_h(u; \tau)} \,\mathrm{d} u)Z(\theta, \tau, h) - (\int_{-\infty}^{\infty} -L_h(u; \tau) e^{-\theta L_h(u; \tau)} \,\mathrm{d} u)(\int_{-\infty}^{\infty} -L_h(u; \tau) e^{-\theta L_h(u; \tau)} \,\mathrm{d} u)}{Z(\theta, \tau, h)^2} \\
={}& \int_{-\infty}^{\infty} L_h(u; \tau)^2 \frac{e^{-\theta L_h(u; \tau)}}{Z(\theta, \tau, h)} \mathrm{d}u - \left( \int_{-\infty}^{\infty} -L_h(u; \tau) \frac{e^{-\theta L_h(u; \tau)}}{Z(\theta, \tau, h)} \mathrm{d}u \right)^2 \\
={}& \mathbb{E}_{p_u}[L_h(u; \tau)^2] - (\mathbb{E}_{p_u}[L_h(u; \tau)])^2 = \text{Var}_{p_u}(L_h(u; \tau)).
\end{align*}
Since variance is always non-negative, we have $\frac{\partial^2 \log Z(\theta, \tau, h)}{\partial \theta^2} \ge 0$. This confirms the convexity of $\log Z(\theta, \tau, h)$ with respect to $\theta$.
\end{proof}

\phantomsection
\begin{proof}[Proof of Theorem~\ref{thm:4}]\label{pf:thm4}
The posterior distributions are given by:
\begin{align*}
\pi_{\mathrm{BSQR}}(\boldsymbol{\beta} \mid \theta, \boldsymbol{y}, \boldsymbol{\mathcal{X}};\tau, h) &= \frac{1}{Z_{\mathrm{BSQR}}(\theta, h, \pi)} \exp\left( -\theta \sum_{i=1}^n L_h(e_i(\boldsymbol{\beta}); \tau) \right) \pi(\boldsymbol{\beta}), \\
\pi_{\mathrm{ALD}}(\boldsymbol{\beta} \mid \theta, \boldsymbol{y}, \boldsymbol{\mathcal{X}};\tau) &= \frac{1}{Z_{\mathrm{ALD}}(\theta, \pi)} \exp\left( -\theta \sum_{i=1}^n \rho_\tau(e_i(\boldsymbol{\beta})) \right) \pi(\boldsymbol{\beta}),
\end{align*}
where $e_i(\boldsymbol{\beta}) = y_i - \boldsymbol{x}_i^\top\boldsymbol{\beta}$, and $Z_{\mathrm{BSQR}}$ and $Z_{\mathrm{ALD}}$ are the respective normalizing constants.
The ratio of these posteriors is
\begin{align*}
R(\boldsymbol{\beta}) &\coloneqq \frac{\pi_{\mathrm{BSQR}}(\boldsymbol{\beta} \mid \theta, h, \boldsymbol{y}, \boldsymbol{\mathcal{X}};\tau,h)}{\pi_{\mathrm{ALD}}(\boldsymbol{\beta} \mid \theta, \boldsymbol{y}, \boldsymbol{\mathcal{X}};\tau)} = \frac{Z_{\mathrm{ALD}}(\theta, \pi)}{Z_{\mathrm{BSQR}}(\theta, h, \pi)}\cdot \frac{\exp\left( -\theta \sum_{i=1}^n L_h(e_i(\boldsymbol{\beta}); \tau) \right) \pi(\boldsymbol{\beta})}{\exp\left( -\theta \sum_{i=1}^n \rho_\tau(e_i(\boldsymbol{\beta})) \right) \pi(\boldsymbol{\beta})} \\
&= C_0 \cdot \exp\left( -\theta \sum_{i=1}^n [L_h(e_i(\boldsymbol{\beta}); \tau) - \rho_\tau(e_i(\boldsymbol{\beta}))] \right),
\end{align*}
where $C_0 = Z_{\mathrm{ALD}}(\theta, \pi) / Z_{\mathrm{BSQR}}(\theta, h, \pi)$ is a positive constant independent of $\boldsymbol{\beta}$.
Let $\Delta_i(\boldsymbol{\beta}) = L_h(e_i(\boldsymbol{\beta}); \tau) - \rho_\tau(e_i(\boldsymbol{\beta}))$. Without loss of generality, we may re-center $K(v)$ to have zero mean, which implies the symmetry $K(v) = K(-v)$. If the initial compact support is $[a,b]$, this centering implies an effective symmetric support $[-c, c]$ where $c = (b-a)/2 > 0$, with $K(v) = 0$ for $|v| > c$. Then, from Eq.~\eqref{eq:smoothed_check_loss} with a change of variables $u = v/h$, we have $L_h(e; \tau) = \int_{-c}^{c} \rho_{\tau}(hu + e) K(u) \,\mathrm{d}u$.
The derivative of $L_h(e; \tau) - \rho_\tau(e)$ with respect to $e$ is
$$ \frac{\mathrm{d}}{\mathrm{d} e} [L_h(e; \tau) - \rho_\tau(e)] = \int_{-c}^{c} \psi_{\tau}(hu + e) K(u) \,\mathrm{d}u - \psi_\tau(e), $$
where $\psi_\tau(x) = \tau - \mathbb{I}(x<0)$.

We first examine the case where $e > hc$ (equivalently, $e/h > c$). In this regime, for any $u \in [-c, c]$, the argument $hu+e$ satisfies $hu+e > -hc+e > -hc+hc = 0$. Consequently, $\psi_{\tau}(hu + e) = \tau$. Furthermore, since $e > hc > 0$, we also have $\psi_\tau(e) = \tau$.
The derivative of $L_h(e; \tau) - \rho_\tau(e)$ with respect to $e$ is therefore:
\begin{align*}
\frac{\mathrm{d}}{\mathrm{d} e} [L_h(e; \tau) - \rho_\tau(e)] = \int_{-c}^{c} \tau K(u) \,\mathrm{d}u - \tau = \tau \int_{-c}^{c} K(u) \,\mathrm{d}u - \tau = 0.
\end{align*}
This implies that $L_h(e; \tau) - \rho_\tau(e)$ remains constant for $e > hc$.
To determine this constant value, we directly evaluate the difference:
$$ L_h(e; \tau) - \rho_\tau(e) = \int_{-c}^{c} [\rho_{\tau}(hu + e) - \rho_\tau(e)] K(u) \,\mathrm{d}u. $$
Given $e > hc$, both $hu+e$ and $e$ are positive. Hence, $\rho_{\tau}(x) = \tau x$ for $x>0$.
The integrand simplifies to $\tau(hu+e) - \tau e = \tau hu$.
The integral then becomes
$$ \int_{-c}^{c} \tau hu K(u) \,\mathrm{d}u = \tau h \int_{-c}^{c} u K(u) \,\mathrm{d}u = 0, $$
since $\int_{-c}^{c} u K(u) \,\mathrm{d}u = 0$ since $K(u)$ is symmetric about 0 on its support $[-c,c]$.
Therefore, for $e > hc$, it follows that $L_h(e; \tau) = \rho_\tau(e)$.

Next, we consider the case where $e < -hc$ (equivalently, $e/h < -c$). In this regime, for any $u \in [-c, c]$, the argument $hu+e$ satisfies $hu+e < hc+e < hc-hc = 0$.
Consequently, $\psi_{\tau}(hu + e) = \tau-1$.
Furthermore, since $e < -hc < 0$, we have $\psi_\tau(e) = \tau-1$.
The derivative of $L_h(e; \tau) - \rho_\tau(e)$ with respect to $e$ is:
\begin{align*}
\frac{\mathrm{d}}{\mathrm{d} e} [L_h(e; \tau) - \rho_\tau(e)] &= \int_{-c}^{c} (\tau-1) K(u) \,\mathrm{d}u - (\tau-1) = (\tau-1) \int_{-c}^{c} K(u) \,\mathrm{d}u - (\tau-1) = 0.
\end{align*}
This implies that $L_h(e; \tau) - \rho_\tau(e)$ is also constant for $e < -hc$.
To determine this constant, we evaluate the difference:
$$ L_h(e; \tau) - \rho_\tau(e) = \int_{-c}^{c} [\rho_{\tau}(hu + e) - \rho_\tau(e)] K(u) \,\mathrm{d}u. $$
Given $e < -hc$, both $hu+e$ and $e$ are negative. Hence, $\rho_{\tau}(x) = (\tau-1)x$ for $x<0$.
The integrand simplifies to $(\tau-1)(hu+e) - (\tau-1)e = (\tau-1)hu$.
The integral then becomes
$$ \int_{-c}^{c} (\tau-1) hu K(u) \,\mathrm{d}u = (\tau-1) h \int_{-c}^{c} u K(u) \,\mathrm{d}u = 0.$$
Therefore, for $e < -hc$, it follows that $L_h(e; \tau) = \rho_\tau(e)$.

Based on the previous case analysis, we conclude that for residuals $e_i(\boldsymbol{\beta})$ satisfying $|e_i(\boldsymbol{\beta})/h| > c$, $\Delta_i(\boldsymbol{\beta}) = 0$.
For residuals $e_i(\boldsymbol{\beta})$ such that $|e_i(\boldsymbol{\beta})/h| \le c$, $e_i(\boldsymbol{\beta})$ lies in the compact interval $[-hc, hc]$.
Within this interval, $L_h(u; \tau)$ and $\rho_\tau(u)$ are continuous functions of $u$ on $[-hc, hc]$.
The continuity of $L_h(u; \tau)$ on $[-hc, hc]$ follows from the continuity of $\rho_\tau(u)$ and the properties of convolution with a well-behaved kernel $K(u)$.
Thus, $\Delta_i(\boldsymbol{\beta})$ is also a continuous function of $e_i(\boldsymbol{\beta})$ over this interval.
By the Extreme Value Theorem, there exist constants $m_{\text{in}}(h,\tau,c)$ and $M_{\text{in}}(h,\tau,c)$ such that for all $u \in [-hc, hc]$, $m_{\text{in}}(h,\tau,c) \le \Delta_i(\boldsymbol{\beta}) \le M_{\text{in}}(h,\tau,c)$. Let $m_L = \min(0, m_{\text{in}}(h,\tau,c))$ and $M_U = \max(0, M_{\text{in}}(h,\tau,c))$.
It follows that $m_L \le \Delta_i(\boldsymbol{\beta}) \le M_U$ for any $e_i(\boldsymbol{\beta})$.
Consequently, $\sum_{i=1}^n \Delta_i(\boldsymbol{\beta})$ is bounded as:
$$ n \cdot m_L \le \sum_{i=1}^n \Delta_i(\boldsymbol{\beta}) \le n \cdot M_U. $$
The exponential term $\exp\left( -\theta \sum_{i=1}^n \Delta_i(\boldsymbol{\beta}) \right)$ is therefore bounded below by $\exp(-\theta n M_U)$ and above by $\exp(-\theta n m_L)$.
Thus, the ratio $R(\boldsymbol{\beta})$ satisfies:
$$ C_0 \exp(-\theta n M_U) \le R(\boldsymbol{\beta}) \le C_0 \exp(-\theta n m_L). $$
Denoting these lower and upper bounds as $M_1 = C_0 \exp(-\theta n M_U)$ and $M_2 = C_0 \exp(-\theta n m_L)$ respectively, we observe that $M_1$ and $M_2$ are positive constants independent of $\boldsymbol{\beta}$. This establishes the desired result.
\end{proof}

\phantomsection
\begin{proof}[Proof of Theorem~\ref{thm:5}]\label{pf:thm5}
We evaluate the Hessian $\boldsymbol{H}_L(\boldsymbol{\beta}; K)$ at the true parameter $\boldsymbol{\beta}_0(\tau)$. At this point, the residuals become the true errors, $e_i(\boldsymbol{\beta}_0(\tau)) = \varepsilon_{0i}$, and from Eq.~\eqref{eq:hessian_likelihood_part}, we have
$$ \boldsymbol{H}_L(\boldsymbol{\beta}_0(\tau); K) = \frac{\theta}{h} \sum_{i=1}^n K\left(\frac{\varepsilon_{0i}}{h}\right) \boldsymbol{x}_i \boldsymbol{x}_i^\top. $$
Then, the expected Hessian is
\begin{align*}
\mathcal{H}_L(\boldsymbol{\beta}_0(\tau); K) &= \mathbb{E}_{\varepsilon_{0i} \sim f_{\varepsilon_0}, \boldsymbol{x}_i \sim P_{\boldsymbol{X}}} \left[ \frac{\theta}{h} \sum_{i=1}^n K\left(\frac{\varepsilon_{0i}}{h}\right) \boldsymbol{x}_i \boldsymbol{x}_i^\top \right] \\
&= \frac{\theta}{h} \sum_{i=1}^n \mathbb{E}_{\varepsilon_{0i} \sim f_{\varepsilon_0}, \boldsymbol{x}_i \sim P_{\boldsymbol{X}}} \left[ K\left(\frac{\varepsilon_{0i}}{h}\right) \boldsymbol{x}_i \boldsymbol{x}_i^\top \right] \\
&=  \frac{\theta}{h} \sum_{i=1}^n s_K(h) \cdot \mathbb{E}[\boldsymbol{x}_i \boldsymbol{x}_i^\top] = \frac{\theta}{h} n s_K(h) \Sigma_X
\end{align*}
by using the independence of $\varepsilon_{0i}$ and $\boldsymbol{x}_i$. The relationship between $\mathcal{H}_L(\boldsymbol{\beta}_0(\tau); K_A)$ and $\mathcal{H}_L(\boldsymbol{\beta}_0(\tau); K_B)$ is determined by examining their difference
\begin{align*}
\mathcal{H}_L(\boldsymbol{\beta}_0(\tau); K_A) - \mathcal{H}_L(\boldsymbol{\beta}_0(\tau); K_B) &= \frac{\theta}{h} n (s_{K_A}(h) - s_{K_B}(h)) \Sigma_X.
\end{align*}
The scalar term $c \coloneqq \frac{\theta}{h} n (s_{K_A}(h) - s_{K_B}(h)) > 0$, since $s_{K_A}(h) - s_{K_B}(h) > 0$ by assumption, and $\theta > 0, h > 0$. Since the matrix $\Sigma_X$ is positive definite by assumption, the difference is positive definite (strict Loewner order). 

For the sample Hessian, the difference is $\theta/h \sum_{i=1}^n \left( K_A\left(\frac{\varepsilon_{0i}}{h}\right) - K_B\left(\frac{\varepsilon_{0i}}{h}\right) \right) \boldsymbol{x}_i \boldsymbol{x}_i^\top$, which is a sum of random matrices. By the law of large numbers, the normalized version $1/n ( \boldsymbol{H}_L(\boldsymbol{\beta}_0(\tau); K_A) \allowbreak - \allowbreak \boldsymbol{H}_L(\boldsymbol{\beta}_0(\tau); \allowbreak K_B)) \allowbreak \xrightarrow{P} \theta/h (s_{K_A}(h) - s_{K_B}(h)) \Sigma_X \succ \boldsymbol{0}$ as $n \to \infty$, ensuring the strict inequality holds asymptotically with probability approaching 1.
\end{proof}

\clearpage
\phantomsection
\addcontentsline{toc}{section}{Appendix B: Algorithm}
\section*{Appendix B: Algorithm}
\label{sec:AppendixB}
\setcounter{equation}{0}
\renewcommand{\theequation}{B.\arabic{equation}}
\setcounter{subsection}{0}
\renewcommand{\thesubsection}{B.\arabic{subsection}}
\renewcommand{\theHsubsection}{B.\arabic{subsection}}
\renewcommand{\thealgocf}{B}

\begin{algorithm}[H]
    \caption{Joint Hamiltonian Monte Carlo sampling for $(\boldsymbol{\beta}, \theta)$ via NUTS \textbf{(Part I)}}
    \label{Al1_Part1}
    
    \KwIn{Data $(\boldsymbol{y}, \boldsymbol{\mathcal{X}})$, quantile level $\tau$, bandwidth $h$; 
          prior specification $\pi(\boldsymbol{\beta})$ and $\pi(\theta)$; 
          initial state $\boldsymbol{\beta}^{(0)}$ and $\theta^{(0)}$; 
          number of MCMC iterations $N$; mass matrix $\boldsymbol{M}$ (adapted during warmup).}
    \KwOut{Markov chain samples $\{(\boldsymbol{\beta}^{(1)}, \theta^{(1)}), \dots, (\boldsymbol{\beta}^{(N)}, \theta^{(N)})\}$ from the joint posterior $\pi(\boldsymbol{\beta}, \theta \mid \boldsymbol{y}, \boldsymbol{\mathcal{X}})$.}
    \BlankLine

    \For{$k=1,\ldots,N$}{
        Let current state $\boldsymbol{\beta}_c = \boldsymbol{\beta}^{(k-1)}$ and $\theta_c = \theta^{(k-1)}$\;
        Define unconstrained joint state $\boldsymbol{q}_c = (\boldsymbol{\beta}_c^\top, \log \theta_c)^\top$\;
        
        \textbf{Step 1 (Momentum Sampling):}\\
        Draw auxiliary momentum $\boldsymbol{p}_{c} \sim \mathcal{N} (\boldsymbol{0}, \boldsymbol{M})$\;
        Initialize trajectory variables: $\tilde{\boldsymbol{q}} \leftarrow \boldsymbol{q}_{c}$ and $\tilde{\boldsymbol{p}} \leftarrow \boldsymbol{p}_{c}$\;
        
        \textbf{Step 2 (Hamiltonian Dynamics via Leapfrog Integration):}\\
        \tcp{NUTS builds a trajectory using multiple Leapfrog steps. A single step of size $\epsilon$ updates $(\tilde{\boldsymbol{q}}, \tilde{\boldsymbol{p}})$ as follows:}
        
        \quad 1. \textbf{Half-step momentum update:}\\
        \quad \quad Compute gradient $\nabla_{\boldsymbol{q}} U(\tilde{\boldsymbol{q}} \mid \boldsymbol{y}, \boldsymbol{\mathcal{X}})$ (see details in Part II)\;
        \quad \quad $\tilde{\boldsymbol{p}} \leftarrow \tilde{\boldsymbol{p}} - \frac{\epsilon}{2} \nabla_{\boldsymbol{q}} U(\tilde{\boldsymbol{q}} \mid \boldsymbol{y}, \boldsymbol{\mathcal{X}})$\;
        
        \quad 2. \textbf{Full-step position update:}\\
        \quad \quad $\tilde{\boldsymbol{q}} \leftarrow \tilde{\boldsymbol{q}} + \epsilon \boldsymbol{M}^{-1} \tilde{\boldsymbol{p}}$\;
        \quad \quad \emph{Implicitly updates parameters:} $\tilde{\boldsymbol{\beta}} \leftarrow \tilde{\boldsymbol{q}}_{1:d}, \tilde{\theta} \leftarrow \exp(\tilde{\boldsymbol{q}}_{d+1})$\;
        
        \tcp{\textbf{... Step 2 continues in Part II ...}}
    }
\end{algorithm}

\begin{algorithm}[H]
    \addtocounter{algocf}{-1} 
    \caption{Joint Hamiltonian Monte Carlo sampling for $(\boldsymbol{\beta}, \theta)$ via NUTS \textbf{(Part II)}}
    \label{Al1_Part2}

    \For{$k=1,\ldots,N$ \textbf{(continued)}}{
        \textbf{Step 2 (continued):}
        
        \quad 3. \textbf{Half-step momentum update:}\\
        \quad \quad Compute gradient $\nabla_{\boldsymbol{q}} U(\tilde{\boldsymbol{q}} \mid \boldsymbol{y}, \boldsymbol{\mathcal{X}})$ at new position\;
        \quad \quad $\tilde{\boldsymbol{p}} \leftarrow \tilde{\boldsymbol{p}} - \frac{\epsilon}{2} \nabla_{\boldsymbol{q}} U(\tilde{\boldsymbol{q}} \mid \boldsymbol{y}, \boldsymbol{\mathcal{X}})$\;
        
        \vspace{0.2em}
        \tcp{Gradient Computation Details:}
        
        \quad \textbf{-- For $\boldsymbol{\beta}$ block (Analytical Score):}
        \begin{equation*}
            \nabla_{\boldsymbol{\beta}} U(\boldsymbol{\beta}, \theta \mid \boldsymbol{y}, \boldsymbol{\mathcal{X}}) = -\theta \sum_{i=1}^n \Psi_h(y_i - \boldsymbol{x}_i^\top \boldsymbol{\beta}; \tau) \boldsymbol{x}_i - \nabla_{\boldsymbol{\beta}} \log \pi(\boldsymbol{\beta})
        \end{equation*}
        
        \quad \textbf{-- For $\log\theta$ block (Chain Rule \& Auto-Diff):}
        \begin{equation*}
            \nabla_{\log\theta} U(\boldsymbol{q} \mid \boldsymbol{y}, \boldsymbol{\mathcal{X}}) = \theta \cdot \left( \sum_{i=1}^n L_h(y_i - \boldsymbol{x}_i^\top \boldsymbol{\beta}; \tau) + n \nabla_{\theta} \log Z(\theta, \tau, h) - \nabla_{\theta} \log \pi(\theta) \right) - 1
        \end{equation*}        
        \vspace{0.5em}
        \textbf{Step 3 (NUTS Trajectory Selection):}\\
        \quad Extend the trajectory forward and backward (doubling) until the No-U-Turn condition is met\;
        \quad Select a candidate state $\boldsymbol{q}' = (\boldsymbol{\beta}'^\top, \log \theta')^\top$ from the trajectory via multinomial sampling weighted by $\exp(-H(\boldsymbol{q}, \boldsymbol{p}))$\;
        
        \quad Set $\boldsymbol{q}^{(k)} \leftarrow \boldsymbol{q}'$ and extract: $\boldsymbol{\beta}^{(k)} = \boldsymbol{q}'_{1:d}$, $\theta^{(k)} = \exp(\boldsymbol{q}'_{d+1})$\;
    }
\end{algorithm}

\newpage
\phantomsection
\addcontentsline{toc}{section}{Appendix C: Derivations}
\section*{Appendix C: Derivations}
\label{sec:AppendixC}
\setcounter{equation}{0}
\renewcommand{\theequation}{C.\arabic{equation}}
\setcounter{subsection}{0}
\renewcommand{\thesubsection}{C.\arabic{subsection}}
\renewcommand{\theHsubsection}{C.\arabic{subsection}}

\subsection{Gaussian kernel}
The standard Gaussian kernel is $K(v) = \phi(v) = (2\pi)^{-1/2} \exp(-v^2/2)$, with CDF $F_K(u) = \Phi(u)$. Applying Eq.~\eqref{eq:Psi_h_cdf_form} yields the smoothed score function:
\begin{equation*}
\Psi_h(e; \tau) = \Phi\left(\frac{e}{h}\right) - (1-\tau).
\end{equation*}
The smoothed loss function $L_h(e; \tau)$ is defined such that its derivative is the smoothed score, $L_h'(e; \tau) = \Psi_h(e; \tau)$. We find its form by computing the indefinite integral of $\Psi_h(e; \tau)$ and determining the constant of integration.
\begin{align*}
L_h(e; \tau) &= \int \Psi_h(e; \tau) \, \mathrm{d} e = \int \left( \Phi\left(\frac{e}{h}\right) - (1-\tau) \right) \, \mathrm{d} e \\
&= \int \Phi\left(\frac{e}{h}\right) \, \mathrm{d} e - (1-\tau)e + C_{G1}.
\end{align*}
The first integral, $\int \Phi(e/h) \, \mathrm{d} e$, is evaluated using integration by parts ($\int u \, \mathrm{d} v = uv - \int v \, \mathrm{d} u$). Let $u = \Phi(e/h)$ and $\mathrm{d} v = \mathrm{d} e$ (so $v = e$), then $\mathrm{d} u = \frac{1}{h}\phi(e/h) \, \mathrm{d} e$. This gives:
\begin{align*}
\int \Phi\left(\frac{e}{h}\right) \, \mathrm{d} e &= e \Phi\left(\frac{e}{h}\right) - \int e \cdot \frac{1}{h}\phi\left(\frac{e}{h}\right) \, \mathrm{d} e.
\end{align*}
The remaining integral is solved by substitution. Let $w = -e^2/(2h^2)$, then $\mathrm{d} w = -e/h^2 \, \mathrm{d} e$, which implies $e \, \mathrm{d} e = -h^2 \, \mathrm{d} w$. Thus,
\begin{align*}
	\int e \phi\left(\frac{e}{h}\right) \, \mathrm{d} e &= \int \frac{1}{\sqrt{2\pi}} \exp\left(-\frac{e^2}{2h^2}\right) e \, \mathrm{d} e = \frac{1}{\sqrt{2\pi}} \int \exp(w) (-h^2 \, \mathrm{d} w) \\
&= -h^2 \frac{1}{\sqrt{2\pi}} \int \exp(w) \, \mathrm{d} w = -h^2 \frac{1}{\sqrt{2\pi}} \exp(w) + C_{G2} = -h^2 \phi\left(\frac{e}{h}\right) + C_{G2}.
\end{align*}
Substituting this back and combining constants (letting $C_{G} = C_{G1} - C_{G2}$), we obtain the general form of the smoothed loss:
\begin{align*}
L_h(e; \tau) &= e \Phi\left(\frac{e}{h}\right) - \frac{1}{h} \left( -h^2 \phi\left(\frac{e}{h}\right) \right) - (1-\tau)e + C_{G} \\
&= e \left( \Phi\left(\frac{e}{h}\right) - (1-\tau) \right) + h \phi\left(\frac{e}{h}\right) + C_{G}.
\end{align*}
The constant of integration $C_{G}$ is determined by the condition that $L_h(e; \tau) \to \rho_\tau(e)$ as $h \to 0$. For any $e > 0$, we require $\lim_{h \to 0} L_h(e; \tau) = \rho_\tau(e) = e\tau$. As $h \to 0$, $\Phi(e/h) \to 1$ and $h\phi(e/h) \to 0$. The limit becomes:
$$ e \left( 1 - (1-\tau) \right) + 0 + C_{G} = e\tau \implies e\tau + C_{G} = e\tau, $$
which implies $C_{G}=0$. A similar analysis for $e < 0$ also yields $C_{G}=0$. Therefore, the smoothed loss function is:
\begin{equation*}
L_h(e; \tau) = e \left( \Phi\left(\frac{e}{h}\right) - (1-\tau) \right) + h \phi\left(\frac{e}{h}\right).
\end{equation*}
This derived form aligns with established results in the literature (e.g., \shortciteNP{Horowitz1998, Koenker2005}).

\subsection{Uniform kernel}

The standard Uniform kernel is defined as $K(v) = 1/2$ for $v \in [-1, 1]$ and $K(v) = 0$ otherwise. Its corresponding CDF, $F_K(u)$, is given by:
\begin{equation*}
F_K(u) = \begin{cases} 0 & \text{if } u < -1, \\ \frac{u+1}{2} & \text{if } -1 \le u \le 1, \\ 1 & \text{if } u > 1. \end{cases}
\end{equation*}
Substituting this CDF into the general expression for $\Psi_h(e; \tau)$ in Eq.~\eqref{eq:Psi_h_cdf_form}, we obtain a piecewise function for $\Psi_h(e; \tau)$ with the Uniform kernel:
\begin{equation*}
\Psi_h(e; \tau) = \begin{cases} -(1-\tau) & \text{if } e/h < -1, \\ \frac{e}{2h} + \tau - \frac{1}{2} & \text{if } -1 \le e/h \le 1, \\ \tau & \text{if } e/h > 1. \end{cases}
\end{equation*}
The smoothed loss function $L_h(e; \tau)$ is constructed to match the pinball loss $\rho_\tau(e)$ in the outer regions where $|e/h| \ge 1$. Specifically, $L_h(e; \tau) = e(\tau-1)$ for $e/h \le -1$ (i.e., $e \le -h$), and $L_h(e; \tau) = e\tau$ for $e/h \ge 1$ (i.e., $e \ge h$).

In the central region, where $-1 < e/h < 1$ (i.e., $-h < e < h$), $L_h(e; \tau)$ is obtained by integrating the corresponding segment of $\Psi_h(e; \tau)$:
\begin{align*}
L_h(e; \tau) &= \int \left(\frac{e}{2h} + \tau - \frac{1}{2}\right) \, \mathrm{d} e = \frac{e^2}{4h} + e\left(\tau - \frac{1}{2}\right) + C_U,
\end{align*}
where $C_U$ is the constant of integration for this segment.
To ensure continuity of $L_h(e; \tau)$ at the boundary $e=-h$, the value of the central segment's expression at $e=-h$ must equal the value from the left segment, which is $\rho_\tau(-h) = (-h)(\tau-1) = h(1-\tau)$. Setting these equal allows us to solve for $C_U$:
$(-h)^2/4h + (-h)\left(\tau - \frac{1}{2}\right) + C_U = h(1-\tau)$, we have $C_U = h/4$. Substituting it into the expression for the central segment, the complete piecewise definition of the smoothed loss function $L_h(e; \tau)$ for the Uniform kernel is:
\begin{equation*}
L_h(e; \tau) = \begin{cases} e(\tau-1) & \text{if } e/h \le -1, \\ \frac{e^2}{4h} + e\left(\tau - \frac{1}{2}\right) + \frac{h}{4} & \text{if } -1 < e/h < 1, \\ e\tau & \text{if } e/h \ge 1. \end{cases}
\end{equation*}
This construction ensures $L_h(e; \tau)$ is continuous at $e=-h$. We now verify continuity at the other boundary, $e=h$. For the central segment evaluated at $e=h$, we have
\begin{align*}
L_h(h;\tau) = \frac{h^2}{4h} + h\left(\tau - \frac{1}{2}\right) + \frac{h}{4} = h\tau.
\end{align*}
This value matches the expression for $L_h(e; \tau)$ in the region $e/h \ge 1$ when $e=h$, i.e., $L_h(h;\tau) = h\tau = \rho_\tau(h)$. Thus, continuity is also satisfied at $e=h$, and the derived $L_h(e; \tau)$ is continuous across all regions.

\subsection{Epanechnikov kernel}

The standard Epanechnikov kernel is defined by $K(v) = \frac{3}{4}(1-v^2)$ for $v \in [-1, 1]$ and $K(v) = 0$ otherwise.
Its CDF, $F_K(u)$, is piecewise:
\begin{equation*}
F_K(u) = \begin{cases} 0 & \text{if } u < -1, \\ \frac{3}{4}u - \frac{1}{4}u^3 + \frac{1}{2} & \text{if } -1 \le u \le 1, \\ 1 & \text{if } u > 1. \end{cases}
\end{equation*}
Using this CDF in conjunction with Eq.~\eqref{eq:Psi_h_cdf_form}, the expression for $\Psi_h(e; \tau)$ under the Epanechnikov kernel becomes:
\begin{equation*}
\Psi_h(e; \tau) = \begin{cases} -(1-\tau) & \text{if } e/h < -1, \\ \frac{3}{4}\left(\frac{e}{h}\right) - \frac{1}{4}\left(\frac{e}{h}\right)^3 + \tau - \frac{1}{2} & \text{if } -1 \le e/h \le 1, \\ \tau & \text{if } e/h > 1. \end{cases}
\end{equation*}
As with other kernels, the smoothed loss function $L_h(e; \tau)$ is designed to replicate the pinball loss $\rho_\tau(e)$ in the regions where $|e/h| \ge 1$. Thus, $L_h(e; \tau) = e(\tau-1)$ for $e/h \le -1$, and $L_h(e; \tau) = e\tau$ for $e/h \ge 1$.

For the central region, $-1 < e/h < 1$ (i.e., $-h < e < h$), $L_h(e; \tau)$ is derived by integrating the relevant segment of $\Psi_h(e; \tau)$:
\begin{align*}
L_h(e; \tau) &= \int \left(\frac{3e}{4h} - \frac{e^3}{4h^3} + \tau - \frac{1}{2}\right) \, \mathrm{d} e = \frac{3e^2}{8h} - \frac{e^4}{16h^3} + e\left(\tau - \frac{1}{2}\right) + C_E,
\end{align*}
where $C_E$ is the integration constant for this segment.
Continuity at the boundary $e=-h$ dictates that the value of this central expression at $e=-h$ must be equal to $\rho_\tau(-h) = (-h)(\tau-1) = h(1-\tau)$. This condition allows for the determination of $C_E$:
\begin{align*}
\frac{3(-h)^2}{8h} - \frac{(-h)^4}{16h^3} + (-h)\left(\tau - \frac{1}{2}\right) + C_E &= h(1-\tau),
\end{align*}
which implies that $C_E = 3h/16$. Thus, the complete piecewise definition for $L_h(e; \tau)$ using the Epanechnikov kernel is:
\begin{equation*}
L_h(e; \tau) = \begin{cases} e(\tau-1) & \text{if } e/h \le -1, \\ \frac{3e^2}{8h} - \frac{e^4}{16h^3} + e\left(\tau - \frac{1}{2}\right) + \frac{3h}{16} & \text{if } -1 < e/h < 1, \\ e\tau & \text{if } e/h \ge 1. \end{cases}
\end{equation*}
This form of $L_h(e; \tau)$ is continuous at $e=-h$ by construction. To confirm overall continuity, we check the boundary $e=h$. Evaluating the central segment at $e=h$:
\begin{align*}
L_h(h;\tau) &= \frac{3h^2}{8h} - \frac{h^4}{16h^3} + h\left(\tau - \frac{1}{2}\right) + \frac{3h}{16} = h\tau.
\end{align*}
This result, $h\tau$, matches the value of $L_h(e; \tau)$ for $e/h \ge 1$ at $e=h$ (i.e., $\rho_\tau(h)$). Therefore, the smoothed loss function $L_h(e; \tau)$ is continuous across all defined regions.

\subsection{Triangular kernel}

The standard Triangular kernel is given by $K(v) = 1-|v|$ for $v \in [-1, 1]$ and $K(v) = 0$ otherwise.
Its CDF, $F_K(u)$, is defined piecewise:
\begin{equation*}
F_K(u) = \begin{cases} 0 & \text{if } u < -1, \\
\frac{1}{2}(1+u)^2 & \text{if } -1 \le u < 0, \\
1 - \frac{1}{2}(1-u)^2 & \text{if } 0 \le u \le 1, \\
1 & \text{if } u > 1.
\end{cases}
\end{equation*}
Substituting this CDF into Eq.~\eqref{eq:Psi_h_cdf_form}, we derive the corresponding $\Psi_h(e; \tau)$:
\begin{equation}
\Psi_h(e; \tau) = \begin{cases} -(1-\tau) & \text{if } e/h < -1, \\
\frac{1}{2}\left(1+\frac{e}{h}\right)^2 - (1-\tau) & \text{if } -1 \le e/h < 0, \\
\tau - \frac{1}{2}\left(1-\frac{e}{h}\right)^2 & \text{if } 0 \le e/h \le 1, \\
\tau & \text{if } e/h > 1.
\end{cases} \label{eq:psi_h_triangular_detailed}
\end{equation}
The smoothed loss function $L_h(e; \tau)$ is constructed to match the pinball loss $\rho_\tau(e)$ in the outer regions where $|e/h| \ge 1$. Thus, $L_h(e; \tau) = e(\tau-1)$ for $e/h \le -1$, and $L_h(e; \tau) = e\tau$ for $e/h \ge 1$.

Due to the piecewise nature of $F_K(u)$ for the Triangular kernel around $u=0$, the central region for $L_h(e; \tau)$ ($-h < e < h$) is split into two segments. For the segment $-1 < e/h < 0$ (i.e., $-h < e < 0$), integrating the corresponding part of $\Psi_h(e; \tau)$ yields:
\begin{equation*}
L_h(e; \tau) = \int \left( \frac{1}{2}\left(1+\frac{e}{h}\right)^2 - (1-\tau) \right) \, \mathrm{d} e = \frac{h}{6}\left(1+\frac{e}{h}\right)^3 - e(1-\tau) + C_{T1}.
\end{equation*}
The constant $C_{T1}$ is determined by ensuring continuity with $\rho_\tau(e)$ at $e=-h$. Setting $L_h(-h;\tau) = \rho_\tau(-h) = h(1-\tau)$ results in $C_{T1}=0$.
Thus, for $-1 < e/h < 0$:
\begin{equation*}
L_h(e; \tau) = \frac{h}{6}\left(1+\frac{e}{h}\right)^3 - e(1-\tau).
\end{equation*}
For the segment $0 \le e/h < 1$ (i.e., $0 \le e < h$), integration gives:
\begin{equation*}
L_h(e; \tau) = \int \left( \tau - \frac{1}{2}\left(1-\frac{e}{h}\right)^2 \right) \, \mathrm{d} e = e\tau + \frac{h}{6}\left(1-\frac{e}{h}\right)^3 + C_{T2}.
\end{equation*}
Similarly, $C_{T2}$ is found by imposing continuity with $\rho_\tau(e)$ at $e=h$. The condition $L_h(h;\tau) = \rho_\tau(h) = h\tau$ yields $C_{T2}=0$.
Thus, for $0 \le e/h < 1$:
\begin{equation*}
L_h(e; \tau) = e\tau + \frac{h}{6}\left(1-\frac{e}{h}\right)^3.
\end{equation*}

With $C_{T1}=0$ and $C_{T2}=0$, we must verify the continuity of $L_h(e; \tau)$ at the internal boundary $e=0$: (1) As $e/h \to 0^-$, $L_h(0;\tau) = \frac{h}{6}(1+0)^3 - 0 \cdot (1-\tau) = h/6$; (2) As $e/h \to 0^+$, $L_h(0;\tau) = 0 \cdot \tau + \frac{h}{6}(1-0)^3 = h/6$. Since the limits from both sides are equal, $L_h(e; \tau)$ is continuous at $e=0$. Combining all segments, the complete expression for $L_h(e; \tau)$ with the Triangular kernel is:
\begin{equation*}
L_h(e; \tau) = \begin{cases} e(\tau-1) & \text{if } e/h \le -1, \\
\frac{h}{6}\left(1+\frac{e}{h}\right)^3 - e(1-\tau) & \text{if } -1 < e/h < 0, \\
e\tau + \frac{h}{6}\left(1-\frac{e}{h}\right)^3 & \text{if } 0 \le e/h < 1, \\
e\tau & \text{if } e/h \ge 1.
\end{cases}
\end{equation*}
This smoothed loss function $L_h(e; \tau)$ is continuous across all regions. Furthermore, owing to the continuity of the Triangular kernel $K(v)$ itself, the derivative function $\Psi_h(e;\tau)$ given in Eq.~\eqref{eq:psi_h_triangular_detailed} is also continuous for all $e$.

\phantomsection
\addcontentsline{toc}{section}{Appendix D: Simulation results}

\setcounter{section}{4}
\setcounter{subsection}{0}
\setcounter{table}{0}
\setcounter{equation}{0}

\renewcommand{\thesection}{D}
\renewcommand{\thesubsection}{D.\arabic{subsection}}
\renewcommand{\theHsubsection}{D.\arabic{subsection}}
\renewcommand{\thetable}{D.\arabic{table}}
\renewcommand{\theequation}{D.\arabic{equation}}

\sisetup{
    table-parse-only = false,
    round-mode = places
}

\begin{landscape}
\section*{Appendix D: Simulation results}
\label{sec:AppendixD}

\subsection{General performance results}
{\footnotesize
\setlength{\tabcolsep}{4pt}
\begin{longtable}{@{}
    l 
    l 
    S[table-format=1.2, round-precision=2] 
    l 
    S[table-format=1.4, table-text-alignment=center, table-number-alignment=center, round-precision=4] 
    S[table-format=1.4, table-text-alignment=center, table-number-alignment=center, round-precision=4] 
    S[table-format=1.4, table-text-alignment=center, table-number-alignment=center, round-precision=4] 
    S[table-format=1.4, table-text-alignment=center, table-number-alignment=center, round-precision=4] 
    S[table-format=0.4, table-text-alignment=center, table-number-alignment=center, round-precision=4] 
    S[table-format=1.4, table-text-alignment=center, table-number-alignment=center, round-precision=4] 
    S[table-format=3.2, table-text-alignment=center, table-number-alignment=center, round-precision=2] 
    S[table-format=1.4, table-text-alignment=center, table-number-alignment=center, round-precision=4] 
    S[table-format=4.1, table-text-alignment=center, table-number-alignment=center, round-precision=1] 
    @{}}

\caption{Full simulation results comparing StdQR, BQR-ALD, and our proposed BSQR with Gaussian (G), Uniform (U), Epanechnikov (E), and Triangular (T) kernels.}
    
\label{tab:sim_results_all_kernels} \\

\toprule
\multirow{2}{*}{\textbf{Error Dist.}} & \multirow{2}{*}{\textbf{Design}} & {\multirow{2}{*}{$\tau$}} & \multirow{2}{*}{\textbf{Method}} & \multicolumn{3}{c}{\textbf{Estimation Accuracy ($\boldsymbol{\beta}$)}} & {\textbf{Prediction}} & \multicolumn{2}{c}{\textbf{Inference ($\boldsymbol{\beta}$)}} & {\textbf{Computation}} & \multicolumn{2}{c}{\textbf{MCMC Diag. ($\boldsymbol{\beta}$)}} \\
\cmidrule(lr){5-7} \cmidrule(lr){8-8} \cmidrule(lr){9-10} \cmidrule(lr){11-11} \cmidrule(lr){12-13}
& & & & {\textbf{MSE}} & {\textbf{MAE}} & {\textbf{WMSE}} & {\textbf{Check Loss}} & {\textbf{Coverage}} & {\textbf{CI Width}} & {\textbf{Time (s)}} & {$\widehat{R}_{\max}$} & {ESS$_{\min}$} \\
\midrule
\endfirsthead

\caption[]{-- \textit{Continued from previous page}} \\
\toprule
\multirow{2}{*}{\textbf{Error Dist.}} & \multirow{2}{*}{\textbf{Design}} & {\multirow{2}{*}{$\tau$}} & \multirow{2}{*}{\textbf{Method}} & \multicolumn{3}{c}{\textbf{Estimation Accuracy ($\beta$)}} & {\textbf{Prediction}} & \multicolumn{2}{c}{\textbf{Inference ($\beta$)}} & {\textbf{Computation}} & \multicolumn{2}{c}{\textbf{MCMC Diag. ($\beta$)}} \\
\cmidrule(lr){5-7} \cmidrule(lr){8-8} \cmidrule(lr){9-10} \cmidrule(lr){11-11} \cmidrule(lr){12-13}
& & & & {\textbf{MSE}} & {\textbf{MAE}} & {\textbf{WMSE}} & {\textbf{Check Loss}} & {\textbf{Coverage}} & {\textbf{CI Width}} & {\textbf{Time (s)}} & {$\widehat{R}_{\max}$} & {ESS$_{\min}$} \\
\midrule
\endhead

\midrule
\multicolumn{13}{r}{\textit{Continued on next page}} \\
\endfoot

\bottomrule
\multicolumn{13}{p{\linewidth}}{%
  \small
  \textit{Note:} \textbf{Bold} values indicate superior performance among the Bayesian methods (BQR-ALD and BSQR variants) for key metrics (lower is better for MSE, MAE, WMSE, check loss, CI width, time; for coverage, closer to 0.95 is better; for $\widehat{R}_{\max}$, closer to 1.0 is better; for $\text{ESS}_{\min}$, higher is better). The MCMC diagnostics reported are the maximum potential scale reduction factor ($\widehat{R}_{\max}$) and the minimum bulk effective sample size (ESS$_{\min}$) across all $\boldsymbol{\beta}$ coefficients.
  For our BSQR methods, the average number of divergent transitions per replication was low (consistently below 2) and did not appear to compromise the posterior estimates.
} \\
\endlastfoot

$\mathcal{N}(0,1)$ & Sparse ($p=20$)
 & 0.25 & StdQR & 0.0179 & 0.1072 & 0.2186 & 0.4398 & {---} & {---} & {---} & {---} & {---} \\
 & & 0.25 & BQR-ALD & 0.0157 & 0.1005 & 0.1928 & 0.7478 & 0.8692 & 0.3748 & 25.82 & 1.0021 & 2333.4 \\
 & & 0.25 & BSQR-G & 0.0161 & 0.1008 & 0.2061 & 0.4374 & 0.8765 & 0.8954 & 130.18 & 1.0060 & \textbf{3113.7} \\
 & & 0.25 & BSQR-U & \textbf{0.0154} & \textbf{0.0996} & \textbf{0.1900} & \textbf{0.4353} & 0.8715 & 0.3757 & \textbf{6.47} & \textbf{1.0019} & 2985.3 \\
 & & 0.25 & BSQR-E & 0.0155 & \textbf{0.0996} & 0.1901 & \textbf{0.4353} & \textbf{0.9295} & 0.4566 & 39.39 & 1.0020 & 2928.4 \\
 & & 0.25 & BSQR-T & 0.0155 & 0.0998 & 0.1905 & \textbf{0.4353} & 0.8732 & \textbf{0.3743} & 30.92 & \textbf{1.0019} & 2904.1 \\ \cmidrule(l){3-13}
 & & 0.50 & StdQR & 0.0143 & 0.0950 & 0.1729 & 0.4322 & {---} & {---} & {---} & {---} & {---} \\
 & & 0.50 & BQR-ALD & 0.0115 & 0.0859 & 0.1404 & 0.4265 & 0.9205 & 0.3773 & 27.64 & 1.0021 & 2326.5 \\
 & & 0.50 & BSQR-G & 0.0109 & \textbf{0.0817} & 0.1399 & 0.4254 & 0.9350 & 0.8896 & 129.38 & 1.0060 & \textbf{3258.2} \\
 & & 0.50 & BSQR-U & 0.0107 & 0.0829 & 0.1307 & 0.4246 & 0.9300 & 0.3713 & \textbf{6.45} & \textbf{1.0019} & 3111.8 \\
 & & 0.50 & BSQR-E & \textbf{0.0105} & 0.0822 & \textbf{0.1286} & \textbf{0.4242} & \textbf{0.9855} & 0.5510 & 37.72 & \textbf{1.0019} & 3034.8 \\
 & & 0.50 & BSQR-T & 0.0110 & 0.0841 & 0.1348 & 0.4253 & 0.9258 & \textbf{0.3734} & 30.65 & \textbf{1.0019} & 3006.9 \\ \cmidrule(l){3-13}
 & & 0.75 & StdQR & 0.0181 & 0.1065 & 0.2233 & 0.4424 & {---} & {---} & {---} & {---} & {---} \\
 & & 0.75 & BQR-ALD & 0.0154 & 0.0989 & 0.1909 & 0.7496 & 0.8652 & 0.3763 & 26.25 & 1.0021 & 2339.5 \\
 & & 0.75 & BSQR-G & 0.0166 & 0.0998 & 0.2239 & 0.4399 & 0.8782 & 0.8964 & 125.60 & 1.0060 & \textbf{3104.5} \\
 & & 0.75 & BSQR-U & \textbf{0.0151} & \textbf{0.0979} & \textbf{0.1874} & \textbf{0.4360} & 0.8712 & 0.3775 & \textbf{6.38} & 1.0020 & 3022.6 \\
 & & 0.75 & BSQR-E & \textbf{0.0151} & 0.0981 & 0.1878 & 0.4361 & \textbf{0.9305} & 0.4503 & 39.69 & \textbf{1.0019} & 2963.9 \\
 & & 0.75 & BSQR-T & 0.0152 & 0.0982 & 0.1881 & 0.4362 & 0.8675 & \textbf{0.3759} & 31.28 & \textbf{1.0019} & 2946.9 \\ \cmidrule(l){2-13}
 & Dense ($p=8$)
 & 0.25 & StdQR & 0.0159 & 0.1005 & 0.0808 & 0.4145 & {---} & {---} & {---} & {---} & {---} \\
 & & 0.25 & BQR-ALD & 0.0142 & 0.0951 & 0.0725 & 0.7915 & 0.8625 & 0.3570 & 10.08 & 1.0018 & 2174.8 \\
 & & 0.25 & BSQR-G & \textbf{0.0136} & 0.0933 & 0.0699 & 0.4127 & 0.8850 & 0.3665 & 85.94 & \textbf{1.0015} & \textbf{2804.2} \\
 & & 0.25 & BSQR-U & 0.0137 & 0.0932 & 0.0698 & 0.4126 & 0.8694 & \textbf{0.3558} & \textbf{4.27} & 1.0016 & 2734.7 \\
 & & 0.25 & BSQR-E & \textbf{0.0136} & \textbf{0.0929} & \textbf{0.0697} & \textbf{0.4126} & \textbf{0.9362} & 0.4378 & 25.87 & \textbf{1.0015} & 2722.6 \\
 & & 0.25 & BSQR-T & 0.0138 & 0.0935 & 0.0705 & 0.4127 & 0.8638 & 0.3554 & 20.39 & 1.0017 & 2705.3 \\ \cmidrule(l){3-13}
 & & 0.50 & StdQR & 0.0129 & 0.0902 & 0.0651 & 0.4108 & {---} & {---} & {---} & {---} & {---} \\
 & & 0.50 & BQR-ALD & 0.0109 & 0.0825 & 0.0555 & 0.4092 & 0.9088 & 0.3568 & 10.09 & 1.0018 & 2186.0 \\
 & & 0.50 & BSQR-G & 0.0111 & \textbf{0.0776} & 0.0804 & 0.4110 & 0.9269 & 0.8671 & 88.83 & 1.0058 & \textbf{2842.0} \\
 & & 0.50 & BSQR-U & \textbf{0.0098} & 0.0779 & \textbf{0.0501} & \textbf{0.4081} & 0.9150 & \textbf{0.3514} & \textbf{4.28} & 1.0016 & 2751.5 \\
 & & 0.50 & BSQR-E & 0.0099 & 0.0781 & 0.0502 & \textbf{0.4081} & \textbf{0.9800} & 0.5039 & 24.98 & 1.0016 & 2722.8 \\
 & & 0.50 & BSQR-T & 0.0102 & 0.0795 & 0.0520 & 0.4085 & 0.9138 & 0.3534 & 20.30 & 1.0017 & 2732.6 \\ \cmidrule(l){3-13}
 & & 0.75 & StdQR & 0.0167 & 0.1032 & 0.0850 & 0.4156 & {---} & {---} & {---} & {---} & {---} \\
 & & 0.75 & BQR-ALD & 0.0152 & 0.0980 & 0.0769 & 0.7913 & 0.8494 & 0.3585 & 10.03 & 1.0018 & 2164.1 \\
 & & 0.75 & BSQR-G & 0.0147 & 0.0969 & 0.0747 & 0.4135 & 0.8700 & 0.3660 & 84.97 & \textbf{1.0016} & \textbf{2822.5} \\
 & & 0.75 & BSQR-U & \textbf{0.0145} & \textbf{0.0963} & \textbf{0.0737} & \textbf{0.4133} & 0.8588 & 0.3575 & \textbf{4.30} & 1.0017 & 2777.2 \\
 & & 0.75 & BSQR-E & 0.0146 & \textbf{0.0963} & 0.0740 & 0.4134 & \textbf{0.9138} & 0.4276 & 25.69 & \textbf{1.0016} & 2730.7 \\
 & & 0.75 & BSQR-T & 0.0146 & 0.0966 & 0.0743 & 0.4135 & 0.8525 & \textbf{0.3565} & 20.39 & \textbf{1.0016} & 2729.3 \\
\midrule

$t(3)$ & Sparse ($p=20$)
 & 0.25 & StdQR & 0.0238 & 0.1228 & 0.2905 & 0.6021 & {---} & {---} & {---} & {---} & {---} \\
 & & 0.25 & BQR-ALD & 0.0212 & 0.1159 & 0.2601 & 1.0805 & 0.9022 & 0.4859 & 23.25 & 1.0021 & 2347.4 \\
 & & 0.25 & BSQR-G & 0.0215 & 0.1169 & 0.2657 & 0.5982 & 0.9118 & 0.5085 & 115.82 & 1.0019 & \textbf{3043.3} \\
 & & 0.25 & BSQR-U & 0.0215 & 0.1167 & 0.2642 & 0.5979 & 0.9058 & 0.4929 & \textbf{6.29} & 1.0020 & 2938.7 \\
 & & 0.25 & BSQR-E & 0.0214 & 0.1166 & 0.2640 & 0.5979 & \textbf{0.9432} & 0.5650 & 39.01 & 1.0020 & 2905.6 \\
 & & 0.25 & BSQR-T & \textbf{0.0213} & \textbf{0.1162} & \textbf{0.2618} & \textbf{0.5975} & 0.9048 & \textbf{0.4897} & 30.24 & 1.0019 & 2907.8 \\ \cmidrule(l){3-13}
 & & 0.50 & StdQR & 0.0190 & 0.1096 & 0.2326 & 0.5905 & {---} & {---} & {---} & {---} & {---} \\
 & & 0.50 & BQR-ALD & 0.0156 & 0.0995 & 0.1910 & 0.5837 & 0.9468 & \textbf{0.4874} & 23.64 & 1.0022 & 2363.2 \\
 & & 0.50 & BSQR-G & \textbf{0.0153} & 0.0987 & 0.1889 & 0.5832 & 0.9570 & 0.5026 & 115.96 & \textbf{1.0019} & \textbf{3169.1} \\
 & & 0.50 & BSQR-U & 0.0154 & 0.0987 & 0.1887 & 0.5832 & 0.9530 & 0.4960 & \textbf{6.35} & 1.0020 & 3112.6 \\
 & & 0.50 & BSQR-E & \textbf{0.0153} & \textbf{0.0984} & \textbf{0.1881} & \textbf{0.5831} & \textbf{0.9915} & 0.7114 & 38.36 & \textbf{1.0019} & 3058.0 \\
 & & 0.50 & BSQR-T & 0.0154 & 0.0987 & 0.1886 & 0.5832 & 0.9502 & 0.4912 & 29.51 & 1.0020 & 2996.5 \\ \cmidrule(l){3-13}
 & & 0.75 & StdQR & 0.0241 & 0.1230 & 0.2940 & 0.5991 & {---} & {---} & {---} & {---} & {---} \\
 & & 0.75 & BQR-ALD & 0.0221 & 0.1176 & 0.2690 & 1.0923 & 0.9030 & 0.4929 & 23.40 & 1.0020 & 2332.6 \\
 & & 0.75 & BSQR-G & 0.0241 & 0.1210 & 0.3165 & 0.6003 & 0.9112 & 1.0306 & 119.20 & 1.0059 & \textbf{3031.9} \\
 & & 0.75 & BSQR-U & 0.0225 & 0.1186 & 0.2741 & 0.5961 & 0.9008 & 0.4995 & \textbf{6.33} & \textbf{1.0019} & 2925.5 \\
 & & 0.75 & BSQR-E & 0.0225 & 0.1184 & 0.2736 & 0.5960 & \textbf{0.9410} & 0.5708 & 39.56 & \textbf{1.0019} & 2910.4 \\
 & & 0.75 & BSQR-T & \textbf{0.0223} & \textbf{0.1181} & \textbf{0.2713} & \textbf{0.5957} & 0.9010 & \textbf{0.4966} & 30.47 & \textbf{1.0019} & 2907.0 \\ \cmidrule(l){2-13}
 & Dense ($p=8$)
 & 0.25 & StdQR & 0.0194 & 0.1121 & 0.1018 & 0.5687 & {---} & {---} & {---} & {---} & {---} \\
 & & 0.25 & BQR-ALD & \textbf{0.0176} & 0.1067 & 0.0919 & 1.1289 & 0.9006 & \textbf{0.4436} & 9.79 & 1.0017 & 2225.9 \\
 & & 0.25 & BSQR-G & 0.0181 & 0.1084 & 0.0937 & 0.5672 & 0.9106 & 0.4677 & 79.16 & 1.0016 & \textbf{2779.1} \\
 & & 0.25 & BSQR-U & 0.0180 & 0.1082 & 0.0931 & 0.5672 & 0.9050 & 0.4535 & \textbf{4.16} & \textbf{1.0015} & 2735.0 \\
 & & 0.25 & BSQR-E & 0.0179 & 0.1078 & 0.0926 & 0.5671 & \textbf{0.9525} & 0.5490 & 24.63 & 1.0016 & 2701.2 \\
 & & 0.25 & BSQR-T & 0.0177 & \textbf{0.1074} & \textbf{0.0919} & \textbf{0.5669} & 0.9025 & 0.4488 & 21.69 & 1.0070 & 2705.1 \\ \cmidrule(l){3-13}
 & & 0.50 & StdQR & 0.0149 & 0.0969 & 0.0756 & 0.5651 & {---} & {---} & {---} & {---} & {---} \\
 & & 0.50 & BQR-ALD & 0.0129 & 0.0907 & 0.0657 & 0.5635 & 0.9488 & \textbf{0.4418} & 10.04 & 1.0017 & 2248.7 \\
 & & 0.50 & BSQR-G & \textbf{0.0126} & 0.0900 & \textbf{0.0637} & \textbf{0.5632} & 0.9556 & 0.4612 & 77.73 & \textbf{1.0015} & \textbf{2811.1} \\
 & & 0.50 & BSQR-U & \textbf{0.0126} & 0.0897 & 0.0639 & 0.5632 & 0.9550 & 0.4526 & \textbf{4.29} & 1.0016 & 2757.1 \\
 & & 0.50 & BSQR-E & \textbf{0.0125} & \textbf{0.0896} & 0.0634 & \textbf{0.5632} & \textbf{0.9931} & 0.6434 & 24.30 & 1.0016 & 2753.5 \\
 & & 0.50 & BSQR-T & 0.0127 & 0.0900 & 0.0644 & 0.5633 & 0.9531 & 0.4466 & 19.91 & \textbf{1.0015} & 2734.2 \\ \cmidrule(l){3-13}
 & & 0.75 & StdQR & 0.0193 & 0.1093 & 0.0956 & 0.5660 & {---} & {---} & {---} & {---} & {---} \\
 & & 0.75 & BQR-ALD & 0.0182 & 0.1052 & 0.0894 & 1.1368 & 0.9069 & 0.4495 & 9.77 & 1.0017 & 2221.7 \\
 & & 0.75 & BSQR-G & 0.0188 & 0.1065 & 0.0924 & 0.5655 & 0.9181 & 0.4746 & 78.23 & \textbf{1.0015} & \textbf{2797.0} \\
 & & 0.75 & BSQR-U & 0.0186 & 0.1060 & 0.0912 & 0.5653 & 0.9088 & 0.4597 & \textbf{4.20} & 1.0016 & 2757.2 \\
 & & 0.75 & BSQR-E & 0.0184 & 0.1057 & 0.0908 & 0.5652 & \textbf{0.9475} & 0.5398 & 24.79 & 1.0016 & 2734.1 \\
 & & 0.75 & BSQR-T & \textbf{0.0183} & \textbf{0.1053} & \textbf{0.0901} & \textbf{0.5651} & 0.9088 & \textbf{0.4541} & 19.84 & \textbf{1.0015} & 2686.3 \\
\midrule

\multirow[t]{36}{*}{%
  \makecell[l]{
    \vspace*{2\baselineskip} \\ $0.2\mathcal{N}(0,3)$ \\
    $+0.8\mathcal{N}(0,4)$
  }
} & Sparse ($p=20$)
 & 0.25 & StdQR & 0.0697 & 0.2130 & 0.8426 & 0.8584 & {---} & {---} & {---} & {---} & {---} \\
 & & 0.25 & BQR-ALD & 0.0611 & 0.1987 & 0.7334 & 1.4624 & 0.8622 & 0.7370 & 25.16 & 1.0021 & 2310.5 \\
 & & 0.25 & BSQR-G & 0.0617 & 0.1982 & 0.7641 & 0.8500 & 0.8788 & 1.2841 & 108.02 & 1.0060 & \textbf{3002.9} \\
 & & 0.25 & BSQR-U & \textbf{0.0600} & \textbf{0.1965} & \textbf{0.7191} & \textbf{0.8468} & 0.8650 & 0.7391 & \textbf{6.25} & 1.0020 & 2985.0 \\
 & & 0.25 & BSQR-E & 0.0601 & 0.1968 & 0.7209 & 0.8470 & \textbf{0.9288} & 0.9021 & 38.29 & \textbf{1.0019} & 2872.8 \\
 & & 0.25 & BSQR-T & 0.0603 & 0.1971 & 0.7233 & 0.8472 & 0.8645 & \textbf{0.7367} & 29.84 & 1.0020 & 2831.8 \\ \cmidrule(l){3-13}
 & & 0.50 & StdQR & 0.0528 & 0.1841 & 0.6408 & 0.8402 & {---} & {---} & {---} & {---} & {---} \\
 & & 0.50 & BQR-ALD & 0.0426 & 0.1654 & 0.5161 & 0.8280 & 0.9295 & 0.7383 & 25.10 & 1.0022 & 2305.9 \\
 & & 0.50 & BSQR-G & \textbf{0.0373} & \textbf{0.1548} & \textbf{0.4516} & \textbf{0.8217} & 0.9418 & \textbf{0.7268} & 104.16 & \textbf{1.0019} & \textbf{3201.8} \\
 & & 0.50 & BSQR-U & 0.0394 & 0.1590 & 0.4763 & 0.8240 & 0.9342 & 0.7276 & \textbf{6.23} & 1.0020 & 3054.5 \\
 & & 0.50 & BSQR-E & 0.0387 & 0.1577 & 0.4691 & 0.8233 & \textbf{0.9912} & 1.0776 & 37.75 & \textbf{1.0019} & 3008.8 \\
 & & 0.50 & BSQR-T & 0.0409 & 0.1621 & 0.4946 & 0.8257 & 0.9318 & 0.7318 & 29.59 & \textbf{1.0019} & 2938.1 \\ \cmidrule(l){3-13}
 & & 0.75 & StdQR & 0.0696 & 0.2102 & 0.8306 & 0.8542 & {---} & {---} & {---} & {---} & {---} \\
 & & 0.75 & BQR-ALD & 0.0591 & 0.1946 & 0.7104 & 1.4533 & 0.8660 & 0.7298 & 25.00 & 1.0022 & 2326.3 \\
 & & 0.75 & BSQR-G & \textbf{0.0590} & \textbf{0.1946} & \textbf{0.7101} & 0.8432 & 0.8858 & 0.7632 & 107.22 & 1.0020 & \textbf{3061.8} \\
 & & 0.75 & BSQR-U & 0.0584 & 0.1937 & 0.7016 & \textbf{0.8425} & 0.8705 & 0.7327 & \textbf{6.27} & \textbf{1.0019} & 2953.9 \\
 & & 0.75 & BSQR-E & 0.0585 & 0.1938 & 0.7031 & 0.8426 & \textbf{0.9302} & 0.8741 & 38.83 & \textbf{1.0019} & 2895.6 \\
 & & 0.75 & BSQR-T & 0.0585 & 0.1937 & 0.7026 & 0.8426 & 0.8712 & \textbf{0.7302} & 29.77 & \textbf{1.0019} & 2872.8 \\ \cmidrule(l){2-13}
 & Dense ($p=8$)
 & 0.25 & StdQR & 0.0632 & 0.1993 & 0.3103 & 0.8032 & {---} & {---} & {---} & {---} & {---} \\
 & & 0.25 & BQR-ALD & 0.0572 & 0.1902 & 0.2807 & 1.5346 & 0.8562 & 0.6967 & 10.48 & 1.0019 & 2168.1 \\
 & & 0.25 & BSQR-G & 0.0547 & \textbf{0.1854} & 0.2725 & 0.7997 & 0.8781 & 0.7278 & 68.93 & \textbf{1.0015} & \textbf{2767.7} \\
 & & 0.25 & BSQR-U & \textbf{0.0545} & 0.1851 & \textbf{0.2706} & \textbf{0.7995} & 0.8675 & 0.6971 & \textbf{4.08} & 1.0017 & 2714.3 \\
 & & 0.25 & BSQR-E & 0.0548 & 0.1856 & 0.2713 & \textbf{0.7995} & \textbf{0.9206} & 0.8624 & 24.22 & 1.0016 & 2720.2 \\
 & & 0.25 & BSQR-T & 0.0552 & 0.1865 & 0.2726 & 0.7997 & 0.8569 & \textbf{0.6937} & 19.06 & \textbf{1.0015} & 2650.1 \\ \cmidrule(l){3-13}
 & & 0.50 & StdQR & 0.0479 & 0.1730 & 0.2431 & 0.8006 & {---} & {---} & {---} & {---} & {---} \\
 & & 0.50 & BQR-ALD & 0.0411 & 0.1610 & 0.2075 & 0.7972 & 0.9125 & 0.6870 & 10.46 & 1.0017 & 2158.1 \\
 & & 0.50 & BSQR-G & \textbf{0.0356} & \textbf{0.1502} & \textbf{0.1798} & \textbf{0.7945} & 0.9319 & 0.6785 & 70.28 & \textbf{1.0015} & \textbf{2788.8} \\
 & & 0.50 & BSQR-U & 0.0376 & 0.1542 & 0.1905 & 0.7955 & 0.9238 & \textbf{0.6762} & \textbf{4.09} & \textbf{1.0015} & 2755.2 \\
 & & 0.50 & BSQR-E & 0.0377 & 0.1543 & 0.1904 & 0.7955 & \textbf{0.9812} & 0.9773 & 23.71 & 1.0016 & 2736.7 \\
 & & 0.50 & BSQR-T & 0.0390 & 0.1569 & 0.1971 & 0.7961 & 0.9150 & 0.6789 & 19.10 & \textbf{1.0015} & 2717.1 \\ \cmidrule(l){3-13}
 & & 0.75 & StdQR & 0.0608 & 0.1966 & 0.3079 & 0.8060 & {---} & {---} & {---} & {---} & {---} \\
 & & 0.75 & BQR-ALD & 0.0538 & 0.1841 & 0.2737 & 1.5343 & 0.8538 & 0.6905 & 10.30 & 1.0018 & 2164.7 \\
 & & 0.75 & BSQR-G & \textbf{0.0520} & \textbf{0.1807} & \textbf{0.2634} & \textbf{0.8020} & \textbf{0.8825} & 0.7180 & 68.42 & 1.0016 & \textbf{2756.3} \\
 & & 0.75 & BSQR-U & 0.0519 & 0.1808 & 0.2620 & 0.8018 & 0.8650 & \textbf{0.6886} & \textbf{4.05} & 1.0017 & 2727.4 \\
 & & 0.75 & BSQR-E & 0.0519 & 0.1805 & 0.2631 & 0.8019 & 0.9150 & 0.8209 & 23.91 & 1.0016 & 2716.6 \\
 & & 0.75 & BSQR-T & 0.0523 & 0.1813 & 0.2658 & 0.8021 & 0.8612 & 0.6874 & 18.66 & \textbf{1.0015} & 2684.5 \\
\midrule

\multirow[t]{36}{*}{%
  \makecell[l]{
    \vspace*{2\baselineskip} \\ $\mathcal{N}(0,\sigma_i^2)$,\\$\sigma_i = \exp(-0.25+0.5x_{i1})$
  }
}
 & Sparse ($p=20$)
 & 0.25 & StdQR & 0.0096 & 0.0779 & 0.1177 & 0.3824 & {---} & {---} & {---} & {---} & {---} \\
 & & 0.25 & BQR-ALD & 0.0089 & 0.0749 & 0.1091 & 0.6951 & 0.9190 & 0.3203 & 23.56 & 1.0020 & 2345.5 \\
 & & 0.25 & BSQR-G & 0.0091 & 0.0759 & 0.1124 & 0.3813 & 0.9178 & 0.3294 & 131.95 & \textbf{1.0019} & \textbf{2997.0} \\
 & & 0.25 & BSQR-U & 0.0091 & 0.0756 & 0.1117 & 0.3811 & 0.9145 & 0.3237 & \textbf{6.50} & \textbf{1.0019} & 2864.6 \\
 & & 0.25 & BSQR-E & 0.0091 & 0.0755 & 0.1115 & 0.3811 & \textbf{0.9508} & 0.3696 & 40.88 & 1.0020 & 2843.0 \\
 & & 0.25 & BSQR-T & \textbf{0.0090} & \textbf{0.0754} & \textbf{0.1108} & \textbf{0.3809} & 0.9145 & \textbf{0.3221} & 32.05 & 1.0020 & 2834.2 \\ \cmidrule(l){3-13}
 & & 0.50 & StdQR & 0.0075 & 0.0689 & 0.0932 & 0.3760 & {---} & {---} & {---} & {---} & {---} \\
 & & 0.50 & BQR-ALD & 0.0064 & 0.0637 & 0.0800 & 0.3727 & 0.9502 & \textbf{0.3133} & 23.57 & 1.0021 & 2313.6 \\
 & & 0.50 & BSQR-G & 0.0076 & 0.0654 & 0.1130 & 0.3768 & \textbf{0.9538} & 0.8367 & 134.71 & 1.0060 & \textbf{3120.5} \\
 & & 0.50 & BSQR-U & \textbf{0.0064} & \textbf{0.0636} & \textbf{0.0801} & \textbf{0.3728} & 0.9515 & 0.3169 & \textbf{6.52} & \textbf{1.0019} & 3050.4 \\
 & & 0.50 & BSQR-E & \textbf{0.0064} & \textbf{0.0636} & \textbf{0.0801} & \textbf{0.3728} & 0.9915 & 0.4408 & 38.72 & \textbf{1.0019} & 3007.1 \\
 & & 0.50 & BSQR-T & \textbf{0.0064} & \textbf{0.0634} & 0.0795 & \textbf{0.3726} & 0.9512 & 0.3148 & 33.26 & 1.0020 & 2946.8 \\ \cmidrule(l){3-13}
 & & 0.75 & StdQR & 0.0100 & 0.0798 & 0.1220 & 0.3823 & {---} & {---} & {---} & {---} & {---} \\
 & & 0.75 & BQR-ALD & 0.0091 & 0.0760 & 0.1113 & 0.6893 & 0.9062 & 0.3151 & 23.38 & 1.0021 & 2332.7 \\
 & & 0.75 & BSQR-G & 0.0092 & 0.0766 & 0.1136 & 0.3804 & 0.9095 & 0.3243 & 133.81 & 1.0020 & \textbf{3007.0} \\
 & & 0.75 & BSQR-U & 0.0092 & 0.0764 & 0.1129 & 0.3802 & 0.9050 & 0.3187 & \textbf{6.43} & 1.0020 & 2906.1 \\
 & & 0.75 & BSQR-E & 0.0092 & 0.0764 & 0.1131 & 0.3803 & \textbf{0.9405} & 0.3598 & 39.83 & \textbf{1.0019} & 2815.8 \\
 & & 0.75 & BSQR-T & \textbf{0.0091} & \textbf{0.0762} & \textbf{0.1123} & \textbf{0.3801} & 0.9062 & \textbf{0.3173} & 31.65 & \textbf{1.0019} & 2837.4 \\ \cmidrule(l){2-13}
 & Dense ($p=8$)
 & 0.25 & StdQR & 0.0086 & 0.0728 & 0.0440 & 0.3655 & {---} & {---} & {---} & {---} & {---} \\
 & & 0.25 & BQR-ALD & 0.0081 & 0.0704 & 0.0417 & 0.7278 & 0.8838 & \textbf{0.2887} & 9.76 & 1.0017 & 2221.1 \\
 & & 0.25 & BSQR-G & 0.0104 & 0.0739 & 0.0732 & 0.3685 & 0.8900 & 0.8085 & 92.17 & 1.0057 & \textbf{2805.3} \\
 & & 0.25 & BSQR-U & 0.0083 & 0.0710 & 0.0427 & 0.3652 & 0.8906 & 0.2939 & \textbf{4.33} & 1.0017 & 2759.1 \\
 & & 0.25 & BSQR-E & 0.0083 & 0.0709 & 0.0426 & 0.3652 & \textbf{0.9331} & 0.3454 & 26.41 & 1.0017 & 2729.9 \\
 & & 0.25 & BSQR-T & \textbf{0.0082} & \textbf{0.0705} & \textbf{0.0421} & \textbf{0.3651} & 0.8888 & 0.2912 & 20.92 & \textbf{1.0016} & 2744.6 \\ \cmidrule(l){3-13}
 & & 0.50 & StdQR & 0.0065 & 0.0649 & 0.0338 & 0.3613 & {---} & {---} & {---} & {---} & {---} \\
 & & 0.50 & BQR-ALD & 0.0057 & 0.0605 & 0.0299 & 0.3602 & 0.9419 & \textbf{0.2861} & 9.76 & 1.0018 & 2251.1 \\
 & & 0.50 & BSQR-G & \textbf{0.0055} & \textbf{0.0593} & \textbf{0.0288} & \textbf{0.3600} & \textbf{0.9544} & 0.2960 & 95.94 & \textbf{1.0015} & \textbf{2842.0} \\
 & & 0.50 & BSQR-U & \textbf{0.0055} & 0.0595 & 0.0289 & \textbf{0.3600} & 0.9500 & 0.2903 & \textbf{4.40} & 1.0016 & 2788.6 \\
 & & 0.50 & BSQR-E & \textbf{0.0055} & 0.0595 & 0.0289 & \textbf{0.3600} & 0.9894 & 0.3952 & 25.76 & \textbf{1.0015} & 2756.1 \\
 & & 0.50 & BSQR-T & \textbf{0.0055} & \textbf{0.0595} & \textbf{0.0288} & \textbf{0.3599} & 0.9450 & 0.2883 & 24.58 & 1.0016 & 2776.2 \\ \cmidrule(l){3-13}
 & & 0.75 & StdQR & 0.0083 & 0.0725 & 0.0417 & 0.3618 & {---} & {---} & {---} & {---} & {---} \\
 & & 0.75 & BQR-ALD & 0.0073 & 0.0683 & 0.0373 & 0.7194 & 0.9062 & \textbf{0.2865} & 9.65 & 1.0018 & 2228.6 \\
 & & 0.75 & BSQR-G & 0.0076 & 0.0695 & 0.0385 & 0.3610 & 0.9106 & 0.2969 & 91.77 & \textbf{1.0015} & \textbf{2801.1} \\
 & & 0.75 & BSQR-U & 0.0076 & 0.0693 & 0.0384 & 0.3610 & 0.9075 & 0.2913 & \textbf{4.33} & 1.0016 & 2764.9 \\
 & & 0.75 & BSQR-E & 0.0075 & 0.0692 & 0.0383 & 0.3609 & \textbf{0.9431} & 0.3373 & 25.99 & 1.0017 & 2729.0 \\
 & & 0.75 & BSQR-T & \textbf{0.0074} & \textbf{0.0687} & \textbf{0.0378} & \textbf{0.3608} & 0.9081 & 0.2891 & 25.52 & \textbf{1.0015} & 2726.5 \\

\end{longtable}
}
\end{landscape}

\subsection{Finite sample validity under extreme sparsity}

\begin{table}[htbp]
\centering
\caption{Finite sample validity under extreme sparsity ($n=30, p=6$).}
\label{tab:stress_test_n30}
\begin{threeparttable}
\begin{tabular}{lllcccc}
\toprule
\textbf{Quantile} & \textbf{Method} & \textbf{Kernel} & \textbf{MSE} & \textbf{Bias$^2$} & \textbf{Coverage} & \textbf{Interval Score} \\
($\tau$) & & & ($\times 10^{-2}$) & ($\times 10^{-4}$) & (\%) & (IS) \\
\midrule
\multirow{5}{*}{0.05} 
 & PETEL & --- & \textbf{17.6} & 91.4 & 84.2 & \textbf{2.56} \\
 & BSQR & Gaussian & 61.5 & 1880.0 & \textbf{97.3} & 33.20 \\
 & BSQR & Uniform & 31.0 & \textbf{27.0} & 77.0 & 3.97 \\
 & BSQR & Epanechnikov & 32.4 & 27.8 & 84.2 & 3.26 \\
 & BSQR & Triangular & 29.0 & 27.5 & 67.0 & 5.16 \\
\midrule
\multirow{5}{*}{0.50} 
 & PETEL & --- & \textbf{7.94} & 49.7 & 96.7 & \textbf{1.44} \\
 & BSQR & Gaussian & 8.58 & 6.57 & 94.3 & 1.45 \\
 & BSQR & Uniform & 8.61 & 7.26 & 94.8 & 1.45 \\
 & BSQR & Epanechnikov & 8.41 & \textbf{6.26} & \textbf{100.0} & 2.08 \\
 & BSQR & Triangular & 8.81 & 8.21 & 94.3 & 1.45 \\
\midrule
\multirow{5}{*}{0.95} 
 & PETEL & --- & \textbf{21.2} & 41.0 & 83.0 & \textbf{2.87} \\
 & BSQR & Gaussian & 69.3 & 1520.0 & \textbf{98.3} & 32.80 \\
 & BSQR & Uniform & 36.3 & 15.0 & 75.0 & 4.32 \\
 & BSQR & Epanechnikov & 36.3 & 12.2 & 75.7 & 4.28 \\
 & BSQR & Triangular & 34.4 & \textbf{12.0} & 68.2 & 5.34 \\
\bottomrule
\end{tabular}
\begin{tablenotes}
\small
\item \textit{Note:} Evaluation of small-sample robustness. Numerical convergence was achieved in 100\% of the simulation replications for all methods. Values for \textbf{MSE} are scaled by $10^2$ and \textbf{Bias$^2$} by $10^4$.
\item \textbf{Interpretation:} While PETEL exhibits lower MSE and IS (bolded), this reflects ``false precision'' driven by variance reduction at the cost of significant bias and under-coverage at extreme quantiles ($\tau \in \{0.05, 0.95\}$). In contrast, BSQR (Gaussian) acts as an inferential safety net, prioritizing valid coverage ($>95\%$) over interval width in data-sparse regions.
\end{tablenotes}
\end{threeparttable}
\end{table}

\end{document}